\documentclass[11pt,english,table]{article}
\usepackage{lmodern}

\usepackage[T1]{fontenc}
\usepackage[latin9]{inputenc}
\usepackage{color}
\usepackage{bm}
\usepackage{amsmath}
\usepackage{amsthm}
\usepackage{amssymb}
\usepackage{graphicx}
\usepackage{geometry}
\usepackage{rotfloat}
\usepackage{setspace}
\usepackage[authoryear]{natbib}
\makeatletter
\theoremstyle{plain}
\newtheorem{assumption}{\protect\assumptionname}
\theoremstyle{plain}
\newtheorem{thm}{\protect\theoremname}
\theoremstyle{plain}
\newtheorem{lem}{\protect\lemmaname}

\setcitestyle{round}
\usepackage{lscape}
\usepackage{longtable}
\usepackage{xcolor}
\usepackage{rotating}

\usepackage{array}
\usepackage{tabularx}\usepackage{multirow}\usepackage{booktabs}
\usepackage{ragged2e}\newcolumntype{C}[1]{>{\centering\arraybackslash}p{#1}}
\usepackage{rotfloat}
\newcolumntype{J}[1]{>{\justify\arraybackslash}p{#1}}
\newcolumntype{R}[1]{>{\RaggedLeft\arraybackslash}p{#1}}
\newcolumntype{Q}[1]{>{\columncolor{Gray}\RaggedLeft\arraybackslash}p{#1}}
\newcolumntype{L}[1]{>{\RaggedRight\arraybackslash}p{#1}}
\newcolumntype{G}{@{\extracolsep{0.5cm}}l@{\extracolsep{0pt}}}%
\newcolumntype{P}[1]{>{\centering\arraybackslash}p{#1}}
\newcolumntype{Y}{>{\centering\arraybackslash}X}
\newcommand{\nhphantom}[1]{\sbox0{#1}\hspace{-\the\wd0}} 

\AtBeginDocument{

}

\usepackage[colorlinks,
            linkcolor=blue,
            anchorcolor=blue,
            citecolor=blue]{hyperref}

\usepackage{babel}

\ifdefined\showcaptionsetup
 \PassOptionsToPackage{caption=false}{subfig}
\fi
\usepackage{subfig}
\makeatother

\usepackage{babel}
\providecommand{\assumptionname}{Assumption}
\providecommand{\lemmaname}{Lemma}
\providecommand{\theoremname}{Theorem}

\begin{document}
\title{Option Pricing with Time-Varying Volatility Risk Aversion\thanks{We are grateful to Szabolcs Blazsek, Stefano Giglio, Jens Jackwerth,
Tobias Sichert, two anonymous reviewers, and conference participants
at the SoFiE 2022 conference for valuable comments.}{\normalsize\emph{\medskip{}
}}}
\author{\textbf{Peter Reinhard Hansen}$^{a}$\textbf{ }and\textbf{ Chen Tong}$^{b,c}$\thanks{Chen Tong acknowledges financial support from the National Natural
Science Foundation of China (72301227) and the Ministry of Education
of China, Humanities and Social Sciences Youth Fund (22YJC790117).}\bigskip{}
 \\
 {\normalsize$^{a}$}{\normalsize\emph{University of North Carolina
\& Copenhagen Business School}}{\normalsize{} }\\
 {\normalsize$^{b}$}{\normalsize\emph{Department of Finance, School
of Economics, Xiamen University}}{\normalsize{} }\\
 {\normalsize$^{c}$}{\normalsize\emph{Wang Yanan Institute for Studies
in Economics, Xiamen University\medskip{}
 }}}
\date{{\normalsize\emph{\today}}}
\maketitle
\begin{abstract}
We introduce a pricing kernel with time-varying volatility risk aversion
to explain observed time variations in the shape of the pricing kernel.
When combined with the Heston-Nandi GARCH model, this framework yields
a tractable option pricing model in which the variance risk ratio
(VRR) emerges as a key variable. We show that the VRR is closely linked
to economic fundamentals, as well as sentiment and uncertainty measures.
A novel approximation method provides analytical option pricing formulas,
and we demonstrate substantial reductions in pricing errors through
an empirical application to the S\&P 500 index, the CBOE VIX, and
option prices.
\end{abstract}
\clearpage{}

\section*{Introduction}

The pricing kernel, which represents the ratio of risk-neutral to
physical state probabilities, is a fundamental concept in economics
with roots in Arrow-Debreu securities, see \citet{ArrowDebreu1954}.
A cornerstone of asset pricing theory, the \citet{Lucas1978} tree
model, implies a monotonically decreasing pricing kernel with respect
to aggregate wealth, see \citet{HansenRenault:2010}. However, empirical
studies consistently find the pricing kernel to be non-monotonic and
time-varying, giving rise to the so-called pricing kernel puzzles,
see e.g. \citet{AitSahaliaLo2000}, \citet{Jackwerth2000}, and \citet{RosenbergEngle2002}.

In this paper, we propose a pricing kernel with time-varying volatility
risk aversion to resolve these pricing kernel puzzles. Building on
\citet{HestonNandi2000} and \citet{ChristoffersenHestonJacobs2013},
we develop a novel framework for derivatives pricing, in which the
\textit{variance risk ratio} (VRR) emerges as the key quantity. The
VRR, defined as the ratio of risk-neutral to physical conditional
variance, characterizes both the shape and time variation of the pricing
kernel and it is closely related to the variance risk premium by \citet{CarrWu2008}.
The dynamic pricing kernel introduces challenges for closed-form option
pricing formulas, which we address through a novel analytical approximation
method applicable to non-affine models. Our empirical implementation
of the dynamic structure for VRR is based on the score-driven framework
of \citet{CrealKoopmanLucas:2013}, which intuitively updates time-varying
quantities in response to first-order conditions. In an empirical
application to the S\&P 500 index, the CBOE VIX, and option prices,
we show that the new model substantially reduces pricing errors, both
in-sample and out-of-sample.

The discrete-time Heston-Nandi GARCH (HNG) model, developed by \citet{HestonNandi2000},
conveniently yields closed-form expressions for option prices while
allowing for time-varying volatility. However, it implies a monotonic
pricing kernel. \citet{ChristoffersenHestonJacobs2013} extended the
HNG model by introducing a variance-dependent pricing kernel, resulting
in a non-monotonic shape. Their model remains tractable, retains closed-form
option pricing expressions, and has become a benchmark in the option
pricing literature. The shape of the pricing kernel in their framework
is governed by the variance risk premium (VRP), where a negative VRP
leads to a U-shaped pricing kernel, a pattern confirmed in their empirical
analysis, see also \citet{CuesdeanuJackwerth2018}. 

While \citet{ChristoffersenHestonJacobs2013} explain non-monotonicity,
their framework cannot account for the time variation observed in
the pricing kernel. When estimated over long sample periods, the pricing
kernel is typically U-shaped, yet the shape differs distinctly in
certain periods, such as the years leading up to the global financial
crisis. The shape of the pricing kernel can be expressed in terms
of the probability weighting function. If individuals overweight low-probability
events and underweight high-probability events, the probability weighting
function will have an inverse S-shape, which corresponds to a U-shaped
pricing kernel. \citet{PolkovnichenkoZhao2013} estimated the probability
weighting function non-parametrically and found it to be time-varying.
Their estimate has an inverse S-shape during most of their sample
periods but a regular S-shape during the years 2004--2006. Similar
results were obtained by \citet{ChabiSong2013}, and an inverted U-shape
was also documented with DAX 30 options during the same years, 2004--2006,
as reported by \citet{GrithHardleKratschmer2017}. Further evidence
comes from \citet{KieselRahe2017} and \citet{BeareSchmidt2016}.
Interestingly, \citet[figure 3]{ChristoffersenHestonJacobs2013} also
contains evidence of time variations in empirical pricing kernels.
For most calendar years, their estimated empirical pricing kernel
is U-shaped; however, in some calendar years, it takes an inverted
U-shape.\footnote{The authors do not comment on this observation, but they estimate
a more flexible structure, which is not based on a pricing kernel
with a fixed shape, see \citet[figure 6]{ChristoffersenHestonJacobs2013}.}

The derivative pricing model proposed in this paper combines a pricing
kernel with time-varying volatility risk aversion and the GARCH model
of \citet{HestonNandi2000}. The VRR, defined as the ratio of risk-neutral
to physical conditional variance, $\eta_{t}=h_{t+1}^{\ast}/h_{t+1}$,
governs the time variation in the pricing kernel and is functionally
linked to the volatility risk aversion parameter, which captures the
curvature of the pricing kernel. We demonstrate that this structure
can generate the observed time variation in the empirical pricing
kernel, including its different shapes. Our new model nests \citet{ChristoffersenHestonJacobs2013}
as the special case where $\eta_{t}$ is constant, which we denote
as CHNG. The HNG option pricing model of \citet{HestonNandi2000}
is also nested by imposing $\eta_{t}=1$. We refer to our new model
as DHNG model, where ``D'' stands for dynamic, and highlight its
key properties below.

We derive a closed-form CBOE VIX pricing formula for the DHNG model
and develop an analytical option pricing formula, which constitutes
a separate and important methodological contribution. When volatility
risk aversion follows a stochastic process, the conditional moment-generating
function (MGF) of future cumulative returns does not have an affine
form, making it very challenging to derive option prices. However,
we introduce a novel approximation method that is applicable to non-affine
models. This approach constructs an auxiliary MGF for a simplified,
related problem and then applies a Taylor expansion to obtain an analytical
option pricing formula. The approximation method is shown to be highly
accurate in an empirically relevant simulation design.

To implement the theoretical framework, we connect the model for $\eta_{t}$
with observed data using an observation-driven approach, which is
a natural choice since the GARCH model under the physical measure
also follows an observation-driven formulation. Specifically, we propose
a score-driven model following \citet{CrealKoopmanLucas:2013}, where
the first-order conditions of the log-likelihood function define the
innovations in the dynamic model for $\eta_{t}$. This design is intuitive
because $\eta_{t}$ is updated to minimize pricing errors. Conveniently,
this structure also enhances robustness to model misspecification.
Another advantage of the observation-driven model is its simplicity
in estimation, which is particularly useful in applications involving
long sample periods and large panels of option prices.

An empirical application using 32 years of daily S\&P 500 returns,
the CBOE VIX, and a large panel of option prices provides strong support
for the model. The estimation incorporates information from both the
physical and risk-neutral measures, as advocated by \citet{ChernovGhysels2000}.
The results show that the pricing kernel with time-varying volatility
risk aversion typically reduces pricing errors by 50\% or more, as
measured by the root mean square error. This reduction is achieved
for both VIX and option prices, in both in-sample and out-of-sample
comparisons. The estimated time variation in the shape of the pricing
kernel aligns with previous studies, generally exhibiting a U-shape,
but taking on an inverted U-shape during certain periods, including
2004--2007 and around 1993 and 2017.

Finally, having established the VRR as a fundamental quantity, it
is natural to investigate its connection to key asset pricing variables,
including those linked to pricing kernel puzzles. Theoretical and
empirical studies suggest that heterogeneous beliefs and disagreements
about the physical distribution can contribute to a U-shaped pricing
kernel, see e.g. \citet[2008]{Shefrin2001},\nocite{Shefrin2008}
\citet{BakshiMadan2008}, \citet{BakshiMadanPanayotov2010}. Similar
arguments have been made about sentiment and uncertainty among investors,
see e.g. \citet{Han2008}, \citet{PolkovnichenkoZhao2013}, \citet{BakerBloomDavis2016},
and \citet{BollerslevLiXue2018}. \citet{BaliZhou2016} argue that
market uncertainty can be approximated by the variance risk premium,
see \citet{CarrWu2008}. The VRR is obviously related to the variance
risk premium, which has been shown to predict stock returns, see \citet{BollerslevTauchenZhou2009}.
A straightforward explanation is that preferences are state-dependent,
with market uncertainty being a possible state variable, see e.g.
\citet{GrithHardleKratschmer2017}. We connect our results to these
studies by showing that the monthly average VRR is closely related
to commonly used measures of sentiment, uncertainty, and disagreement,
as well as other key economic indicators.

Our paper is related to and builds on a large body of literature,
including studies that have sought to address pricing kernel puzzles
by augmenting existing models with additional state variables, see
e.g. \citet{Chabi-YoGarciaRenault2008}, \citet{chabi2012}, \citet{BrownJackwerth2012}
and \citet{SongXiu2016}. In our framework, the quantity VRR can be
viewed as an additional state variable that emerges naturally from
the model structure. 

Our work is also related to \citet{Barone-AdesiEngleMancini2008},
who estimated a GJR-GARCH model for returns under the physical measure.
They did not specify a pricing kernel but simply assumed that the
model under the risk-neutral measure is also a GJR-GARCH model, with
time-varying parameters calibrated to minimize option pricing errors.
While we also seek to minimize pricing errors, our approach is fundamentally
different. Our starting point is a pricing kernel and a model for
the physical probability measure, $\mathbb{P}$, and the two imply
the model under the risk-neutral probability measure, $\mathbb{Q}$.
An empirical comparison in \citet{ChristoffersenHestonJacobs2013}
explores the same idea as in \citet{Barone-AdesiEngleMancini2008},
as they estimate separate HNG models under both $\mathbb{P}$ and
$\mathbb{Q}$ measures. They referred to this structure as an \emph{ad-hoc
model}, because it is not derived from their pricing kernel. However,
we show that the pricing kernel with a time-varying VRR is incoherent
with HNG models (with constant parameters) under both $\mathbb{P}$
and $\mathbb{Q}$ measures. 

Additionally, our paper contributes to the literature on time-varying
risk aversion, which relates to the local shape of the pricing kernel.
A seminal paper by \citet{CampbellCochrane1999} addresses this concept,
while \citet{Li2007} demonstrated how time-varying risk aversion
links to a time-varying risk premium. Moreover, \citet{GonzalezNaveRubio2018}
identified risk aversion dynamics as a key determinant of stock market
betas. More recently, \citet{BekertEngstromXu2020} introduced a no-arbitrage
asset pricing model with time-varying risk aversion for pricing equities
and corporate bonds, where risk aversion is driven by uncertainty
shocks.

The rest of the paper is organized as follows. We present the theoretical
model in Section \ref{sec:Model} and derive closed-form pricing formula
for the VIX and an option pricing formula, based on the new approximation
method in Section \ref{sec:Derivative-Prices}. Section \ref{sec:ScoreDrivenModel}
introduces a dynamic model for the variance risk ratio using a score-driven
approach. An empirical application with 32 years of S\&P 500 returns,
VIX, and option prices is presented in Section \ref{sec:Empirical-Analysis}.
We show that the variance risk ratio relates to well-known measures
of sentiment, disagreement, and uncertainty in Section \ref{sec:Eta-and-EconomicFundamentals}.
A summary is presented in Section \ref{sec:Summary}. All proofs are
provided in the Appendix \ref{sec:Appendix-of-Proofs} and supplementary
empirical results can be found in the Online Appendix.

\everymath{\setlength{\abovedisplayskip}{11pt} \setlength{\belowdisplayskip}{11pt}} 
\everydisplay{\setlength{\abovedisplayskip}{11pt} \setlength{\belowdisplayskip}{11pt}} 

\section{The Model\label{sec:Model}}

This section introduces the new model with dynamic variance risk aversion
(DHNG), building on the models by \citet{HestonNandi2000} (HNG) and
\citet{ChristoffersenHestonJacobs2013} (CHNG).

The observed variables include the daily returns, $R_{t}\equiv\log\left(S_{t}/S_{t-1}\right)$,
where $S_{t}$ is the underlying asset price, and a vector of derivative
prices, $X_{t}$. We will work with two filtrations, $\mathcal{F}_{t}=\ensuremath{\sigma}(\{R_{j},X_{j}\},j\leq t)$
and $\mathcal{G}_{t}=\ensuremath{\sigma}(\{R_{j},X_{j-1}\},j\leq t)$,
where the latter arises naturally from a factorization of the joint
likelihood function.\textcolor{red}{{} }Clearly, $\mathcal{F}_{t-1}\subset\mathcal{G}_{t}\subset\mathcal{F}_{t}$
for all $t$. Throughout this paper, we will use the notation $\mathbb{E}_{t}^{\mathbb{P}}\left(\cdot\right)\equiv\mathbb{E}^{\mathbb{P}}\left(\cdot|\mathcal{G}_{t}\right)$
and $\mathbb{E}_{t}^{\mathbb{Q}}\left(\cdot\right)\equiv\mathbb{E}_{t}^{\mathbb{Q}}\left(\cdot|\mathcal{G}_{t}\right)$
to denote the conditional expectations with respect to $\mathcal{G}_{t}$
under $\mathbb{P}$ and $\mathbb{Q}$ measures, respectively.

We model returns in physical measure $\mathbb{P}$ using the classical
Heston-Nandi GARCH model (HNG), which has a convenient structure for
derivatives pricing. This model is given by
\begin{eqnarray}
R_{t+1} & = & r+(\lambda-\tfrac{1}{2})h_{t+1}+\sqrt{h_{t+1}}z_{t+1},\label{eq:HNGreturn}\\
h_{t+1} & = & \omega+\beta h_{t}+\alpha(z_{t}-\gamma\sqrt{h_{t}})^{2},\label{eq:HNGgarch}
\end{eqnarray}
where the return shock, $z_{t+1}$, is independent and identically
distributed (iid) with a standard normal distribution, $N(0,1)$;
$r$ is the risk-free rate; $\lambda$ is the equity risk premium
because the expected return is given by $\mathbb{E}^{\mathbb{P}}\left(\exp\left(R_{t+1}\right)|\mathcal{F}_{t}\right)=\exp\left(r+\lambda h_{t+1}\right)$;
and $h_{t+1}=\mathrm{var}^{\mathbb{P}}(R_{t+1}|\mathcal{F}_{t})$
is the daily conditional variance, see \citet{HestonNandi2000}. The
HNG model captures time variation in the conditional variance, as
is the case for ARCH and GARCH models, see \citet{Engle:1982} and
\citet{bollerslev:86}. However, the HNG model also allows for dependence
between returns and volatility (leverage). It is straightforward to
verify that $\mathrm{cov}^{\mathbb{P}}(R_{t},h_{t+1}|\mathcal{F}_{t-1})=-2\gamma\alpha h_{t}$,
which shows that the magnitude of the leverage effect is defined by
$\gamma$. A leverage effect is required to generate the empirically
important volatility smirk in option prices. The dynamic structure
of the HNG model is carefully crafted to yield a closed-form solution
for option valuation, and \citet{HestonNandi2000} showed that the
continuous limit of the HNG model (as the time interval between observations
shrinks to zero) yields a variance process, $h_{t}$, that converges
weakly to a continuous-time square-root variance process, see \citet{Feller1951},
\citet{CoxIngersollRoss85}, and \citet{Heston1993}.

Option pricing with GARCH models can be achieved with a simple risk-neutralization
by \citet{Duan1995}, known as the \emph{locally risk-neutral valuation
relationship} (LRNVR). This approach is equivalent to imposing a pricing
kernel with a single risk premium for equity risk, see \citet{HuangWangHansen2017}.
However, the LRNVR-based pricing kernel is inadequate for explaining
many discrepancies between $\mathbb{P}$ and $\mathbb{Q}$, including
the variance risk premium, see \citet{HaoZhang2013} and \citet{ChristoffersenHestonJacobs2013}.
We will review their model next and then proceed to introduce the
new model with time-varying volatility risk aversion.

\subsection{Pricing Kernel with Constant Parameters (CHNG)}

The model in \citet{ChristoffersenHestonJacobs2013} (CHNG) generalizes
the HNG model by \citet{HestonNandi2000}, to have a variance-dependent
pricing kernel, which is given by,
\[
\frac{M_{t+1}}{M_{t}}=\left(\frac{S_{t+1}}{S_{t}}\right)^{\phi}\exp\left[\delta+\pi h_{t+1}+\xi(h_{t+2}-h_{t+1})\right],
\]
and \citet{CorsiFusariVecchia2013} showed that it can be conveniently
expressed as
\[
M_{t+1,t}\equiv\frac{M_{t+1}}{\mathbb{E}^{\mathbb{P}}[M_{t+1}|\mathcal{F}_{t}^{R}]}=\frac{\exp\left(\phi R_{t+1}+\xi h_{t+2}\right)}{\mathbb{E}_{}^{\mathbb{P}}[\exp(\phi R_{t+1}+\xi h_{t+2})|\mathcal{F}_{t}^{R}]},
\]
where $\mathcal{F}_{t}^{R}=\sigma\left(\{R_{j}\},j\leq t\right)$
is the natural filtration for returns only. The pricing kernel above
depends on both equity risk and variance risk, where the latter is
typically characterized by the parameter $\xi$. The equity risk is
governed by $\phi$, as well as $\xi$, because $h_{t+2}$ depends
on return $R_{t+1}$. This pricing kernel implies the following risk-neutral
dynamics:
\begin{eqnarray*}
R_{t+1} & = & r-\tfrac{1}{2}h_{t+1}^{*}+\sqrt{h_{t+1}^{*}}z_{t+1}^{*}\\
h_{t+1}^{*} & = & \omega^{*}+\beta^{*}h_{t}^{*}+\alpha^{*}(z_{t}^{*}-\gamma^{*}\sqrt{h_{t}^{*}})^{2},
\end{eqnarray*}
with $z_{t+1}^{*}|\mathcal{F}_{t}^{R}\overset{\mathbb{Q}}{\sim}iid\ N(0,1)$
and the following relations between $\mathbb{P}$ and $\mathbb{Q}$
parameters:
\[
h_{t}^{*}=h_{t}\eta,\quad\omega^{*}=\omega\eta,\quad\beta^{\ast}=\beta,\quad\alpha^{*}=\alpha\eta^{2},\quad\gamma^{\ast}=\tfrac{1}{\eta}(\gamma+\lambda-\tfrac{1}{2})+\tfrac{1}{2},\quad\eta\equiv(1-2\alpha\xi)^{-1}.
\]

The logarithm of the pricing kernel is a quadratic function of the
daily market return,
\begin{equation}
\log\frac{M_{t}}{M_{t-1}}=-\lambda(R_{t}-r)+\xi\alpha\frac{(R_{t}-r)^{2}}{h_{t}}+\kappa_{0}+\kappa_{1}h_{t},\label{eq:QuadraticInReturns}
\end{equation}
where $\kappa_{0}=-\frac{1}{2}\log\eta$ and $\kappa_{1}=\frac{1}{2}(\lambda-\tfrac{1}{2})^{2}-\tfrac{1}{8}\eta$.\footnote{The expressions for $\kappa_{0}$ and $\kappa_{1}$ in \citet{ChristoffersenHestonJacobs2013}
are $\kappa_{0}=\delta+\xi\omega+\phi r$ and $\kappa_{1}=\pi+\xi(\beta-1+\alpha(\lambda-\tfrac{1}{2}+\gamma)^{2})$.
It follows from our results in Lemma \ref{lem:ShapeDynamicKernel}
that they are equivalent to those presented here.} Expression (\ref{eq:QuadraticInReturns}) shows that $\lambda$ and
$\xi$ define the linear and quadratic terms, respectively, with positive
$\xi$ leading to a U-shaped and negative $\xi$ to an inverted U-shaped
pricing kernel. Empirically, this relationship tends to be generally
U-shaped, but it is unstable over time as seen in \citet[figure 3]{ChristoffersenHestonJacobs2013}.
Figure \ref{fig:CHJ-3-and-5} includes parts of figures 3 and 5 in
\citet{ChristoffersenHestonJacobs2013}; the four left panels are
based on model-free market prices and the corresponding four right
panels are based on model prices using their estimated CHNG model.
\begin{figure}
\centering{}\subfloat[Log ratios inferred from market prices]{\centering{}\includegraphics[width=0.5\textwidth]{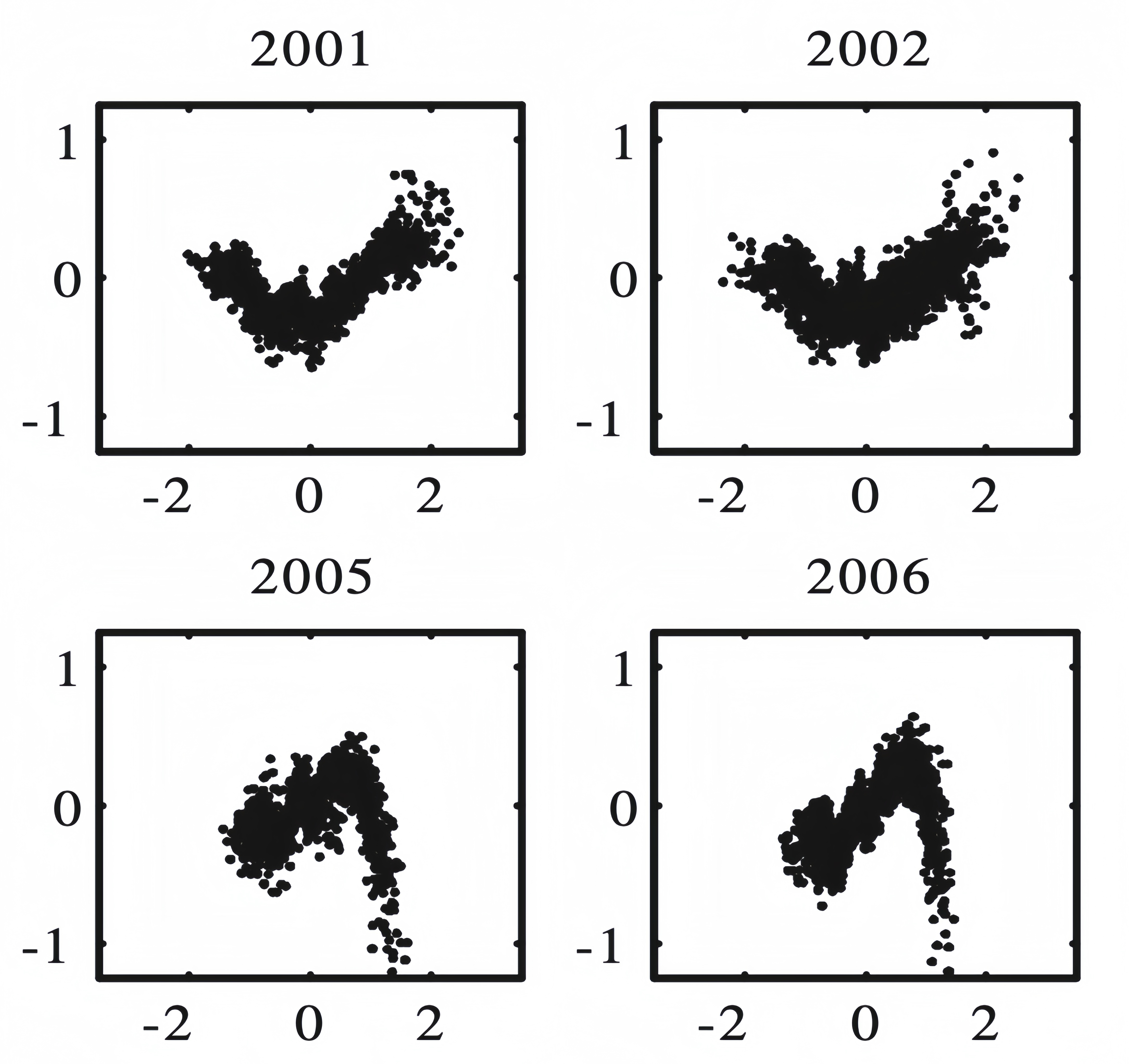}}\subfloat[CHNG model-based log ratios]{\centering{}\includegraphics[width=0.5\textwidth]{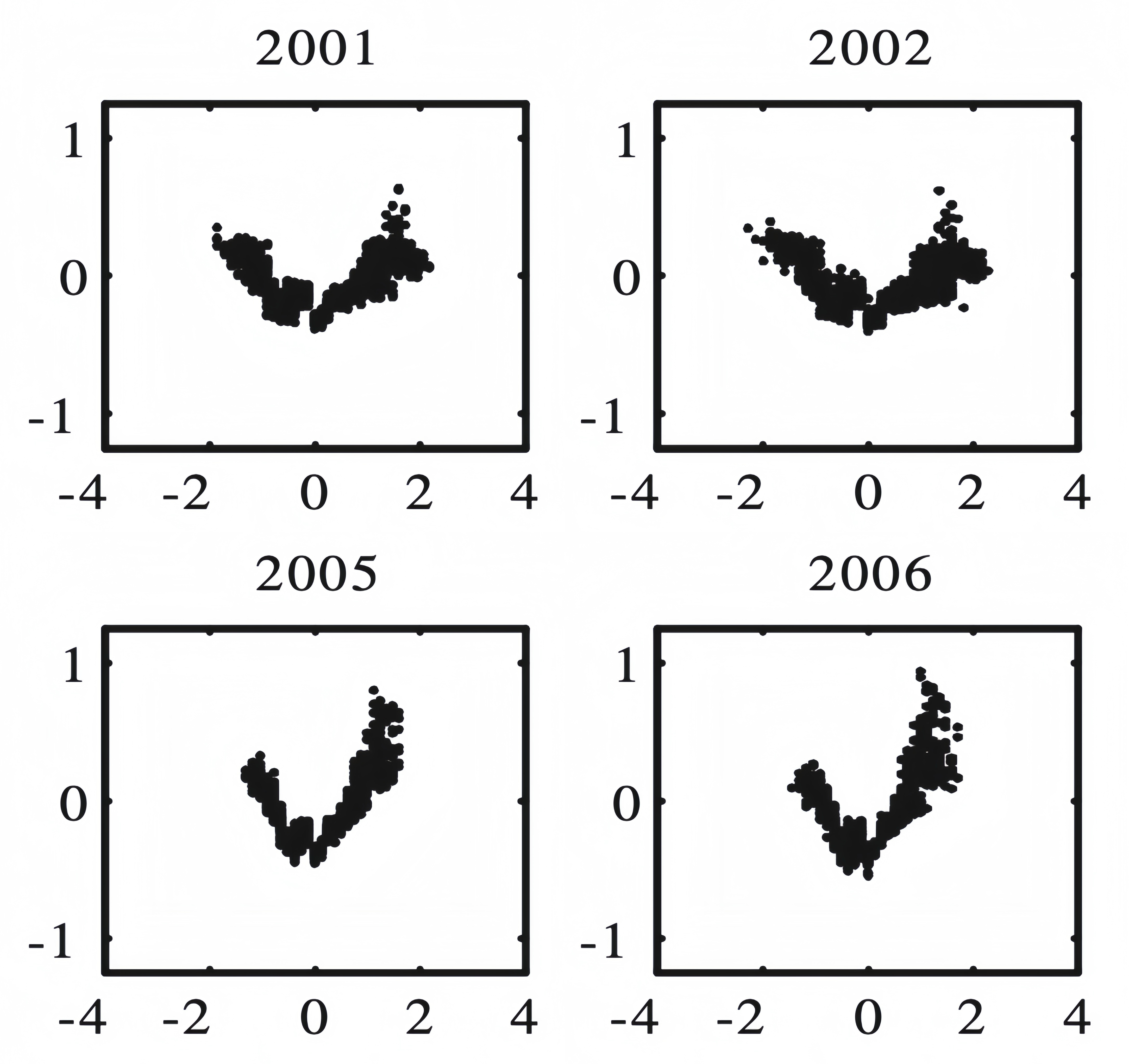}}\caption{{\small Log ratios of risk-neutral one-month densities and physical
one-month histogram against log returns in monthly standard deviations.
Panels (a) and (b) are subsets of figures 3 and 5, respectively, in
\citet{ChristoffersenHestonJacobs2013}. The former is based on empirical
market prices and the latter has the corresponding CHNG model-based
log ratios.\label{fig:CHJ-3-and-5}}\protect \\
{\small}\protect \\
{\small Alt text: The shape of the pricing kernel based on market prices
and the CHNG model for four calendar years. The figure plots log ratios
of risk-neutral to physical one-month densities against log returns.
The shapes for 2001 and 2002 are similar, while those for 2005 and
2006 are distinctively different. }}
\end{figure}
 There are periods where the relationship has an inverted $U$-shape,
such as the years 2004--2007 (of which 2005 and 2006 are shown in
Figure \ref{fig:CHJ-3-and-5}), which is incompatible with a static
quadratic coefficient $\xi$. 

These observations, and the fact that $\xi$ defines the shape of
the relationship between $\log M_{t}/M_{t-1}$ and returns, motivate
us to develop a new model where $\xi$ can be time-varying. This parameter,
$\xi$, is directly tied to volatility risk aversion, as it is the
coefficient to the state variable, $h_{t+2}$, in the pricing kernel.
Moreover, $\xi$ is also negatively related to the VRP, because $h_{t}-h_{t}^{\ast}=-\xi\times\tfrac{2\alpha h_{t}}{1-2\alpha\xi}\approx-\xi\times2\alpha h_{t}$,
where $2\alpha h_{t}>0$.

\subsection{Dynamic Pricing Kernel by Means of Variance Risk Aversion}

We seek a more flexible pricing kernel that is coherent with the observed
variation over time. Next, we introduce a new dynamic pricing kernel
while maintaining the Heston-Nandi GARCH structure (\ref{eq:HNGreturn})-(\ref{eq:HNGgarch})
under physical measure, $\mathbb{P}$. 
\begin{assumption}
\label{assu:HNG}Returns are given by the Heston-Nandi GARCH model,
(\ref{eq:HNGreturn})-(\ref{eq:HNGgarch}), under $\mathbb{P}$.
\end{assumption}
\begin{assumption}
\label{assu:Dyn-pricing-kernel}The pricing kernel takes the following
form:
\begin{equation}
M_{t+1,t}=\frac{\exp\left(\phi_{t}R_{t+1}+\xi_{t}h_{t+2}\right)}{\mathbb{E}_{t}^{\mathbb{P}}[\exp(\phi_{t}R_{t+1}+\xi_{t}h_{t+2})]},\label{eq:NewPK}
\end{equation}
where $\phi_{t}$ and $\xi_{t}$ are $\mathcal{G}_{t}$-measurable.
\end{assumption}
Assumptions \ref{assu:HNG} and \ref{assu:Dyn-pricing-kernel} induce
the following dynamic properties under risk-neutral measure.
\begin{thm}
\label{thm:DynUnderQ}Suppose Assumptions \ref{assu:HNG} and \ref{assu:Dyn-pricing-kernel}
hold and define $\eta_{t}\equiv(1-2\alpha\xi_{t})^{-1}$. \\
$(i)$ The risk aversion parameters satisfy
\[
\xi_{t}=\tfrac{1}{2\alpha}\tfrac{\eta_{t}-1}{\eta_{t}},\qquad\phi_{t}=\tfrac{\eta_{t}-1}{\eta_{t}}(\gamma-\tfrac{1}{2})-\tfrac{1}{\eta_{t}}\lambda,
\]
such that their dynamic properties are defined by that of $\eta_{t}$.\\
$(ii)$ The return dynamics under the risk-neutral measure $\mathbb{Q}$
are given by:
\begin{eqnarray*}
R_{t+1} & = & r-\tfrac{1}{2}h_{t+1}^{*}+\sqrt{h_{t+1}^{*}}z_{t+1}^{*},\\
h_{t+1}^{*} & = & \omega_{t}^{*}+\beta_{t}^{*}h_{t}^{*}+\alpha_{t}^{*}\left(z_{t}^{*}-\gamma_{t-1}^{*}\sqrt{h_{t}^{*}}\right)^{2},
\end{eqnarray*}
where $z_{t+1}^{*}|\mathcal{G}_{t}\overset{\mathbb{Q}}{\sim}iid\ N(0,1)$
and the following relations between $\mathbb{P}$ and $\mathbb{Q}$
parameters:
\[
\omega_{t}^{*}=\omega\eta_{t},\quad\beta_{t}^{*}=\beta\tfrac{\eta_{t}}{\eta_{t-1}},\quad\alpha_{t}^{*}=\alpha\eta_{t}\eta_{t-1},\quad\gamma_{t}^{*}=\tfrac{1}{\eta_{t}}(\gamma+\lambda-\tfrac{1}{2})+\tfrac{1}{2}.
\]
$(iii)$ The dynamic parameter, $\eta_{t}$, equals the ratio of the
conditional variances under $\mathbb{Q}$ and $\mathbb{P}$,
\[
\eta_{t}=\frac{h_{t+1}^{\ast}}{h_{t+1}}.
\]
\end{thm}
Theorem \ref{thm:DynUnderQ} reveals the following six interesting
properties of the risk-neutral probability measure and its dynamics.

First, $\eta_{t}=h_{t+1}^{\ast}/h_{t+1}$ emerges as a fundamental
quantity\footnote{Parameterize the model with the inverse ratio, $\theta_{t}=h_{t+1}/h_{t+1}^{\ast}$,
simplifies several expressions, such as $\xi_{t}=\tfrac{1}{2\alpha}(1-\theta_{t})$
and $\phi_{t}=(1-\theta_{t})(\gamma-\tfrac{1}{2})-\theta_{t}\lambda$.
We adopt $\eta_{t}$, because it is more common in the literature.}, which is termed the \emph{variance risk ratio}. It is $\mathcal{G}_{t}$-measurable
since $\eta_{t}$ is a function of $\xi_{t}\in\mathcal{G}_{t}$ defined
in Assumption \ref{assu:Dyn-pricing-kernel}. This variable appears
in all expressions relating $\mathbb{P}$ and $\mathbb{Q}$, and it
has a straightforward interpretation. Its functional relationship
with $\xi_{t}$ shows that $\eta_{t}$ is an indirect measure of volatility
risk aversion. When $\eta_{t}$ is large, agents demand a large compensation
for taking on variance risk, while a small value of $\eta_{t}$ corresponds
to an appetite for variance risk. Moreover, $\eta_{t}$ is obviously
related to the VRP, where the latter concerns the expectation of future
values of $h_{t+1}-h_{t+1}^{\ast}$. This follows from $h_{t+1}-h_{t+1}^{\ast}=\left(1-\eta_{t}\right)h_{t+1}$.
This relation could be used to construct a model-based measure of
the VRP. Note that the VRP can be positive in the most general version
of the model. A negative VRP can be guaranteed in the model design
by restricting $\eta_{t}$ to be greater than one.\footnote{This can be be achieved by specifying a dynamic model for $\log(\eta_{t}-1)$.
We do not impose this restriction in our analysis, as we specify a
model for $\log\eta_{t}$.} For comparison, the ratio, $h_{t}^{\ast}/h_{t}$, is constant in
CHNG model, and the relation $h_{t}-h_{t}^{\ast}=(1-\eta)h_{t}$ shows
that CHNG model implicitly restricts the VRP to be proportional to
the conditional variance under $\mathbb{P}$.

Second, the shape of the pricing kernel is time-varying, as shown
by the following results. 
\begin{lem}[Shape of Pricing Kernel]
\label{lem:ShapeDynamicKernel}The following expression
\begin{align*}
\log M_{t+1,t} & =-\lambda R_{t+1}+\xi_{t}\alpha\frac{\left(R_{t+1}-r\right)^{2}}{h_{t+1}}+\kappa_{0,t}+\kappa_{1,t}h_{t+1},
\end{align*}
is equivalent to (\ref{eq:NewPK}), where $\kappa_{0,t}=-\frac{1}{2}\log\eta_{t}$,
and $\kappa_{1,t}=\tfrac{1}{2}(\lambda-\tfrac{1}{2})^{2}-\tfrac{1}{8}\eta_{t}$.
\end{lem}
\noindent The expression in Lemma \ref{lem:ShapeDynamicKernel} is
the generalization of (\ref{eq:QuadraticInReturns}) to the case with
time-varying preference parameter, $\phi_{t}$ and $\xi_{t}$. The
quadratic coefficient can be time-varying, because it depends on $\xi_{t}$,
which makes it clear that $\xi_{t}$ influences the shape of the pricing
kernels. Or, equivalently, $\eta_{t}$ governs the shape. We have
a U-shape if $\xi_{t}>0$ ($\Leftrightarrow\eta_{t}>1$) and an inverted
U-shape if $\xi_{t}<0$ ($\Leftrightarrow\eta_{t}<1$).\footnote{We have $\alpha>0$ in the Heston-Nandi GARCH model.}
Thus, this framework can generate a variety of shapes of the pricing
kernel. This is illustrated in the upper left panel of Figure \ref{fig:PricingKernels-Shape-NIC-Skew-Kurt},
where the pricing kernel for cumulated returns over one month is presented
for two levels of $\eta$. The design is based on the empirical estimates
of the model in Section \ref{sec:Empirical-Analysis}, where $\eta_{\mathrm{low}}=0.70$
and $\eta_{\mathrm{high}}=1.70$ correspond to the $10\%$ and $90\%$
quantiles of the unconditional distribution of $\eta_{t}$, respectively.
The low value of $\eta$ produce an inverted U-shape similar to that
observed during the years 2004 to 2007, whereas the high value of
$\eta$ results in a shape with a pronounced U-shape.
\begin{figure}
\centering{}\includegraphics[width=1\textwidth]{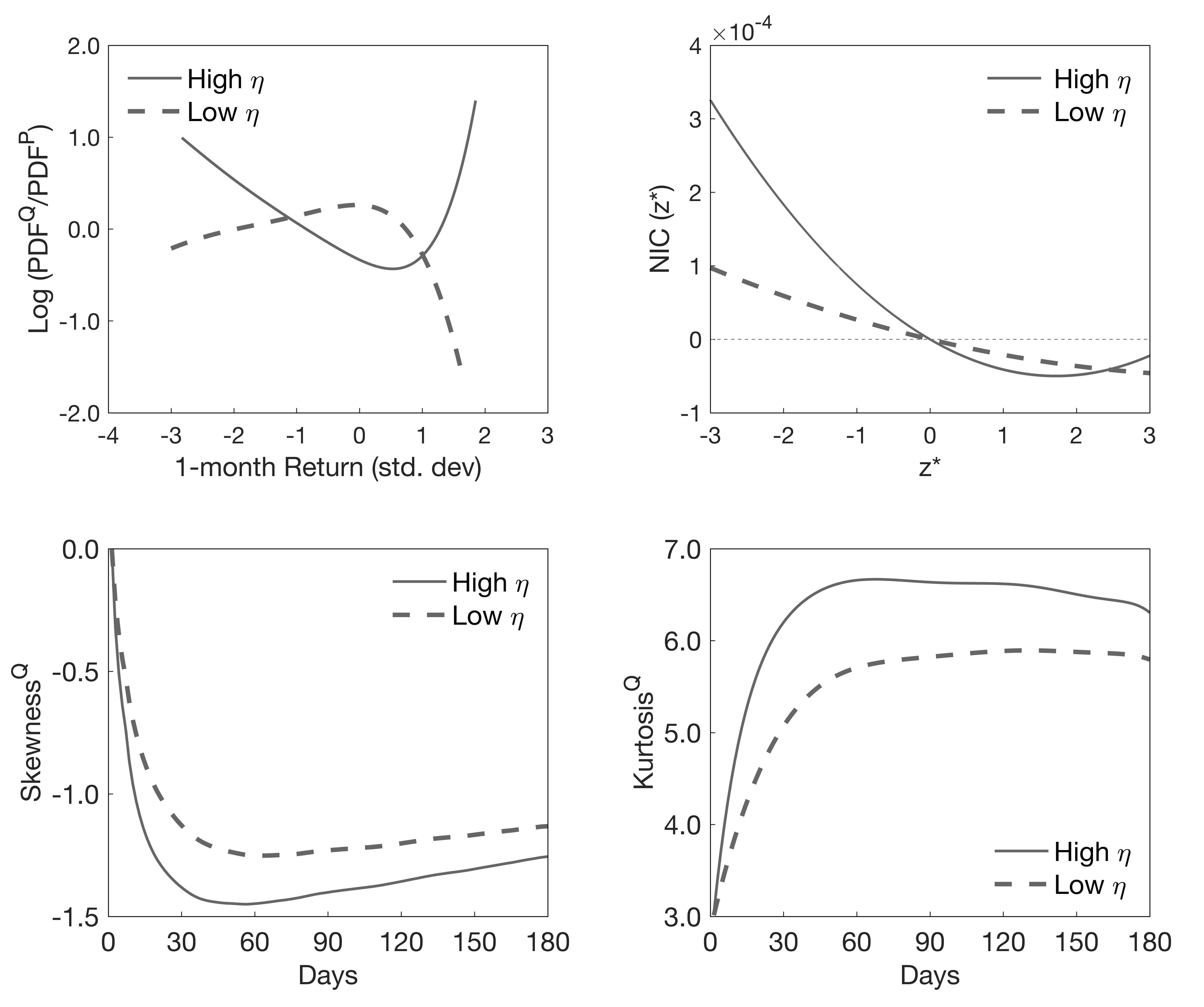}\caption{{\small Properties of cumulative returns for two levels of variance
risk ratio, $\eta_{\mathrm{low}}=0.70$ (dashed lines) and $\eta_{\mathrm{high}}=1.70$
(solid lines). The upper left panel presents the log-ratios of the
$\mathbb{Q}$-to-$\mathbb{P}$ densities for cumulative returns over
one month, the corresponding news impact curves under $\mathbb{Q}$
are shown in the upper right panel, and skewness and kurtosis of multi-period
cumulative returns under the risk-neutral measure $\mathbb{Q}$ are
shown in the two lower panels. These results are based on 1 million
simulations using a design based on our model estimates (last column
in Table \ref{tab:JointEstimation}), where $\eta_{\mathrm{low}}=0.70$
and $\eta_{\mathrm{high}}=1.70$ correspond to  the $10\%$ and $90\%$
quantiles of the unconditional distribution of $\eta_{t}$, respectively.\label{fig:PricingKernels-Shape-NIC-Skew-Kurt}}\protect \\
{\small}\protect \\
{\small Alt text: Properties of cumulative returns for two levels of
the variance risk ratio: low and high. The figure consists of four
panels displaying (1) the log-ratio of risk-neutral to physical densities,
(2) news impact curves, (3) skewness under the risk-neutral measure,
and (4) kurtosis under the risk-neutral measure.}}
\end{figure}

Third, the equity risk premium, $\lambda$, is constant, despite $\phi_{t}$
and $\xi_{t}$ being time-varying. This is an implication of the no-arbitrage
condition, that ties the linear coefficient to $\lambda$, which is
determined under the physical measure in the Heston-Nandi GARCH model.

Fourth, an important reason for permitting $\eta_{t}$ to be time-varying
is that it enables $h_{t}^{\ast}$ to incorporate information from
derivative prices. In conventional GARCH models, the conditional variance
depends only on lagged returns, such that $h_{t+1}\in\mathcal{F}_{t}^{R}$.
A restrictive implication of a constant VRR is that it requires the
two conditional variances to be proportional to each other. Consequently,
$h_{t+1}^{\ast}$ is driven solely by lagged returns, leaving no room
for derivative prices to influence its dynamic properties. Allowing
$\eta$ to be time-varying eliminates this rigidity and enables the
model to incorporate information from derivative prices into the dynamic
model of $h_{t+1}^{\ast}$. We take advantage of this feature in our
empirical implementation.

Fifth, the structure in Theorem \ref{thm:DynUnderQ} nests the CHNG
model of \citet{ChristoffersenHestonJacobs2013} and the HNG model,
which corresponds to $\eta_{t}=\eta$ and $\eta_{t}=1$, respectively.
If $h_{t}$ is constant and $\eta_{t}=1$, then it leads to the pricing
kernel given by the power utility function, see \citet{Rubinstein1976}.
So the famous Black-Scholes model is also nested as the special case.
Interestingly, when $\eta_{t}$ is time-varying, there is time variation
in the risk aversion parameters, which is similar to the habit formation
models, see \citet{CampbellCochrane1999}.

Sixth, the structure introduced in Theorem \ref{thm:DynUnderQ} generates
time-varying leverage under $\mathbb{Q}$, which manifests itself
in a number of ways. A time-varying leverage effect will impact the
distribution of cumulative returns, and the lower panels of Figure
\ref{fig:PricingKernels-Shape-NIC-Skew-Kurt} show the skewness and
kurtosis of cumulative returns under $\mathbb{Q}$ as a function of
the number of days that returns are cumulated over (the $x$-axis).
Skewness and kurtosis are shown for the two levels of the VRR, $\eta_{\mathrm{low}}$
and $\eta_{\mathrm{high}}$, and these result in distinctively different
levels of skewness and kurtosis under $\text{\ensuremath{\mathbb{Q}}}$.
Both skewness and kurtosis are far more pronounced for the large value
of $\eta$. Another way to illustrate the leverage effect is with
the news impact curve by \citet{Engle_Ng_1993}. The news impact curve
under $\mathbb{Q}$ is defined by $\mathrm{NIC}(z^{\ast})=\mathbb{E}_{t}^{\mathbb{Q}}(h_{t+1}^{\ast}|z_{t}^{\ast}=z^{\ast})-\mathbb{E}_{t}^{\mathbb{Q}}(h_{t+1}^{\ast}|z_{t}^{\ast}=0)$,
and from Theorem \ref{thm:DynUnderQ} it follows that $\mathrm{NIC}(z^{\ast})=\alpha\eta_{t}\eta_{t-1}z^{\ast2}-2\alpha\gamma\eta_{t}\sqrt{h_{t}^{*}}z^{\ast}$,
whose shape depends on the VRR. The upper right panel of Figure \ref{fig:PricingKernels-Shape-NIC-Skew-Kurt}
presents the news impact curve, with $h_{t}^{\ast}$ equal to its
unconditional mean under $\mathbb{Q}$, as estimated in Section \ref{sec:Empirical-Analysis},
and with $\eta=\eta_{t}=\eta_{t-1}$ equal to either $\eta_{\mathrm{low}}$
(dashed line) or $\eta_{\mathrm{high}}$ (solid line). We observe
the asymmetric shape of the news impact curve, see e.g. \citet{ChenGhysels:2011-news-impact-curve}.
However, the level of asymmetry depends on the level of $\eta$, such
that volatility is more responsive to negative return shocks when
$\eta$ is large, than when $\eta$ is small. 

\section{Derivatives Pricing with Dynamic Volatility Risk Aversion\label{sec:Derivative-Prices}}

In this section, we consider cases where $\eta_{t}$ is time-varying
and stochastic. We derive a closed-form expression for the VIX under
the assumption that $\log\eta_{t}$ follows an autoregressive moving
average process, ARMA$(p,q)$. We also obtain analytical expressions
for option prices under the same assumptions. The option pricing formula
is obtained with a novel approximation method based on a Taylor expansion
of the MGF for cumulative returns. 

\subsection{VIX  Pricing}

It is relatively straightforward to obtain closed-form expressions
for VIX pricing, once the dynamic properties of $\eta_{t}$ under
$\mathbb{Q}$ are known. This only requires expressions for $\mathbb{E}_{t}^{\mathbb{Q}}(h_{t+k}^{*})$,
$k=1,\ldots,M$, because the risk-neutral dynamics of returns under
Theorem \ref{thm:DynUnderQ} implies the following $M$-period ahead
VIX pricing formula 
\begin{equation}
{\rm VIX}_{t}=A\times\sqrt{\frac{1}{M}\sum_{k=1}^{M}\mathbb{E}_{t}^{\mathbb{Q}}(h_{t+k}^{*})},\label{eq:VIXform}
\end{equation}
where $A=100\sqrt{252}$ is the annualizing factor.
\begin{assumption}
\label{assu:ARMA}Suppose that the logarithm of variance risk ratio,
$\log\eta_{t}$, follows an autoregressive moving average process
of orders $(p,q)$,
\begin{equation}
\varphi(L)(\log\eta_{t}-\zeta)=\theta(L)\varepsilon_{t},\label{eq:ARMA}
\end{equation}
where $\varphi(x)=0\Rightarrow|x|>1$, $\theta(x)$ and $\varphi(x)$
have no common roots. The innovation $\varepsilon_{t}$ is $\mathcal{G}_{t}$-measurable,
and $\varepsilon_{t+1}|\mathcal{G}_{t}$ follows an iid distribution
with $\mathbb{E}_{t}^{\mathbb{Q}}(\varepsilon_{t+1})=0$, $\mathrm{var}_{t}^{\mathbb{Q}}(\varepsilon_{t+1})=\sigma^{2}$,
and a MGF satisfying $\operatorname{mgf}_{\varepsilon}(s)\equiv\mathbb{E}_{t}^{\mathbb{Q}}\left[\exp\left(s\varepsilon_{t+1}\right)\right]<\infty$
a.s. for $s\in\mathbb{R}$. Additionally, $\varepsilon_{t}$ is independent
of the return shock $z_{\tau}$ for all $\tau$. 
\end{assumption}
Assumption \ref{assu:ARMA} is explicit about the first two conditional
moments of $\varepsilon_{t}$, but does not impose restrictive distributional
assumptions. However, the exact distribution is important for VIX
and option prices, as these depend on $\operatorname{mgf}_{\varepsilon}(s)$.
That $\varepsilon_{t}$ is $\mathcal{G}_{t}$-measurable follows from
Assumption \ref{assu:Dyn-pricing-kernel} and therefore not a new
assumption.
\begin{thm}[VIX pricing]
\label{thm:VixPricing}Suppose that Assumptions \ref{assu:HNG}-\ref{assu:ARMA}
hold, then the analytical expression for $M$-period ahead VIX in
(\ref{eq:VIXform}) is given by
\begin{equation}
\mathrm{VIX}(M,\sigma^{2},h_{t+1}^{\ast})=\sqrt{a_{1}(M,\sigma^{2})+a_{2}(M,\sigma^{2})h_{t+1}^{*},}\label{eq:VixPricing}
\end{equation}
where $a_{1}(M,\sigma^{2})$ and $a_{2}(M,\sigma^{2})$ both depend
on $(\eta_{t},\ldots,\eta_{t-p+1})$, $(\varepsilon_{t},\ldots,\varepsilon_{t-q+1})$,
and $\operatorname{mgf}_{\varepsilon}(s)$. See Appendix \ref{sec:Proof_VIXpricing}
for their exact expressions.\footnote{We have suppressed the terms dependence on $\varphi$, $\zeta$, and
the parameters of the HNG model and (\ref{eq:VixPricing}) suppresses
an approximation term that is negligible in practice (see Lemma \ref{lem:bound}).}
\end{thm}

\subsection{Option Pricing: A Novel Analytical Approximation}

In this section, we derive the European option pricing formula for
the case where the variance risk aversion is time-varying. 

The European call option price at time $t$ is given by the conditional
risk-neutral expectation $C_{t}=\ensuremath{e^{-r\left(T-t\right)}\mathbb{E}_{t}^{\mathbb{Q}}\left[\max\left(S_{T}-K,0\right)\right]}$,
where $T$ is the maturity date, $K$ is the strike price, and $S_{T}$
is the terminal price of underlying asset. Let $M=T-t$ be the number
of periods to maturity, it is well-known that an affine structure
of the MGF of future cumulative returns,
\[
g_{t,M}(s)=\mathbb{E}_{t}^{\mathbb{Q}}\left[\exp\left(s\sum_{i=1}^{M}R_{t+i}\right)\right],\quad s\in\mathbb{R}
\]
is the key to closed-form option pricing formula, see \citet{HestonNandi2000}.
Unfortunately, $g_{t,M}(s)$ does not have an affine structure when
the variance risk aversion follows a stochastic process. We propose
a novel approximation method that extrapolates from the solution to
a simpler auxiliary problem, leading to an approximation of $g_{t,M}(s)$. 

The future values of $\log\eta_{t}$ that are relevant for pricing
an option with $M$ periods to maturity are the elements of the vector,
$\boldsymbol{\eta}_{t,M}=\left(\log\eta_{t+1},\ldots,\log\eta_{t+M}\right)^{\prime}$.
Under Assumption \ref{assu:ARMA}, this is a random vector. For later
use, we use the notation, $\bar{\boldsymbol{\eta}}_{t,M}$, to represent
a predetermined path (i.e. $\bar{\boldsymbol{\eta}}_{t,M}\in\mathcal{G}_{t}$),
where the prime example is the conditionally expected trajectory,
given by $\bar{\boldsymbol{\eta}}_{t,M}^{e}=\mathbb{E}_{t}^{\mathbb{Q}}(\boldsymbol{\eta}_{t,M})\in\mathbb{R}^{M\times1}$.

\subsubsection{MGF with Predetermined Path}

The assumption that $\log\eta_{t}$ will follow a deterministic trajectory,
$\bar{\boldsymbol{\eta}}_{t,M}$, implies that $\phi_{t+j}$ and $\xi_{t+j}$
are also predetermined for $j=1,\ldots,M$, and similarly, $\omega_{t+j}^{*}$,
$\alpha_{t+j}^{*}$, $\beta_{t+j}^{*}$, and $\gamma_{t+j}^{*}$ (the
GARCH parameters under $\mathbb{Q}$) are also predetermined for $j=1,\ldots,M$.
This is a direct consequence of Theorem \ref{thm:DynUnderQ}. We show
that the affine structure is preserved in this case with the following
closed-form expression for the MGF of cumulative returns. 
\begin{lem}
\label{lem:Closed-form-OP}Suppose that Assumptions \ref{assu:HNG}-\ref{assu:Dyn-pricing-kernel}
hold and that $\log\eta_{t}$ take the predetermined path, $\bar{\boldsymbol{\eta}}_{t,M}$.
Then the MGF for future cumulative returns under $\mathbb{Q}$ has
the affine form,
\begin{align*}
g_{t,M}(s|\bar{\boldsymbol{\eta}}_{t,M}) & =\exp\left(A_{T}(s,M)+B_{T}(s,M)h_{t+1}^{*}\right),\quad T=t+M,
\end{align*}
where the expressions for $A_{T}(s,M)$ and $B_{T}(s,M)$ depend on
$\bar{\boldsymbol{\eta}}_{t,M}$, and both are given from simple recursive
expressions, see Appendix \ref{sec:Proof_OP}.
\end{lem}
The expression for the MGF in Lemma \ref{lem:Closed-form-OP} simplifies
to that in \citet{ChristoffersenHestonJacobs2013} if $\bar{\boldsymbol{\eta}}_{t,M}=(\log\eta,\ldots,\log\eta)^{\prime}$,
and it simplifies further to that in \citet{HestonNandi2000} when
$\bar{\boldsymbol{\eta}}_{t,M}=0_{M\times1}$. Aside from these two
special cases, this naive approach does not constitute a coherent
model for option pricing. However, it serves as an important auxiliary
``model'' with the desired affine structure.

\subsubsection{Option Pricing with Stochastic Variance Risk Aversion\label{subsec:OPRand}}

We now turn to the more challenging problem where $\boldsymbol{\eta}_{t,M}$
is stochastic. The lack of an affine structure thwarts the standard
approach to obtaining closed-form option pricing formula, and it is
common to resort to simulation methods in this case. 

The idea behind our approximation method is simply to extrapolate
from the simple case, $g_{t,M}(s|\bar{\boldsymbol{\eta}}_{t,M}^{e})$,
and apply a second-order Taylor expansion about $\bar{\boldsymbol{\eta}}_{t,M}^{e}$.
This leads to a quadratic expression in the $M$-dimensional vector,
$\boldsymbol{\varepsilon}_{t,M}=\boldsymbol{\eta}_{t,M}-\bar{\boldsymbol{\eta}}_{t,M}^{e}$.
After taking the conditional expectation, we arrive at the approximate
MGF, which is a perturbation of $g_{t,M}(s|\bar{\boldsymbol{\eta}}_{t,M}^{e})$
to account for the randomness in $\boldsymbol{\eta}_{t,M}$. The adjustment
of $g_{t,M}(s|\bar{\boldsymbol{\eta}}_{t,M}^{e})$ is simple, because
Assumption \ref{assu:ARMA} implies that $\mathbb{E}_{t}^{\mathbb{Q}}[\boldsymbol{\varepsilon}_{t,M}]=0$
and makes it simple to evaluate $\mathbb{E}_{t}^{\mathbb{Q}}[\boldsymbol{\varepsilon}_{t,M}\boldsymbol{\varepsilon}_{t,M}^{\prime}]=\Sigma_{M}$.
\begin{thm}
\label{thm:OptionPricingRandom}Suppose that Assumptions \ref{assu:HNG}-\ref{assu:ARMA}
hold, then the approximate model-implied MGF, based on a second-order
Taylor expansion, is given by
\[
\hat{g}_{t,M}(s)=g_{t,M}(s|\bar{\boldsymbol{\eta}}_{t,M}^{e})\left[1+\frac{1}{2}{\rm tr}\{H_{t,M}(s)\Sigma_{M}\}\right],
\]
where $H_{t,M}(s)\in\mathbb{R}^{M\times M}$ is the Hessian from the
Taylor expansion and $\Sigma_{M}=\operatorname{var}_{t}^{\mathbb{Q}}(\boldsymbol{\varepsilon}_{t,M})$,
see Appendix \ref{sec:ProofOPrandom} for details.\footnote{While $g_{t,M}(s)$ is a MGF, the approximation, $\hat{g}_{t,M}(s)$,
need not be one, because there may not exist a density for which $\hat{g}_{t,M}(s)$
is the corresponding MGF.} The corresponding European call option price is given by
\[
\hat{C}(S_{t},M,K,r;h_{t+1}^{\ast})=S_{t}P_{1,t}-K\exp(-rM)P_{2,t},
\]
where
\begin{eqnarray*}
P_{1,t} & = & \text{\ensuremath{\frac{1}{2}+\frac{\exp(-rM)}{\pi}\int_{0}^{\infty}{\rm Re}\left[\frac{K^{-iu}\hat{g}_{t,M}(iu+1)}{iuS_{t}}\right]du},}\\
P_{2,t} & = & \frac{1}{2}+\frac{1}{\pi}\int_{0}^{\infty}{\rm Re}\left[\frac{K^{-iu}\hat{g}_{t,M}(iu)}{iu}\right]du.
\end{eqnarray*}
\end{thm}
Theorem \ref{thm:OptionPricingRandom} clarifies how random variation
in $\eta_{t}$ impacts the MGF.\footnote{The integrations in Theorem \ref{thm:OptionPricingRandom} can be
numerically evaluated through the method in Appendix \ref{subsec:NumeCom}.} The core principle behind the proposed approximation method is akin
to perturbation methods used to find approximate solutions to Schrödinger
equations for which standard methods are not applicable. Perturbation
methods also starts from an exact solution of a simpler, related problem,
from which an approximate solution to the actual problem is deduced.
Our approximation is based on a Taylor expansion about an $M$-dimensional
vector, where the conditional moments of the terms in the expansion
are deduced from Assumption \ref{assu:ARMA}. Next, we evaluate the
accuracy of the approximation.
\begin{figure}
\centering{}\includegraphics[width=1\textwidth]{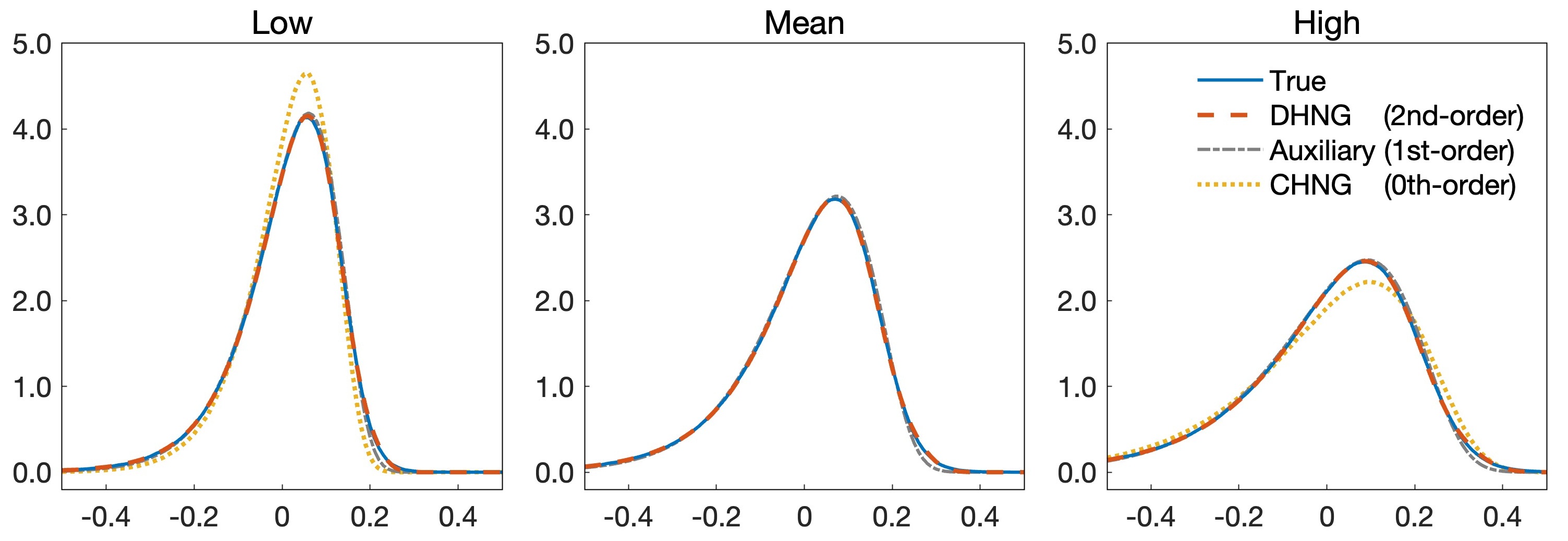}\caption{{\small Risk-neutral densities of (six months) cumulative return, for
different initial values of variance risk ratio, $\eta_{0}=0.70$
(left panel), $\eta_{0}=1.11$ (middle panel), and $\eta_{0}=1.70$
(right panel). The blue solid line represents the true density from
1 million simulations, the red dashed line is the (2nd-order) approximate
density by DHNG, the gray dash-dotted line is the auxiliary model
based on $\bar{\boldsymbol{\eta}}_{t,M}^{e}$ (1st-order), and the
yellow dotted line is the (0th-order) density by CHNG with constant
$\eta_{t}$. The model parameters are based on the estimates with
option prices (the last column in Table \ref{tab:JointEstimation}).\label{fig:Density_True_Appro}}\protect \\
{\small}\protect \\
{\small Alt text: Risk-neutral densities for cumulative six-month returns,
shown for three initial levels of the variance risk ratio: low, middle,
and high. Each density is compared with three approximations---CHNG,
the auxiliary model, and DHNG. DHNG provides the most accurate approximation.}}
\end{figure}

\subsection{Accuracy of Analytical Approximation}

The option pricing formula in Theorem \ref{thm:OptionPricingRandom}
involves a second-order approximation that accounts for the first
two conditional moments of $\text{\ensuremath{\boldsymbol{\eta}_{t,M}}}$.\footnote{If the distribution of $\boldsymbol{\varepsilon}_{t,M}$ is symmetric,
(e.g. Gaussian), then it is an approximation to third-order.} A first-order approximation is equivalent to assuming that $\eta_{t}$
will take the path of its conditional expectation, $\bar{\boldsymbol{\eta}}_{t,M}^{e}$.
Thus, $g_{t,M}(\cdot|\bar{\boldsymbol{\eta}}_{t,M}^{e})$ is the MGF
for the first-order approximation, which we previously labeled ``Auxiliary'',
because it was used as an intermediate step towards our preferred
approximation, $\hat{g}_{t,M}(\cdot)$. Taking $\eta_{t}=\eta_{0}$
to be constant, $\bar{\boldsymbol{\eta}}_{t,M}^{c}=(\log\eta_{0},\ldots,\log\eta_{0})^{\prime}$,
is a characteristic of CHNG, and this case can be interpreted as the
zero-order ``approximation''.

We first consider the risk-neutral density of cumulative returns.
The blue solid lines in Figure \ref{fig:Density_True_Appro} present
the true risk-neutral densities for cumulative returns over six months
for three initial values, $\eta_{0}=0.70$ (left panel), $\eta_{0}=1.11$
(middle panel), and $\eta_{0}=1.70$ (right panel), based on 1 million
simulations. The data generation process is based on the structure
in Theorem \ref{thm:DynUnderQ}, using the parameter estimates from
our empirical analysis of option prices with $\log\eta_{t}\sim\operatorname{AR}(1)$,
normally distributed innovations, and $h_{1}$ initialized at its
unconditional mean. In this design, the exponential of unconditional
mean, $\exp\left(\mathbb{E}\log\eta_{t}\right)=\exp\left(\zeta\right)$,
is used as the initial value $\eta_{0}$ in the middle panel.

The red dashed lines are the risk-neutral densities of DHNG implied
by $\hat{g}_{t,M}(\cdot)$, the gray dot-dashed lines are based on
the auxiliary, $g_{t,M}(\cdot|\bar{\boldsymbol{\eta}}_{t,M}^{e})$,
and the yellow dotted line is the risk-neutral density for CHNG, i.e.
$g_{t,M}(\cdot|\bar{\boldsymbol{\eta}}_{t,M}^{c})$. DHNG, which is
based on the second-order approximation of Theorem \ref{thm:OptionPricingRandom},
is accurate, whereas the Auxiliary structure fails to match the upper
tail of the densities. The CHNG is the worst approximation of the
true density in all cases. It is tied with Auxiliary in the middle
panel, because the two are identical when $\log\eta_{0}$ is set to
have its unconditional mean $\zeta$.

The accuracy of the densities can also be assessed by the moments
of cumulative returns under $\mathbb{Q}$. We plot the variance, skewness
and kurtosis of multi-period cumulative returns in Figure \ref{fig:True_Appro}.
The variance is presented in annualized units and scaled by 100, such
that 4 corresponds to $\sqrt{0.04}=20\%$ annualized volatility. The
$k$-th moment of cumulative returns can be computed by evaluating
the $k$-th derivative of the MGF in Theorem \ref{thm:OptionPricingRandom}
at zero. However, Theorem \ref{thm:moments}, presented below, provides
a much simpler and more direct method for evaluating moments. Figure
\ref{fig:True_Appro} shows that neglecting the random variation in
$\eta_{t}$ yields moments that are far from their true values, whereas
the second-order approximation is quite accurate, especially for the
second and third moment. Both CHNG and Auxiliary are far less accurate,
with CHNG being the least accurate, except in the middle column of
panels, where $\bar{\boldsymbol{\eta}}_{t,M}^{c}=\bar{\boldsymbol{\eta}}_{t,M}^{e}$.
\begin{sidewaysfigure}
\centering{}\includegraphics[width=1\textwidth]{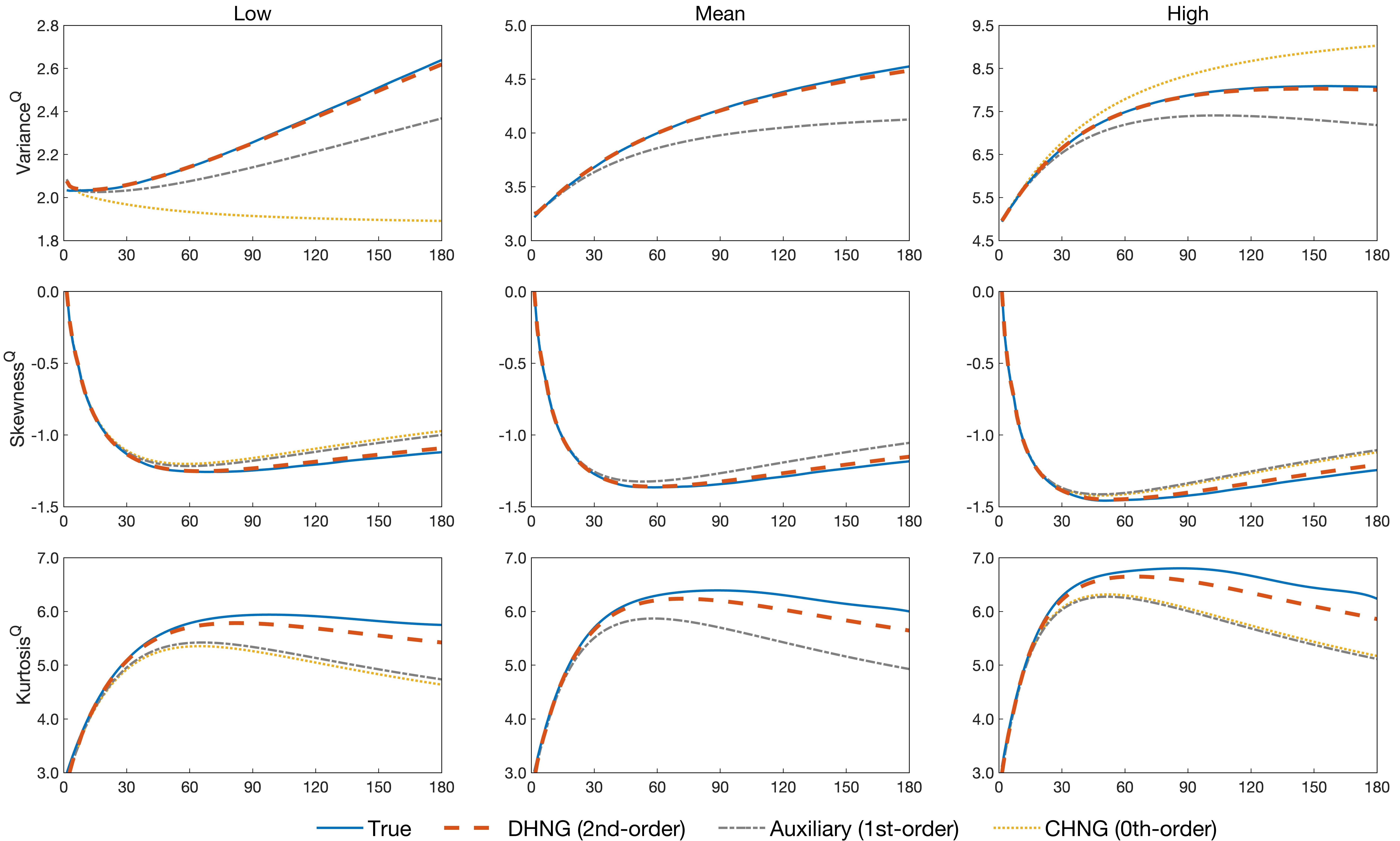}\caption{{\small Annualized variance, $\frac{252}{M}{\rm var}_{t}^{\mathbb{Q}}(\sum_{i=1}^{M}R_{t+i})\times100\%$,
skewness, and kurtosis of cumulative returns over $M$ days, are shown
with three different initial variance risk ratios, $\eta_{0}=0.70$
(left panels), $\eta_{0}=1.11$ (middle panels), and $\eta_{0}=1.70$
(right panels). Blue lines represent true population moments, red
dashed lines are those of DHNG, yellow dotted lines are those for
CHNG, and the gray dash-dotted lines are those for the auxiliary ``model''.
Population moments were obtained from 1 million simulated trajectories,
using a design based on the parameter estimates obtained with option
prices (the last column in Table \ref{tab:JointEstimation}).\label{fig:True_Appro}}\protect \\
{\small}\protect \\
{\small Alt text: Variance, skewness, and kurtosis of cumulative returns
for three initial levels of the variance risk ratio: low, middle,
and high. The figure compares the term structure of true moments with
those estimated using CHNG, the auxiliary model, and DHNG, with DHNG
providing the most accurate approximation.}}
\end{sidewaysfigure}

\begin{thm}
\label{thm:moments}The conditional risk-neutral $k$-th moment for
$k\in\mathbb{N}$ of the $M$-period ahead cumulative returns can
be expressed as
\begin{eqnarray*}
\mathbb{E}_{t}^{\mathbb{Q}}\left(R_{t,M}^{k}\right) & = & \frac{1}{\pi}\int_{0}^{\infty}\ensuremath{{\rm Re}\left[\frac{k!}{v^{k+1}}g_{t,M}(v)-\frac{k!}{u^{k+1}}g_{t,M}(u)\right]\mathrm{d}s,}
\end{eqnarray*}
where $R_{t,M}=\sum_{i=1}^{M}R_{t+i}$, and $u,v$ are two complex
numbers denoted by $u=u_{R}+is$ and $v=v_{R}+is$, with $u_{R}<0$,
$v_{R}>0$ and $s\in\mathbb{R}$. The function, $g_{t,M}\left(u\right)=\mathbb{E}_{t}^{\mathbb{Q}}\left[\exp\left(uR_{t,M}\right)\right]$,
is the conditional MGF of the cumulated returns over $M$ periods.
\end{thm}
The expression for the conditional moments in Theorem \ref{thm:moments}
is, to the best of our knowledge, a new result, and it is applicable
to any dynamic process with a well-defined MGF.

As a third way to assess the accuracy, we compare the (model) implied
volatilities with the true implied volatilities for options with different
levels of moneyness (Black-Scholes Delta) and days-to-maturity (DTM).
Table \ref{tab:ApproError} reports the percentage errors, $e_{i}\equiv100\times\left[\log{\rm IV}_{i}^{{\rm model}}-\log{\rm IV}_{i}^{{\rm true}}\right]$,
where the implied volatilities are deduced from option prices by the
Black-Scholes formula. The data generating process is the same as
that used in Figures \ref{fig:Density_True_Appro} and \ref{fig:True_Appro}.
We consider three initial values of $\eta_{0}$ over a two-dimensional
grid for Delta and DTM. 

Table \ref{tab:ApproError} shows that the DHNG option pricing formula
of Theorem \ref{thm:OptionPricingRandom} is far more accurate than
those of Auxiliary and CHNG. This is true across all types of options
and initial values of $\eta_{0}$. The absolute approximation error
(MAE) is just 0.08\% on average over all designs in Table \ref{tab:ApproError},
and the average error (ME) is -0.05\%. Unsurprisingly, CHNG is the
least accurate option pricing model. The assumptions that $\eta_{t}$
is constant leads to underpricing of options when $\eta_{0}$ is small
and overpricing when $\eta_{0}$ is large. The option pricing formula
for the Auxiliary structure systematically underprices options, and
we observe that the pricing error is increasing in DTM. This is consistent
with the results in Figure \ref{fig:True_Appro}, where the predicted
values of $h_{t+1}^{\ast},\ldots,h_{t+M}^{\ast}$ for the Auxiliary
structure are downward-biased and increasingly so as $M$ increases.
This has implications for empirical estimations under the Auxiliary
structure. To rationalize the observed option prices, the Auxiliary
structure will need to exaggerate the predicted values of $\log\eta_{t}$,
such that Auxiliary is expected to yield larger predictions of $\log\eta_{t}$
than those based on DHNG when estimated with option prices. This is
indeed what we find in the supplementary results, see Table \ref{tab:JointEstimation_Appro}
in Appendix \ref{subsec:ResultsAuxi}.
\begin{table}
\caption{Approximation Errors by Logarithm of Implied Volatility (in percentage)}

\begin{centering}
\vspace{0.2cm}
\begin{footnotesize}
\begin{tabularx}{\textwidth}{YYp{-0.1cm}YYYp{-0.1cm}YYp{-0.1cm}YYYY}
\toprule  
\midrule
          &       &       & \multicolumn{3}{c}{Low $\eta$} &       & \multicolumn{2}{c}{Mean $\eta$} &       & \multicolumn{3}{c}{High $\eta$} \\
\cmidrule{4-6}\cmidrule{8-9}\cmidrule{11-13}    Delta & DTM   &       & CHNG  & Auxiliary & DHNG &       & Auxiliary & DHNG &       & CHNG  & Auxiliary & DHNG \\
    \midrule
          &       &       &       &       &       &       &       &       &       &       &       &  \\
          & 30    &       & -2.09 & -0.41 & \ 0.10  &       & -0.48 & \ 0.11  &       & 1.16  & -0.72 & -0.02 \\
          & 60    &       & -5.06 & -1.30 & \ 0.02  &       & -1.55 & -0.04 &       & 2.55  & -1.69 & -0.05 \\
    0.3   & 90    &       & -8.26 & -2.23 & -0.06 &       & -2.41 & -0.07 &       & 3.93  & -2.56 & -0.06 \\
          & 120   &       & -11.21 & -2.99 & -0.07 &       & -3.15 & -0.06 &       & 5.26  & -3.28 & -0.06 \\
          & 150   &       & -13.99 & -3.70 & -0.09 &       & -3.84 & -0.09 &       & 6.50  & -3.93 & -0.09 \\
          &       &       &       &       &       &       &       &       &       &       &       &  \\
          & 30    &       & -1.62 & -0.02 & \ 0.41  &       & -0.49 & 0.04  &       & 1.12  & -0.68 & -0.07 \\
          & 60    &       & -4.88 & -1.23 & -0.04 &       & -1.34 & -0.05 &       & 2.61  & -1.44 & -0.05 \\
    0.4   & 90    &       & -7.83 & -1.96 & -0.07 &       & -2.08 & -0.06 &       & 4.03  & -2.18 & -0.06 \\
          & 120   &       & -10.64 & -2.62 & -0.05 &       & -2.72 & -0.05 &       & 5.41  & -2.80 & -0.05 \\
          & 150   &       & -13.34 & -3.26 & -0.07 &       & -3.33 & -0.07 &       & 6.69  & -3.37 & -0.07 \\
          &       &       &       &       &       &       &       &       &       &       &       &  \\
          & 30    &       & -1.74 & -0.20 & \ 0.22  &       & -0.69 & -0.18 &       & 1.11  & -0.63 & -0.09 \\
          & 60    &       & -4.83 & -1.21 & -0.10 &       & -1.23 & -0.05 &       & 2.56  & -1.32 & -0.05 \\
    0.5   & 90    &       & -7.59 & -1.82 & -0.06 &       & -1.92 & -0.06 &       & 3.98  & -1.99 & -0.06 \\
          & 120   &       & -10.38 & -2.46 & -0.04 &       & -2.53 & -0.04 &       & 5.37  & -2.57 & -0.04 \\
          & 150   &       & -13.07 & -3.07 & -0.06 &       & -3.10 & -0.06 &       & 6.66  & -3.11 & -0.05 \\
          &       &       &       &       &       &       &       &       &       &       &       &  \\
          & 30    &       & -2.32 & -0.76 & -0.30 &       & -0.65 & -0.17 &       & 1.13  & -0.55 & -0.03 \\
          & 60    &       & -4.72 & -1.11 & -0.05 &       & -1.22 & -0.05 &       & 2.46  & -1.29 & -0.06 \\
    0.6   & 90    &       & -7.57 & -1.82 & -0.04 &       & -1.89 & -0.05 &       & 3.85  & -1.94 & -0.05 \\
          & 120   &       & -10.40 & -2.46 & -0.04 &       & -2.50 & -0.03 &       & 5.22  & -2.51 & -0.03 \\
          & 150   &       & -13.16 & -3.09 & -0.05 &       & -3.08 & -0.05 &       & 6.50  & -3.05 & -0.05 \\
          &       &       &       &       &       &       &       &       &       &       &       &  \\
          & 30    &       & -2.70 & -1.05 & -0.56 &       & -0.38 & \ 0.10  &       & 0.99  & -0.65 & -0.10 \\
          & 60    &       & -4.74 & -1.14 & -0.01 &       & -1.28 & -0.08 &       & 2.31  & -1.34 & -0.07 \\
    0.7   & 90    &       & -7.82 & -1.94 & -0.06 &       & -1.98 & -0.05 &       & 3.66  & -2.01 & -0.05 \\
          & 120   &       & -10.81 & -2.65 & -0.04 &       & -2.65 & -0.04 &       & 4.97  & -2.62 & -0.04 \\
          & 150   &       & -13.78 & -3.32 & -0.04 &       & -3.26 & -0.04 &       & 6.25  & -3.19 & -0.04 \\
          &       &       &       &       &       &       &       &       &       &       &       &  \\
    MAE   &       &       & \ 7.78  & \ 1.91  & \ 0.11  &       & \ 1.99  & \ 0.07  &       & 3.85  & \ 2.06  & \ 0.06 \\
    ME    &       &       & -7.78 & -1.91 & -0.05 &       & -1.99 & -0.05 &       & 3.85  & -2.06 & -0.06 \\
\\[0.0cm]
\\[-0.5cm]
\midrule
\bottomrule
\end{tabularx}
\end{footnotesize}
\par\end{centering}
{\small Note: This table reports the approximation errors, $e_{i}=\left[\log{\rm IV}_{i}^{{\rm model}}-\log{\rm IV}_{i}^{{\rm true}}\right]\times100$
for options with different Moneyness (Delta), days to maturity (DTM),
and initial value of $\eta$. The true option prices are obtained
with 10 million simulated paths from the true model, and the implied
volatilities are computed using Black-Scholes formula. The column-wise
averages for the mean absolute error (MAE) and the mean error (ME),
are presented in the last two rows. CHNG, Auxiliary, and DHNG correspond
to zero-, first-, and second-order approximations, respectively. The
model parameters are based on the empirical estimates obtained with
option prices (the last column in Table \ref{tab:JointEstimation}).\label{tab:ApproError}}{\small\par}
\end{table}

\section{Observation-Driven Model for Variance Risk Aversion\label{sec:ScoreDrivenModel}}

In this section, we develop an observation-driven model for $\eta_{t}$,
which facilitates an implementation of the results from the previous
section. In practice, this requires inferring the current value of
$\eta_{t}$ and estimate its dynamic model. To this end, we adopt
an intuitive observation-driven model for $\eta_{t}$, inspired by
the score-driven framework of \citet{CrealKoopmanLucas:2013}. We
use the first-order conditions for minimizing pricing errors to define
innovations to $\eta_{t}$, ensuring that $\eta_{t}$ is adjusted
to reduce pricing errors on average.\footnote{Alternatively, one could adopt state space approach. This is computationally
more complicated without necessarily providing any benefit. Even when
the true model is a state space model, the score-driven models are
typically found to be competitive, see \citet{KoopmanLucasScharth:2016}.} The first-order condition of the log-likelihood function is presented
in the empirical section. For now, it suffices to express the log-likelihood
function in its generic form, $\sum_{t=1}^{T}\ell(R_{t},X_{t}|\mathcal{F}_{t-1})$,
where $T$ is the sample size, and we should emphasize that it relies
on information from both $\mathbb{P}$ and $\mathbb{Q}$. We factorize
$\ell(R_{t},X_{t}|\mathcal{F}_{t-1})=\ell(R_{t}|\mathcal{F}_{t-1})+\ell(X_{t}|R_{t},\mathcal{F}_{t-1})$,
where the latter can be expressed as $\ell(X_{t}|\mathcal{G}_{t})$.
This likelihood term measures how well the observed derivative prices
are explained by the statistical model for returns and the pricing
kernel.

A score-driven model updates a parameter in the direction dictated
by the first-order conditions of the log-likelihood function, known
as the score.\footnote{The score-driven approach is locally optimal in the Kullback-Leibler
sense, see \citet{BlasquesKoopmanLucas2015}.} In our model, the relevant score is $\partial\ell(X_{t}|\mathcal{G}_{t})/\partial\log\eta_{t}$,
because $\ell(R_{t}|\mathcal{F}_{t-1})$ does not depend on $\eta_{t}$.
The required structure for $\varepsilon_{t}$, as specified in Assumption
\ref{assu:ARMA}, motivates the choice
\begin{equation}
\varepsilon_{t+1}=\sigma s_{t},\label{eq:Def-of-e}
\end{equation}
where the normalized score is defined by
\begin{equation}
s_{t}=\frac{\nabla_{t}}{\sqrt{\mathbb{E}_{t}^{\mathbb{P}}\left(\nabla_{t}^{2}\right)}},\quad\text{with}\quad\nabla_{t}=\frac{\partial\ell(X_{t}|\mathcal{G}_{t})}{\partial\log\eta_{t}}.\label{eq:ScaledScore}
\end{equation}
Since $s_{t}$ is $\mathcal{F}_{t}$-measurable, we have $\varepsilon_{t+1}\in\mathcal{F}_{t}\subset\mathcal{G}_{t+1}$
as required by Assumption \ref{assu:ARMA}. Additionally, when the
score is evaluated at the true parameters we have
\[
\mathbb{E}_{t}^{\mathbb{P}}\left(s_{t}\right)=0,\quad\text{and}\qquad\mathrm{var}_{t}^{\mathbb{P}}\left(s_{t}\right)=1,
\]
such that $\mathbb{E}_{t}^{\mathbb{P}}(\varepsilon_{t+1})=0$ and
$\mathrm{var}_{t}^{\mathbb{P}}(\varepsilon_{t+1})=\sigma^{2}$. Note
that this definition of the score does not ensure that $\{s_{t}\}$
is a sequence of iid random variables with zero mean and unit variance
under $\mathbb{Q}$, as needed by Assumption \ref{assu:ARMA}. This
requires additional distributional assumptions about the pricing errors.
For instance, we will make assumptions about the pricing errors that
imply $s_{t}|\mathcal{G}_{t}\sim iid\ N(0,1)$ under both in $\mathbb{P}$
and $\mathbb{Q}$ measures (see Theorem \ref{thm:ScoreProp} below),
in which case the score satisfies all the requirements in Assumption
\ref{assu:ARMA}. From (\ref{eq:ARMA}) and the specifications (\ref{eq:Def-of-e})-(\ref{eq:ScaledScore}),
it follows that our constructed $\eta_{t}$ is $\mathcal{F}_{t-1}$-measurable,
which ensures that the information about current derivative prices
are not used to price themselves.

This approach to modeling $\log\eta_{t}$ is analogous to the way
the conditional variance is modeled in GARCH models. For instance,
the GARCH(1,1) model by \citet{bollerslev:86} implies that $h_{t}=c+bh_{t-1}+a(r_{t-1}^{2}-h_{t-1})$,
such that $h_{t}\sim\operatorname{AR}(1)$ and changes in $h_{t}$
are driven by discrepancies between squared returns and the conditional
variance. The score model invokes a similar self-adjusting property,
where $\log\eta_{t}$ is updated in response to derivative prices
that indicate that pricing errors can be reduced by revising the value
of $\log\eta_{t}$.

In our empirical analysis, we will adopt a score-driven model with
an $\operatorname{AR}(1)$ structure for $\log\eta_{t}$. Theorem
\ref{thm:OptionPricingRandom} can accommodate the case where $\log\eta_{t}\sim\operatorname{ARMA}(p,q)$,
whereas more general models for $\log\eta_{t}$, such as exogenous/endogenous
explanatory variables and long-memory specifications, would require
analogous option pricing results to be established first.

\subsection{The Log-Likelihood Function and the Form of Score\label{subsec:Log-Likelihood-Function}}

Next, we turn to the log-likelihood function used to estimate model
parameters. Our observed data consist of returns, $R_{t}$, and a
vector of derivative prices, $X_{t}$. Without loss of generality,
we can factorize the log-likelihood function as follows:
\[
\ell(R_{1},\ldots,R_{T},X_{1},\ldots,X_{T},\mathcal{F}_{0})=\sum_{t=1}^{T}\ell(R_{t},X_{t}|\mathcal{F}_{t-1})=\sum_{t=1}^{T}\ell(R_{t}|\mathcal{F}_{t-1})+\ell(X_{t}|R_{t},\mathcal{F}_{t-1}).
\]
This decomposition was also used by \citet{ChristoffersenHestonJacobs2013}.
The log-likelihood function for returns, $\ell(R_{t}|\mathcal{F}_{t-1})$,
is that of the HNG model in (\ref{eq:HNGreturn})-(\ref{eq:HNGgarch}),
\[
\ell(R_{t}|\mathcal{F}_{t-1})=-\tfrac{1}{2}\left[\log(2\pi)+\log h_{t}+(R_{t}-r-(\lambda-\tfrac{1}{2})h_{t})^{2}/h_{t}\right],
\]
which assumes that $z_{t}\overset{\mathbb{P}}{\sim}iid\ N(0,1)$.
To complete the model, we need to specify the log-likelihood function
for the vector of derivative prices, $\ell(X_{t}|R_{t},\mathcal{F}_{t-1})=\ell(X_{t}|\mathcal{G}_{t})$.
To this end, we let $X_{t}^{m}\in\mathcal{G}_{t}$ denote the vector
of model-based derivative prices and make the following assumption
for pricing errors.
\begin{assumption}
\label{assu:Pricing-Errors}The derivative pricing errors follow an
iid multivariate normal distribution
\[
e_{t}=X_{t}-X_{t}^{m}\left|\mathcal{G}_{t}\right.\overset{\mathbb{P}}{\sim}iid\ N\left(0,\sigma_{e}^{2}\Omega_{N_{t}}\right),
\]
with correlation matrix $\Omega_{N_{t}}$, where $N_{t}$ is the number
of derivative prices at time $t$. Additionally, $e_{t}$ is independent
of the return shock $z_{\tau}$ for all $\tau$.
\end{assumption}
The vector of derivative ``prices'', $X_{t}$, can be defined in
several ways. \citet{WangShenJiangHuang2017} measured $X_{t}$ in
the unit of volatility (the level of VIX), \citet{ChristoffersenHestonJacobs2013}
used the Vega-weighted option price, and \citet{FeunouOkou2019} used
the Black-Scholes implied volatility. In our empirical analysis, we
use the logarithmically transformed VIX and the logarithmically transformed
Black-Scholes implied volatilities for options, which are natural
choices because we specific a model for $\log\eta_{t}$. Moreover,
the logarithmic transformation reduces the risk of severe model misspecification,
because the distribution of log-volatilities is often well-approximated
by a Gaussian distribution, see \citet{Andersen2003}.

The existing empirical studies that used derivative pricing errors
in this manner, all assumed that the pricing errors were uncorrelated,
i.e. $\Omega_{N_{t}}=I_{N_{t}}$ where $I_{N_{t}}$ is the $N_{t}$-dimensional
identity matrix. This is unrealistic, because contemporaneous pricing
errors tend to be positively correlated. We will therefore allow for
a non-zero common correlation, which is important for the statistical
properties of the score, see Theorem \ref{thm:ScoreProp} below. For
instance, if we impose $\Omega_{N_{t}}=I_{N_{t}}$ in our empirical
analysis with option prices, then it would result in a type of misspecification
that induces an upward bias in the estimated variance of $s_{t}$.
\begin{assumption}
\label{assu:Option-Pricing-Errors}The correlation matrix for pricing
errors has the form: $\Omega_{N_{t}}=(1-\rho)I_{N_{t}}+\rho U_{N_{t}}$,
where $I_{N_{t}}\in\mathbb{R}^{N_{t}\times N_{t}}$ is the identity
matrix and $U_{N_{t}}\in\mathbb{R}^{N_{t}\times N_{t}}$ is a matrix
of ones.
\end{assumption}
The correlation matrix in Assumption \ref{assu:Option-Pricing-Errors},
$\Omega_{N_{t}}$, is known as an equicorrelation matrix. Having a
common correlation for all pairs of pricing errors is particularly
useful in applications where the panel of derivative prices is unbalanced
over time. In this case, the dimension of $\Omega_{N_{t}}$ will be
time-varying, but this is not problematic if an equicorrelation structure
is assumed. A common correlation can be motivated by the assumption
that pricing errors, $e_{i,t}=u_{t}+v_{i,t}$, have a common component,
$u_{t}$, and uncorrelated idiosyncratic components, $v_{i,t}$, $i=1,\ldots,N_{t}$.
This will bring about the structure in Assumption \ref{assu:Option-Pricing-Errors}
with $\rho=\sigma_{u}^{2}/\sigma_{e}^{2}$ and $\sigma_{e}^{2}=\sigma_{u}^{2}+\sigma_{v}^{2}$. 

Under Assumption \ref{assu:Pricing-Errors}, the part of log-likelihood
function that relates to derivative prices is
\[
\ell(X_{t}|\mathcal{G}_{t})=\ell(e_{t}|\mathcal{G}_{t})=-\tfrac{1}{2}\left(N_{t}\log(2\pi\sigma_{e}^{2})+\log\left|\Omega_{N_{t}}\right|+\sigma_{e}^{-2}e_{t}^{\prime}\Omega_{N_{t}}^{-1}e_{t}\right).
\]
The model for returns under $\mathbb{P}$ continues to be a time-homogeneous
Heston-GARCH model, which does not depend on derivative prices. This
is also true in the DHNG model, despite the time variations in the
pricing kernel. For this reason, the model parameters can (if needed)
be estimated by a two-stage estimation method, where the parameters
in the GARCH model is estimated from returns in the first stage, followed
by the remaining parameters being estimated in the second stage with
derivative prices.\footnote{In our empirical application we estimated the model by maximizing
the full log-likelihood function, which was straightforward and did
not cause computational issues. If needed, two-stage estimation can
be adopted to reduce the computational burden, and the approach to
estimation is not uncommon in this setting, see e.g. \citet{BroadieChernovJohannes2007},
\citet{CorsiFusariVecchia2013}, \citet{ChristoffersenHestonJacobs2013},
and \citet{MajewskiBormettiCorsi2015}.}
\begin{thm}
\label{thm:ScoreProp}Suppose that Assumption \ref{assu:Pricing-Errors}
holds. Then the score, (\ref{eq:ScaledScore}), is given by
\begin{align*}
\nabla_{t} & =\frac{1}{\sigma_{e}^{2}}\left(\frac{\partial X_{t}^{m}}{\partial\log\eta_{t}}\right)^{\prime}\Omega_{N_{t}}^{-1}e_{t},
\end{align*}
which has zero conditional mean, $\mathbb{E}_{t}^{\mathbb{P}}\left(\nabla_{t}\right)=0$,
and conditional variance:
\[
\mathbb{E}_{t}^{\mathbb{P}}\left(\nabla_{t}^{2}\right)=\frac{1}{\sigma_{e}^{2}}\left(\frac{\partial X_{t}^{m}}{\partial\log\eta_{t}}\right)^{\prime}\Omega_{N_{t}}^{-1}\left(\frac{\partial X_{t}^{m}}{\partial\log\eta_{t}}\right).
\]
Moreover, the scaled score is $\mathcal{F}_{t}$-measurable and satisfies
$s_{t}|\mathcal{G}_{t}\sim iid\ N(0,1)$, under $\mathbb{P}$ and
$\mathbb{Q}$ measures.
\end{thm}
To obtain the score, we need to derive $\partial X_{t}^{m}/\partial\log\eta_{t}$,
and this must be done separately for the VIX and option prices. These
terms are derived in Appendix \ref{subsec:ScoresDeri}. Note that
if $X_{t}^{m}$ is univariate (the case with a single derivative),
then the scaled score simplifies to 
\[
s_{t}=\frac{1}{\sigma_{e}}{\rm sign}\left(\frac{\partial X_{t}^{m}}{\partial\log\eta_{t}}\right)e_{t}.
\]
From this expression, it is evident that the score will indicate the
direction of change for $\eta_{t}$, which is expected to reduce pricing
errors. 

\section{Empirical Analysis\label{sec:Empirical-Analysis}}

\subsection{Data}

Our empirical analysis is based on daily returns for the S\&P 500
index, the CBOE VIX, and the panel of SPX option prices based on the
S\&P 500 index. Our sample period spans 32 years from January 2nd,
1990, to December 31, 2021, with 8,064 trading days. Returns are defined
from cum-dividend logarithmic transformed closing prices, which were
downloaded from the CRSP of Wharton Research Data Services (WRDS).
Daily VIX prices were obtained from the Chicago Board Options Exchange
(CBOE) website. The SPX option prices were obtained from two sources.
Prices for the first six years (1990--1995) are the so-called Optsum
data, which were purchased from CBOE website. Option prices for the
remaining 26 years are OptionMetrics data downloaded from the WRDS
database.

Option prices were primarily preprocessed following \citet{ChristoffersenFeunouJacobsMeddahi2014}
and \citet{BakshiCaoChen1997}. Specifically, we include out-of-the-money
put and call option prices with positive trading volume and maturities
between two weeks and six months. Options with missing implied volatilities
or prices below one dollar are excluded. Put option prices are converted
to call options using the put-call parity. Very deep out-of-the-money
options with deltas larger than 0.85 or less than 0.15 are discarded.
Much of the existing literature uses weekly data, and a panel of option
prices is typically sampled on Wednesdays, because liquidity tends
to be highest on Wednesdays. However, in our analysis, we use daily
option prices, because a daily score, $s_{t}$, is needed to update
$\eta_{t}$. The number of available options has grown rapidly over
the sample period, both in terms of available maturities and the range
of moneyness at each maturity.\footnote{The number of available option prices has increased almost 50-fold
over our sample period, initially from about 38 daily option prices
to well over 1,700.} From the pool of available options, we select up to six options per
trading day. The inclusion criteria are as follows: we first determine
the most liquid option for each maturity on day $t$, as measured
by daily trading volume. We sort these options by maturity in ascending
order and index these by $j=1,\ldots,\tilde{N}_{t}$, where $\tilde{N}_{t}$
is the number of distinct maturities on day $t$. From this set, we
include all options if $\tilde{N}_{t}\leq6$; the first six options
if $7\leq\tilde{N}_{t}\leq10$; options $\{1,3,5,7,9,11\}$ if $11\leq\tilde{N}_{t}\leq15$;
and options $\{1,4,7,10,13,16\}$ if $16\leq$ $\tilde{N}_{t}\leq20$;
and so forth. This resulting set of options will be representative
for the range of available maturities, and we have $N_{t}=6$ on most
days. The total number of option prices in our full sample period
is 37,152.
\begin{table}
\begin{centering}
\caption{Summary Statistics\label{Tab:Descriptive}}
\par\end{centering}
\begin{centering}
\vspace{0.2cm}
\begin{footnotesize}
\begin{tabularx}{\textwidth}{XYYYYYYY}
\toprule
\midrule
         \\
    \multicolumn{7}{l}{{\it A: S\&P 500 returns and the CBOE VIX}} \\[6pt]
       &   & Mean(\%) & Std(\%) & Skewness & Kurtosis & Obs. \\[3pt]
    \multicolumn{2}{l}{Returns (annualized)} & 8.13 & 18.11 &	-0.41	& 14.37	& 8,064 \\
    VIX & & 19.48 & 8.01  & 2.21  & 11.49 & 8,064 \\
          &       &       &       &       &  \\
    \\
    \multicolumn{7}{l}{{\it B: Option Price Data}} \\[6pt]
        &  &    \multicolumn{2}{>{\hsize=\dimexpr2\hsize+2\tabcolsep+\arrayrulewidth\relax}c}
        {Implied Volatility (\%) }
        &            \multicolumn{2}{>{\hsize=\dimexpr2\hsize+2\tabcolsep+\arrayrulewidth\relax}c}
        {Average price (\$) }
      & Observations \\[3pt]
    All options & & \multicolumn{2}{c}{18.47}& \multicolumn{2}{c}{63.63}  & 37,152 \\
         & &       &       &       &       &  \\
    \multicolumn{7}{l}{\it Partitioned by Moneyness } \\[2pt]
    Delta<0.3          &  & \multicolumn{2}{c}{14.18} & \multicolumn{2}{c}{11.04  } & 5,452 \\
    0.3$\leq$Delta<0.4 &  & \multicolumn{2}{c}{15.52} & \multicolumn{2}{c}{20.92 } & 2,955 \\
    0.4$\leq$Delta<0.5 &  & \multicolumn{2}{c}{16.76} & \multicolumn{2}{c}{33.92 } & 3,912 \\
    0.5$\leq$Delta<0.6 &  & \multicolumn{2}{c}{18.94} & \multicolumn{2}{c}{48.78 } & 7,091 \\
    0.6$\leq$Delta<0.7 &  & \multicolumn{2}{c}{19.40} & \multicolumn{2}{c}{64.86 } & 6,098 \\
    0.7$\leq$Delta     &  & \multicolumn{2}{c}{21.02} & \multicolumn{2}{c}{117.47} & 11,644 \\
&          &       &       &       &       &  \\
    \multicolumn{7}{l}{\it Partitioned by Maturity} \\[2pt]
     DTM<30          & & \multicolumn{2}{c}{16.89} & \multicolumn{2}{c}{41.41 } & 9,614 \\
     30$\leq$DTM<60  & & \multicolumn{2}{c}{17.91} & \multicolumn{2}{c}{52.29 } & 9,696 \\
     60$\leq$DTM<90  & & \multicolumn{2}{c}{19.00} & \multicolumn{2}{c}{65.63 } & 7,465 \\
    90$\leq$DTM<120  & & \multicolumn{2}{c}{20.25} & \multicolumn{2}{c}{87.43 } & 4,421 \\
    120$\leq$DTM<150 & & \multicolumn{2}{c}{20.22} & \multicolumn{2}{c}{98.51}  & 2,931 \\
     150$\leq$DTM    & & \multicolumn{2}{c}{19.68} & \multicolumn{2}{c}{97.06}  & 3,025 \\
          &       & &      &       &       &  \\
    \multicolumn{7}{l}{\it Partitioned by the level of VIX} \\[2pt] 
     VIX<15         & & \multicolumn{2}{c}{12.15} & \multicolumn{2}{c}{44.77}  & 12,638 \\
     15$\leq$VIX<20 & & \multicolumn{2}{c}{16.70} & \multicolumn{2}{c}{62.49}  & 10,788 \\
     20$\leq$VIX<25 & & \multicolumn{2}{c}{21.45} & \multicolumn{2}{c}{74.72}  & 7,143 \\
     25$\leq$VIX<30 & & \multicolumn{2}{c}{25.43} & \multicolumn{2}{c}{84.16}  & 3,395 \\
     30$\leq$VIX<35 & & \multicolumn{2}{c}{29.61} & \multicolumn{2}{c}{88.12}  & 1,456 \\
     35$\leq$VIX    & & \multicolumn{2}{c}{40.23} & \multicolumn{2}{c}{101.70}  & 1,732 \\
\\[0.0cm]
\\[-0.5cm]
\midrule
\bottomrule
\end{tabularx}
\end{footnotesize}

\par\end{centering}
{\small Note: Summary statistics for close-to-close S\&P 500 index
log returns, CBOE VIX, and SPX option prices from January 1990 to
December 2021. We report the sample mean (Mean), standard deviation
(Std), skewness (Skew), kurtosis (Kurt), number of observations (Obs),
for returns and the VIX. Option prices are based on closing prices
of out-of-the-money call and put options. We report the average Black-Scholes
implied volatility (IV), average price, and the number of option prices
for different ranges of Moneyness, days to maturity (DTM), and VIX
levels. Data sources: S\&P 500 returns from CRSP, VIX from CBOE's
website. Option prices from Optsum data (1990--1995) and OptionMetrics
(1996--2021).}{\small\par}
\end{table}

Table \ref{Tab:Descriptive} contains descriptive statistics for our
S\&P 500 returns and the CBOE VIX in Panel A, and option prices in
Panel B. As expected, the sample average of the VIX (19.48\%) is larger
than the standard deviation for annualized returns (18.11\%). This
difference reflects the (average) negative volatility risk premium.
The S\&P 500 returns exhibit slight negative skewness and a high degree
of kurtosis, whereas the VIX has positive skewness and slightly lower
kurtosis. We also report summary statistics for option prices partitioned
by moneyness, maturity, and the contemporaneous level of VIX. Options
with deltas below 0.5 are out-of-the-money call options, and options
with deltas above 0.5 are based on out-of-the-money put options. Deep
out-of-the-money put options (i.e., deltas greater than 0.7) are expensive
relative to out-of-the-money call options, reflecting the well-known
volatility smirk. The implied volatility has a relatively flat term
structure with respect to time to maturity. Unsurprisingly, the implied
volatility increases with the VIX level, as shown at the bottom of
Table \ref{Tab:Descriptive}, where option prices are partitioned
by the contemporaneous level of the VIX.

\subsection{Parameter estimation}

We estimate the model with both constant and time-varying variance
risk aversion, CHNG and DHNG, respectively.\footnote{The corresponding results for the Auxiliary structure are presented
in the Appendix \ref{subsec:ResultsAuxi}.} Both specifications are estimated using two types of derivative prices,
VIX data and the panel of option prices. Parameters are estimated
by maximizing the joint log-likelihood function for the full sample
period from January 1990 to December 2021. Each column reports the
parameter estimates for the specification listed in the first row
of Table \ref{tab:JointEstimation}. The type of derivatives, VIX
or option prices, used in the estimation is indicated with {[}VIX{]}
and {[}Opt{]}, respectively. Robust standard errors are give in parentheses
below the estimates.\footnote{Following \citet{ChristoffersenHestonJacobs2013}, we impose $\omega$
= 0 in estimation when the non-negativity constraint, which ensures
positive variances, is binding.} We also report the implied persistence of volatility under both $\mathbb{P}$
and $\mathbb{Q}$. These are given by $\pi^{\ensuremath{\mathbb{P}}}=\beta+\alpha\gamma^{2}$
and $\pi^{\ensuremath{\mathbb{Q}}}=\mathbb{E}^{\mathbb{Q}}[\beta_{t}^{*}+\alpha_{t}^{*}\gamma_{t-1}^{*2}]$,
respectively, where the latter simplifies to $\pi^{\ensuremath{\mathbb{Q}}}=\beta^{*}+\alpha^{*}\gamma^{*2}$
for the CHNG model with constant parameters. We also report the different
terms of the log-likelihood (for returns and different sets of derivative
prices). Some of these likelihood terms can be interpreted as pseudo
out-of-sample log-likelihood values, indicated in italic. For instance,
the CHNG{[}VIX{]} model is estimated using returns and the VIX, but
the estimated model also yields model-based option prices that can
be compared with actual option prices. The reported pseudo log-likelihood
for option prices is evaluated with the resulting option pricing errors
that are implicitly used to obtain estimates of $\rho$ and $\sigma_{e}$
for option prices. Similarly, we evaluate the log-likelihood for VIX
pricing errors for the specifications estimated with option prices.
In this case, we compute the implied estimate of $\sigma_{e}$ for
VIX prices. The largest log-likelihood within each row is highlighted
in bold. 
\begin{table}
\caption{Joint Estimation Results}

\begin{centering}
\vspace{0.2cm}
\begin{footnotesize}
\begin{tabularx}{\textwidth}{p{2.5cm}YYYY}
\toprule
\midrule
    Model & CHNG   & CHNG     & DHNG  & DHNG   \\
          & [VIX] & [Opt]  & [VIX] & [Opt]  \\
    \midrule 
          &       &       &       &       \\
    $\lambda$ & 3.377 & 2.828 & 3.398 & 3.088 \\
          & \textit{(0.662)} & \textit{(0.022)} & \textit{(1.281)} & \textit{(0.644)} \\
          &       &       &       &  \\
    $\beta$ & 0.902 & 0.713 & 0.720 & 0.570 \\
          & \textit{(0.001)} & \textit{(0.008)} & \textit{(0.018)} & \textit{(0.015)} \\
          &       &       &       &  \\
    $\alpha(\times10^{-6})$ & 1.166 & 2.495 & 4.428 & 5.833 \\
          & \textit{(0.206)} & \textit{(0.012)} & \textit{(0.416)} & \textit{(0.256)} \\
          &       &       &       &  \\
    $\gamma$ & 272.45 & 327.49 & 220.45 & 249.59 \\
          & \textit{(23.81)} & \textit{(3.81)} & \textit{(9.474)} & \textit{(10.28)} \\
          &       &       &       &  \\
    $\zeta$ &       &       & 0.192 & 0.102 \\
          &       &       & \textit{(0.057)} & \textit{(0.011)} \\
          &       &       &       &  \\
    $\varphi$ &       &       & 0.988 & 0.994 \\
          &       &       & \textit{(0.002)} & \textit{(0.001)} \\
          &       &       &       &  \\
    $\sigma$ &       &       & 0.061 & 0.042 \\
          &       &       & \textit{(0.002)} & \textit{(0.010)} \\
          &       &       &       &  \\
    $\rho$ &       & 0.822 &       & 0.087 \\
          &       & \textit{(0.015)} &       & \textit{(0.010)} \\
          &       &       &       &  \\
    $\sigma_e$ & 0.175 & 0.229 & 0.043 & 0.091 \\
          & \textit{(0.071)} & \textit{(0.024)} & \textit{(0.001)} & \textit{(0.006)} \\
          &       &       &       &  \\
    $\eta, \mathbb{E}{\eta}$ & 1.371 & 1.252 & 1.314 & 1.190 \\
          & \textit{(0.041)} & \textit{(0.013)} &       &  \\
          &       &       &       &  \\
    $\widehat{\mathrm{var}}(s_{t})$ &  &  & 1.000 & 1.007 \\
        &       &       &       &  \\
    $\pi^\mathbb{P}$ & 0.989 & 0.981 & 0.935 & 0.934 \\
    $\pi^\mathbb{Q}$ & 0.991 & 0.985 & 0.944 & 0.944 \\

          &       &       &       &  \\
    LogL  &       &       &       &  \\
    $\ell(\rm{R})$ & 26,324 & 26,438 & \textbf{26,450} & 26,419 \\
    $\ell(\rm{VIX})$ & 2,619  & \textit{1,602}  & \textbf{13,917} & \textit{7,932} \\
    $\ell(\rm{Opt})$ & \textit{15,177} & 21,647 & \textit{27,857} & \textbf{36,682} \\
    $\ell(\rm{R,VIX})$ & 28,944 & \textit{28,041} & \textbf{40,367} & \textit{34,351} \\
    $\ell(\rm{R,Opt})$ & \textit{41,501} & 48,085 & \textit{54,307} & \textbf{63,101} \\    
\\[0.0cm]
\\[-0.5cm]
\midrule
\bottomrule
\end{tabularx}
\end{footnotesize}
\par\end{centering}
{\small Note: Estimation results for the full sample period, January
1990 to December 2021. CHNG is the model by }\citet{ChristoffersenHestonJacobs2013}{\small{}
and DHNG is the model introduced in this paper. The type of derivatives
used to estimate each model, VIX or option prices, is indicated with
{[}VIX{]} and {[}Opt{]}, respectively. Estimates are reported with
robust standard errors (in parentheses), and $\pi^{\mathbb{P}}$ and
$\pi^{\mathbb{Q}}$ refer to the persistence of volatility under $\mathbb{P}$
and $\mathbb{Q}$, respectively. The components of the maximized log-likelihood
function are reported at the bottom of the table, and bold is used
to highlight the largest log-likelihood in each row. Italic font is
used to highlight pseudo out-of-sample log-likelihoods.\label{tab:JointEstimation}}{\small\par}
\end{table}

There are several interesting observations to be made from Table \ref{tab:JointEstimation}.
First, the volatility process is found to be highly persistent across
all specifications, and the persistence is larger under $\mathbb{Q}$
than under $\mathbb{P}$, which is consistent with the existing literature.
Second, the estimate of the equity risk premium, $\lambda$, is positive
and significant in all specification. Third, the VRR, $\eta$, is
significantly larger than one in both CHNG specifications. Similarly,
for the DHNG model the expected value of $\log\eta_{t}$, $\zeta$,
is significantly positive in both DHNG specifications. Their corresponding
expected $\eta_{t}$ (denoted $\mathbb{E}\eta$) inferred from the
AR(1) model for $\log\eta_{t}$ are both notably larger than one.
This implies that the risk-neutral volatility, $h^{\ast}$, is (on
average) larger than the physical volatility, $h$. Fourth, the parameter
that defines the leverage effects, $\gamma$, is also found to be
significant across all specification. Fifth, for the DHNG specification
we note that, $\eta_{t}$, is highly persistent, as $\varphi$ is
estimated to be close to one. In fact, it is estimated to be more
persistent than $h_{t}$ and $h_{t}^{\ast}$ across all DHNG specifications.
Sixth, the coefficient, $\sigma$, which measures the impact of $s_{t-1}$
on $\eta_{t}$, is estimated to be positive and significant.\footnote{In principle, $\sigma$ could be negative, but $\operatorname{var}(\sigma s_{t})=\sigma^{2}$
is non-negative.} The implied unconditional variance for $\eta_{t}$ is similar for
the two DHNG specifications, $\mathrm{var}(\eta_{t})=0.29$ when estimated
with VIX and $\mathrm{var}(\eta_{t})=0.23$ when estimated with option
prices.\footnote{The variance for $\eta_{t}$ is computed from AR(1) parameters for
$\log\eta_{t}$ and the assumptions imply that the unconditional distribution
for $\eta_{t}$ is log-normal.} Seventh, all specifications yield similar likelihood values for the
returns. This is to be expected because they all rely on the same
Heston-Nandi GARCH model for returns under physical measure. 

Eighth, the key difference between the CHNG and DHNG models is the
enhanced flexibility in DHNG's pricing kernel. This generalization
leads to large improvements in the likelihood for derivative prices.
The reason is that the adaptive pricing kernel with time-varying variance
risk aversion yields model-implied derivative prices that are much
closer to observed prices, and the substantial reduction in the pricing
errors translates to much higher values of the log-likelihood for
derivative prices. Ninth, estimating the models with VIX or option
prices results in some differences across the estimated parameters.
This is to be expected, as the vector of option prices contains more
information about the distribution of future returns than the VIX.
For instance, the leverage parameter, $\gamma$, is estimated to be
larger with option prices than with the VIX. Tenth, the correlation
among option pricing errors, $\rho$, is estimated to be positive
and significant for both models. The estimate has the staggering large
value of $82.2\%$ for the CHNG model and a more moderate value of
$8.7\%$ for DHNG model. Ignoring this correlation (assuming it to
be zero) would greatly underestimate the condition variance of $\text{\ensuremath{\nabla_{t}}}$,
which makes the variance of $s_{t}$ larger than one.
\begin{sidewaysfigure}
\centering{}\includegraphics[viewport=1050bp 100bp 10225bp 5150bp,width=1\textwidth]{figures/Figure5.jpg}\caption{{\small This figure displays the full-sample time series of the daily
estimated variance risk ratio, $\eta_{t}$, from January 1990 to December
2021. The blue line is based on the VIX and the red line on option
prices. The two time series are very similar, albeit the $\eta_{t}$
based on the VIX tend to be slightly larger than that based on option
prices.\label{fig:eta}}\protect \\
{\small}\protect \\
{\small Alt text: Time series of the variance risk ratio from 1990
to 2021, estimated using VIX data and option prices. The two series
follow similar trends, with occasional level differences, particularly
during financial crises and market stress periods.}}
\end{sidewaysfigure}

Figure \ref{fig:eta} presents the estimated daily time series of
the VRR, $\eta_{t}=h_{t+1}^{\ast}/h_{t+1}$, based on the VIX (blue
line) and option prices (red line). The two have a high degree of
comovement and, as we discussed earlier, the unconditional variance
of the two series is very similar. Importantly, the VRR occasionally
falls below one, as is seen during the years 1993--1995, 2004--2007,
and around 2017. A value below one ($\eta_{t}<1$) indicates that
investors have an increased appetite for variance risk, whereas a
large value of $\eta$ implies that investors demand additional compensation
for variance risk. The latter was particularly the case during financial
crises, such as the immediate aftermath of the Lehman collapse and
the recent COVID-19 pandemic. The observed conditional variance was
unusually high during these episodes, and the large value of $\eta_{t}$
implies that $h_{t}^{\ast}$ increased to far higher levels and was,
briefly, more than three times larger than $h_{t}$. 

The level of $\eta_{t}$ based on the VIX tend to be slightly larger
than that based on option prices. This discrepancy can be attributed
to the fact that the VIX is based on options with 30 days to maturity,
whereas the options have maturities up to 180 days. A possible explanation
is that short-term investors demand larger compensation for variance
risk than long-term investors, which would be consistent with the
findings in \citet{EisenbachSchmalz2016} and \citet{AndriesEisenbachSchmalz2018}.
\begin{table}
\caption{VIX and Option Pricing Performance}

\begin{centering}
\vspace{0.2cm}
\begin{footnotesize}%
\begin{tabularx}{\textwidth}{p{2.5cm}YYYYYYYY}
\toprule 
\midrule
    Model & CHNG   & CHNG      & DHNG  & DHNG   \\
      & [VIX] & [Opt]  & [VIX] & [Opt]  \\
    \midrule
          &       &       &       &        \\
    \multicolumn{5}{l}{\it A: RMSE for VIX Pricing}\\[6pt]
    Full Sample   & 17.49 & 19.81 & 4.308 & 9.049  \\
          &       &       &       &         \\
    \multicolumn{5}{l}{\it B: RMSE for Option Pricing} \\[6pt]
    Full Sample   & 23.06 & 22.88 & 13.34 & 9.191 \\
          &       &       &       &         \\
    \multicolumn{5}{l}{\it Partitioned by moneyness } \\[3pt]
   \ Delta<0.3      & 33.05 & 29.80 & 21.11 & 12.77 \\
   \ 0.3$\leq$Delta<0.4      & 27.84 & 26.59 & 17.86 & 11.28 \\
   \ 0.4$\leq$Delta<0.5     & 23.29 & 23.87 & 14.68 & 9.375 \\
   \ 0.5$\leq$Delta<0.6      & 19.74 & 21.49 & 11.41 & 7.271 \\
   \ 0.6$\leq$Delta<0.7      & 18.83 & 20.62 & 9.72  & 6.859 \\
   \ 0.7$\leq$Delta    & 19.50 & 19.48 & 8.78  & 8.634 \\
           &       &       &       &        \\
    \multicolumn{5}{l}{\it Partitioned by maturity} \\[3pt]
    \ DTM<30 & 25.48 & 23.51 & 13.72 & 12.19 \\
    \ 30$\leq$DTM<60 & 23.48 & 23.27 & 12.80 & 7.848 \\
    \ 60$\leq$DTM<90 & 21.78 & 22.30 & 12.75 & 6.914 \\
   \ 90$\leq$DTM<120 & 20.47 & 21.81 & 13.51 & 7.963 \\
   \ 120$\leq$DTM<150 & 20.74 & 21.68 & 13.84 & 8.622 \\
    \ 150$\leq$DTM & 22.18 & 23.58 & 14.45 & 9.226 \\
          &       &       &       &       \\
    \multicolumn{5}{l}{\it Partitioned by the level of VIX} \\[3pt]
    \ VIX<15     & 29.25 & 28.06 & 15.60 & 10.02 \\
    \ 15$\leq$VIX<20     & 18.47 & 19.72 & 13.16 & 9.021 \\
    \ 20$\leq$VIX<25     & 18.23 & 19.51 & 10.89 & 8.224 \\
    \ 25$\leq$VIX<30     & 18.74 & 18.52 & 11.02 & 8.518 \\
    \ 30$\leq$VIX<35    & 19.48 & 19.23 & 10.97 & 8.542 \\
    \ 35$\leq$VIX     & 25.36 & 22.41 & 11.57 & 9.508 \\
\\[0.0cm]
\\[-0.5cm]
\midrule
\bottomrule
\end{tabularx}
\end{footnotesize}
\par\end{centering}
{\small Note: This table presents the full-sample VIX and option pricing
performance from January 1990 to December 2021 for each model listed
in \ref{tab:JointEstimation}. We evaluate the model's VIX and option
pricing ability through the root mean square error for $\log{\rm VIX}$
and log implied volatility ($\log{\rm IV}$), respectively, see (\ref{eq:RMSE_VIX_log})
and (\ref{eq:RMSE_IV_log}). Results for option pricing are presented
for different ranges of moneyness (Black-Scholes Delta), days-to-maturity
(DTM), and the VIX levels.\label{tab:DerivativePricing}}{\small\par}
\end{table}

\subsection{In-Sample VIX and Option Pricing Performance}

We evaluate the performance of the models in terms of their ability
to accurately price the VIX and options. We report the root mean squared
errors (RMSE) for the logarithm of VIX and the logarithmically transformed
implied volatilities.\footnote{As a robustness check, we also computed the RMSE for the level of
VIX and implied volatilities. These results are reported in Appendix
\ref{subsec:RMSEIV}.} The RMSE multiplied by 100 can be interpreted as the (absolute) relative
pricing errors in percent. This loss function is coherent with our
log-likelihood function, where the term that involves derivative prices
is given in Assumption \ref{assu:Pricing-Errors}. A similar loss
function was adopted in \citet{Ornthanalai2014}, who used the relative
implied volatility to prevent days with high volatility from being
given a disproportionately large weight in the comparisons. For each
of the specifications listed in Table \ref{tab:JointEstimation},
we report their (in-sample) RMSEs in Table \ref{tab:DerivativePricing}.

We present the in-sample RMSEs for the VIX in Panel A. These are defined
by 
\begin{equation}
{\rm RMSE_{VIX}}=\sqrt{\frac{1}{T}\sum_{t=1}^{T}\text{\ensuremath{\left[\log\mathrm{VIX}_{t}^{m}-\log\mathrm{VIX}_{t}\right]^{2}}}}\times100,\label{eq:RMSE_VIX_log}
\end{equation}
where $\mathrm{VIX}_{t}^{m}$ is the model-based quantity and $\mathrm{VIX}_{t}$
is the observed market VIX value. Panel A of Table \ref{tab:DerivativePricing}
shows that a time-varying VRR, $\eta_{t}$, results in substantially
smaller average pricing errors, relative to a constant $\eta$, which
is a characteristic of CHNG. The improved VIX pricing is substantial
and impressive. The most accurate VIX pricing is achieved with by
the DHNG model, when it is estimated with VIX data. This reduces the
RMSE by a factor of four relative to both CHNG specifications. Even
the DHNG model that is estimated solely with option prices, has a
much smaller RMSE for VIX prices than the CHNG specification that
uses VIX prices as part of the objective in the estimation. This is
impressive, because the DHNG model estimated with option prices is
based on an objective that simultaneously seeks to price options with
maturities ranging from one to six months. This entails a trade-off
across maturities, unlike models estimated with solely with VIX, that
specifically targets the one-month maturity.

Panel B of Table \ref{tab:DerivativePricing} reports the in-sample
performance for option pricing. We follow the literature and convert
option prices to their corresponding Black-Scholes implied volatilities.
We compare the model-based implied volatility, $\mathrm{IV}^{m}$,
with the market-based implied volatility, $\mathrm{IV}$, where the
latter is defined by the observed option price. The resulting RMSE,
\begin{equation}
{\rm RMSE_{IV}}=\sqrt{\frac{1}{\sum_{t=1}^{T}N_{t}}\sum_{t=1}^{T}\text{\ensuremath{\sum_{i=1}^{N_{t}}\left[\log\mathrm{IV}_{t,i}^{m}-\log\mathrm{IV}_{t,i}\right]^{2}}}}\times100,\label{eq:RMSE_IV_log}
\end{equation}
is reported in the first row of Panel B. We also compute the ${\rm RMSE_{IV}}$
for options partitioned by moneyness, time to maturity, and the contemporaneous
VIX level. The resulting RMSEs can be used to identify shortcomings
in a model, such as its inability to generate sufficient leverage
effect, capture the dynamic properties, and generate a proper variance
risk premium.

The DHNG specifications clearly dominate the CHNG specifications in
terms of option pricing. The average option pricing errors are substantially
smaller for DHNG models than for CHNG models, and the smallest RMSE
is obtained by the DHNG model estimated with option prices. Its RMSE
is less than half that of any of the CHNG specifications. Impressively,
the DHNG specifications uniformly dominate all CHNG specifications
for both VIX and option pricing across all partitions. This further
supports the advantage of having a dynamic variance risk aversion
in the model. Note that the largest gains in pricing accuracy are
found during periods with low volatility and for out-of-the-money
call options, which tend to be options with low variance risk premia.
\begin{figure}
\centering{}\includegraphics[width=1\textwidth]{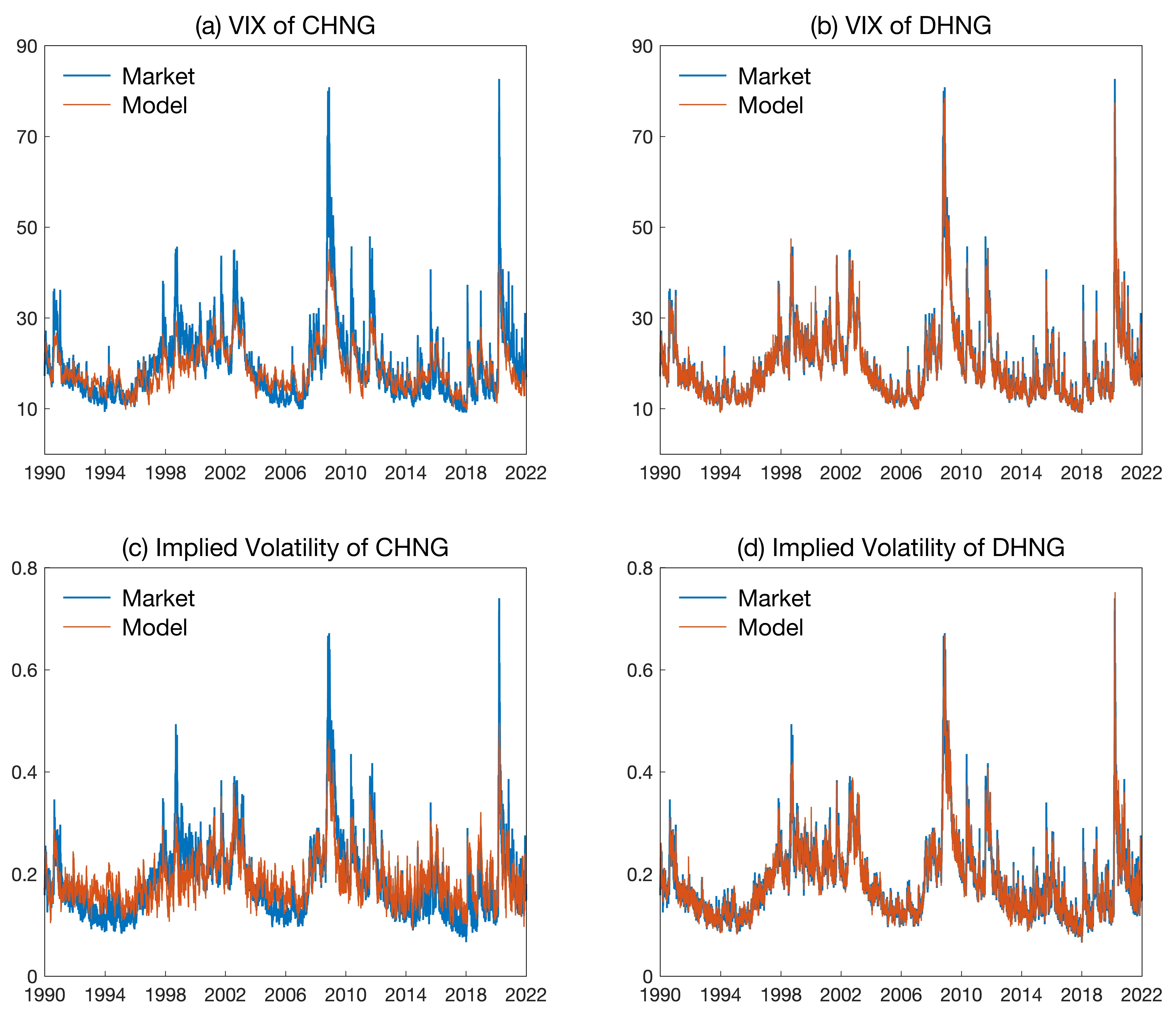}\caption{{\small This figure displays the full-sample time series of the daily
model-based (red lines) and market-based (blue lines) VIX and average
implied volatility. The CHNG and DHNG models are estimated with VIX
data in the upper panels and option prices in the lower panels. Model
parameters are based on the Table \ref{tab:JointEstimation}.\label{fig:Model_VIX_IV}}\protect \\
{\small}\protect \\
{\small Alt text: Comparison of model-based and market-based values
for the VIX and average implied volatility over time. The figure presents
results for both the CHNG and DHNG models, estimated using VIX data
in the upper panels and option prices in the lower panels. DHNG provides
a closer match to market values.}}
\end{figure}

Figure \ref{fig:Model_VIX_IV} displays the model-based VIX and implied
volatility (red lines) alongside the corresponding market-based quantities
(blue lines). The upper panels show the model-based VIX for CHNG and
DHNG, with both models estimated using VIX data. The lower panels
present the daily average of model-based implied volatilities for
CHNG and DHNG, estimated with option prices. For CHNG, we see larger
discrepancies between model-based and market-based quantities as evident
in Figure \ref{fig:Model_VIX_IV}. Specifically, CHNG underprices
when volatility is high and overprices when volatility is low, whereas
no such patterns are observed for DHNG.

Interestingly, \citet{ChristoffersenHestonJacobs2013} also estimated
an ad-hoc model where the ratio of volatility under $\mathbb{P}$
and $\mathbb{Q}$ is not held constant. This model was used as a benchmark
for the CHNG model. The ad-hoc model consisted of two separate HNGs:
one fitted to returns under $\mathbb{P}$, and one that was estimated
with option prices under $\mathbb{Q}$. This model is referred to
as ad-hoc, because it makes no attempt to uncover a pricing kernel
that explains the differences between $\mathbb{P}$ and $\mathbb{Q}$.
Their ad-hoc model did not improve option pricing accuracy relative
to CHNG. Theorem \ref{thm:DynUnderQ} provides a theoretical explanation
for the poor performance of their ad-hoc model, as it shows that any
time variation in $\xi_{t}$ (or, equivalently, in $\eta_{t})$ makes
it impossible to have HNGs with time-invariant parameters under both
$\mathbb{P}$ and $\mathbb{Q}$ measures.\footnote{The ad-hoc model is therefore internally inconsistent. Additionally,
the ad-hoc model only leads to a minuscule improvement in the empirical
fit, as the log-likelihood for option prices improves by only about
0.1 units.} 

In the next section, we evaluate if the substantial in-sample improvements
in derivatives pricing also holds out-of-sample. 
\begin{table}
\caption{Out-of-sample VIX and Option Pricing}

\begin{centering}
\vspace{0.2cm}
\begin{footnotesize}%
\begin{tabularx}{\textwidth}{p{2.5cm}YYYYYYYY}
\toprule 
\midrule
    Model & CHNG   & CHNG      & DHNG  & DHNG   \\
      & [VIX] & [Opt]  & [VIX] & [Opt]  \\
    \midrule
          &       &       &       &        \\
    \multicolumn{5}{l}{\it A: RMSE for VIX Pricing}\\[6pt]
    Out-of-Sample   & 18.33 & 19.60 & 4.782 & 10.03  \\
          &       &       &       &         \\
    \multicolumn{5}{l}{\it B: RMSE for Option Pricing} \\[6pt]
    Out-of-Sample   & 24.33 & 20.98 & 15.52 & 9.815 \\
          &       &       &       &         \\
    \multicolumn{5}{l}{\it Partitioned by moneyness } \\[3pt]
   \ Delta<0.3      & 36.55 & 29.30 & 25.22 & 13.58 \\
   \ 0.3$\leq$Delta<0.4      & 31.16 & 25.64 & 21.19 & 11.66 \\
   \ 0.4$\leq$Delta<0.5     & 25.81 & 22.17 & 17.05 & 9.540 \\
   \ 0.5$\leq$Delta<0.6      & 20.56 & 19.05 & 12.62 & 7.145 \\
   \ 0.6$\leq$Delta<0.7      & 18.51 & 17.39 & 9.373 & 6.883 \\
   \ 0.7$\leq$Delta    & 16.81 & 15.81 & 8.670 & 9.381 \\
           &       &       &       &        \\
    \multicolumn{5}{l}{\it Partitioned by maturity} \\[3pt]
    \ DTM<30 & 28.51 & 22.73 & 16.86 & 11.92 \\
    \ 30$\leq$DTM<60 & 25.27 & 21.61 & 15.64 & 8.372 \\
    \ 60$\leq$DTM<90 & 20.23 & 18.80 & 14.07 & 7.737 \\
   \ 90$\leq$DTM<120 & 19.73 & 19.22 & 14.34 & 8.809 \\
   \ 120$\leq$DTM<150 & 18.02 & 17.64 & 14.11 & 9.539 \\
   \  150$\leq$DTM & 21.33 & 21.02 & 15.27 & 10.62 \\
          &       &       &       &       \\
    \multicolumn{5}{l}{\it Partitioned by the level of VIX} \\[3pt]
    \ VIX<15     & 32.40 & 24.86 & 17.71 & 10.82 \\
    \ 15$\leq$VIX<20     & 19.31 & 18.82 & 15.62 & 9.580 \\
    \ 20$\leq$VIX<25     & 14.41 & 18.11 & 12.99 & 8.233 \\
    \ 25$\leq$VIX<30     & 16.27 & 17.42 & 13.20 & 9.281 \\
    \ 30$\leq$VIX<35    & 18.40 & 17.24 & 13.32 & 9.569 \\
    \ 35$\leq$VIX     & 27.69 & 21.02 & 12.43 & 9.924 \\
\\[0.0cm]
\\[-0.5cm]
\midrule
\bottomrule
\end{tabularx}
\end{footnotesize}
\par\end{centering}
{\small Note: Each model is estimated once, using data from the years
1990--2007 (in-sample). The pricing accuracy of each estimated model
is evaluated out-of-sample using the remaining 14 years of data, 2008--2021,
in terms of the RMSE for log prices. We also present results for subsets
of options, partitioning them by moneyness (Black-Scholes Delta),
days-to-maturity (DTM), and the VIX levels.\label{tab:OutOfSample}}{\small\par}
\end{table}

\subsection{Out-of-Sample Pricing Performance}

In this section, we shift our focus to out-of-sample comparisons.
Since the DHNG model nests the CHNG model as a special case, it naturally
achieves a higher in-sample log-likelihood. Out-of-sample comparisons
are crucial to determine whether the improved in-sample performance
is due to overfitting or reflects the true quality of the DHNG model.
We investigate whether the DHNG model also produces more accurate
derivative prices in out-of-sample comparisons. The full dataset is
divided into two subsamples: data from the years 1990--2007 (in-sample)
are used to estimate each of the four specifications, while data from
the period 2008--2021 (out-of-sample) are used to evaluate the estimated
models. Similar to the in-sample comparisons, we assess and compare
the models based on their root mean squared pricing errors.

We find that the out-of-sample performance of the DHNG model is just
as impressive as its in-sample performance. The most accurate VIX
pricing is once again achieved with the DHNG model estimated using
VIX data, while the best option pricing is similarly achieved by the
DHNG model estimated with option prices. Although these two specifications
have similar point estimates and produce very similar paths for $\eta_{t}$,
the small differences between the two estimated model do affect derivative
pricing. Overall, the out-of-sample results are very encouraging and
align closely with the in-sample findings. 

Next, we turn our attention to the autocorrelations of the pricing
errors.

\subsection{Key Insight from Autocorrelations of Pricing Error}

The autocorrelations of pricing errors offer valuable insights into
the improved derivative pricing achieved by the DHNG model. The estimated
DHNG models reveal substantial time variation in $\eta_{t}$. In contrast,
the CHNG model, which relies on a constant $\eta$, tends to produce
positive pricing errors when $\eta_{t}$ is small and negative pricing
errors when $\eta_{t}$ is large. Given the high persistence of $\eta_{t}$,
this is expected to induce autocorrelation in the pricing errors for
the CHNG model, which is indeed what we observe.

Figure \ref{fig:ACF-pricing-errors} displays the autocorrelation
functions (ACF) for VIX pricing errors in the upper panels and ACFs
for average option pricing errors in the lower panels, covering the
full sample period. The very high and persistent ACFs for CHNG indicate
that a constant $\eta$ induces a high degree of predictability in
the pricing errors. In stark contrast, the results for DHNG, shown
in the right panels, reveal autocorrelations that are substantially
closer to zero. The horizontal lines in each plot represent two standard
deviations from zero. For DHNG, the autocorrelations are largely insignificant,
with the exception of the first few autocorrelations in the lower-right
panel.

The remarkable reduction in these autocorrelations is a clear demonstration
of the adaptive nature of score-driven models. The DHNG model continuously
adjusts $\eta_{t}$ to minimize pricing errors, using the latest first-order
conditions (and curvature) to determine the direction and magnitude
of each adjustment. As a result, the first-order conditions are more
consistently satisfied throughout the sample. In contrast, models
with static parameters, such as CHNG, focus only on the average first-order
condition. This approach allows for large pricing errors, as long
as the positive errors are offset by negative errors. This explains
why the CHNG model exhibits much larger RMSE and significant autocorrelations
in its pricing errors.
\begin{figure}
\centering{}\includegraphics[width=1\textwidth]{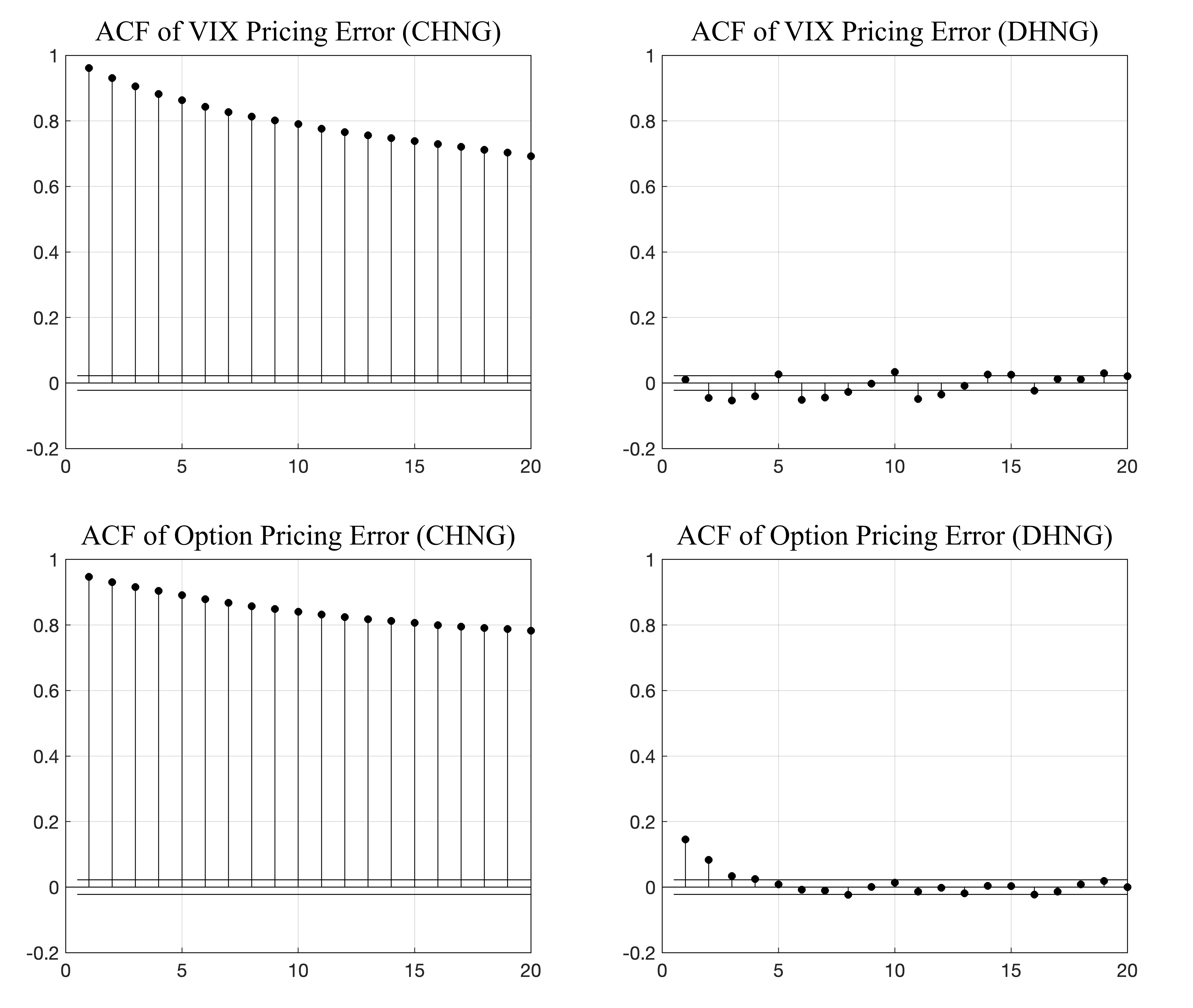}\caption{{\small The autocorrelation function (ACF) for derivative pricing errors
over the full sample period (January 1990 to December 2021). The ACFs
for the VIX are shown in the upper panels, and the ACFs for the average
option pricing errors are shown in the lower panels. The CHNG and
DHNG models are estimated with VIX data in the upper panels and option
prices in the lower panels. Model parameters are based on the Table
\ref{tab:JointEstimation}.\label{fig:ACF-pricing-errors}}\protect \\
{\small}\protect \\
{\small Alt text: Autocorrelation functions (ACFs) of derivative pricing
errors over the full sample period (1990--2021) for the CHNG and
DHNG models. DHNG exhibits vastly smaller autocorrelations for both
VIX pricing errors and average option pricing errors, consistent with
improved pricing accuracy.}}
\end{figure}

\subsection{Adaptive Compensation for Misspecification\label{subsec:Misspecification}}

The adaptive nature of the DHNG model, where parameters are adjusted
according to the first-order conditions, enhances the robustness of
the derivative pricing formulae against model misspecification. This
robustness arises because $\eta_{t}$ is automatically adjusted to
compensate for any misspecification that would otherwise lead to increased
pricing errors. While this adaptive feature is a strength of the model,
it also reveals a potential drawback: misspecification can undermine
the interpretation of $\eta_{t}$ as the variance risk ratio.

Expected volatility under $\mathbb{Q}$ can be measured accurately
using observed derivative prices. However, the situation differs under
$\mathbb{P}$, where we must rely on a model to estimate expected
volatility. If the model is misspecified, the resulting model-based
expected volatility may be biased. One common form of misspecification
is parameter instability. \citet{Sichert:2022} demonstrates that
a GARCH model with structural breaks can resolve the pricing kernel
puzzle related to the inverted U-shape. This explanation is further
explored in the empirical analysis by \citet{TongHansenHuang2022},
where they estimate a Markov switching Realized GARCH model with two
states. Their findings show significant U-shaped pricing kernels in
the high-volatility regime, which aligns with our results. However,
the positive variance risk premium (associated with the inverted U-shaped
pricing kernel) they observed in the low-volatility regime is close
to zero and statistically insignificant.

If the model-based expected volatility is biased, this bias will affect
$\eta_{t}$, meaning $\eta_{t}$ may not accurately reflect the true
variance risk ratio. For example, a small $\eta_{t}$ might not indicate
an investor's appetite for variance risk; instead, it could be an
artifact of the Heston-Nandi model overestimating future volatility
relative to actual expectations under $\mathbb{P}$. 

As a robustness check, we computed an alternative VRR that does not
rely on the Heston-Nandi GARCH model. This alternative measure is
based on model-free realized variances and uses a simple AR(1) model
to define expected variance under $\mathbb{P}$. The resulting time
series of this alternative VRR, reported in Appendix \ref{subsec:EMVRR},
closely resembles that of $\eta_{t}$ in Figure \ref{fig:eta}. This
similarity suggests that the observed fluctuations in $\eta_{t}$
are not driven by a flaw specific to the GARCH model (\ref{eq:HNGreturn})-(\ref{eq:HNGgarch}).
However, an oversimplified description of how expected volatility
evolves can be the cause. The AR(1) model and standard GARCH models
both produce volatility forecasts with mean-reverting characteristics.

Regarding misspecification as a potential source of time variation
in $\eta_{t}$, we generated scatterplots of pricing errors against
$\eta_{t}$. These scatterplots do not suggest a systematic relationship
between pricing errors and $\eta_{t}$; see Figure \ref{fig:Miss}
in Appendix \ref{subsec:EMVRR}.

\section{Variance Risk Aversion and Economic Fundamentals\label{sec:Eta-and-EconomicFundamentals}}

Time variation in volatility risk aversion is the key characteristic
of the new pricing kernel. In this section, we relate the time variation
in $\eta_{t}$ to economic fundamentals and several other variables
that are widely used in the asset pricing literature.

We focus on seven core measures of uncertainty, disagreement, and
sentiment. The first set of variables is related to investor disagreements.
Following \citet{BollerslevLiXue2018}, we use two types of proxies
for disagreement. The first concerns macroeconomic fundamentals, for
which we use the forecast dispersion for both the unemployment rate
and GDP growth from the Survey of Professional Forecasters.\footnote{\citet{BollerslevLiXue2018} only use the forecast dispersion for
the one-quarter-ahead unemployment rate.} The second proxy is disagreement about economic policy, for which
we adopt the \textit{Economic Policy Uncertainty} index by \citet{BakerBloomDavis2016}.

We also include the \emph{variance risk premium} (VRP) by \citet{CarrWu2008},
which \citet{GrithHardleKratschmer2017} suggest serves as a proxy
for market uncertainty. Additionally, we incorporate three uncertainty
measures from existing literature: the \emph{Economic Uncertainty
Index} by \citet{BaliBrownCaglayan2014}, the \emph{Survey-based Uncertainty
Index} by \citet{OzturkSheng2018};\footnote{This index is constructed from the Consensus Forecasts publication
by Consensus Economics Inc. The authors provide data up to October
2021, and we fill in the values for the last two months using a simple
predictive regression on all our explanatory variables.} and the well-known \emph{Sentiment} variable by \citet{BakerWurgler2006}.

In addition to these seven core variables, we include the economic
variables previously used in \citet{BakerWurgler2006} and \citet{WelchGoyal2007},
which we label as ``control variables''. Most variables are available
at a monthly frequency. Therefore, we take the monthly average of
$\eta_{t}$ and regress the logarithmically transformed average on
the variables listed above, as well as subsets of these variables.
\begin{sidewaystable}
\caption{Economic Explanation for Time-Varying Variance Risk Aversion}

\begin{centering}
\vspace{0.2cm}
\begin{footnotesize}
\begin{tabularx}{\textwidth}{p{4.5cm}YYYYYYYYYYYYYYYYYYYYYYY}
\toprule 
\midrule
           \multicolumn{9}{c}{Log Variance Risk Ratio (Monthly Average)} \\
    \midrule
 &       &   &       &       &       &       &       &       &  \\[-5pt]
 &  (1)     &  (2) &  (3)     &  (4)     &   (5)    &  (6)      &   (7)     &   (8)     &  (9) \\ [-2pt]
 &       &   &       &       &       &       &       &       &  \\
    Sentiment &       & -0.095*** &       &       &       &       &       &       & -0.093*** \\
          &       & \textit{(0.026)} &       &       &       &       &       &       & \textit{(0.020)} \\
    Dispersion-UNEMP &       &       & 0.139*** &       &       &       &       &       & 0.008 \\
          &       &       & \textit{(0.052)} &       &       &       &       &       & \textit{(0.044)} \\
    Dispersion-GDP &       &       &       & 0.176*** &       &       &       &       & 0.050 \\
          &       &       &       & \textit{(0.040)} &       &       &       &       & \textit{(0.049)} \\
    Economic Policy Uncertainty &       &       &       &       & 0.380*** &       &       &       & 0.208*** \\
          &       &       &       &       & \textit{(0.060)} &       &       &       & \textit{(0.045)} \\
    Survey-based Uncertainty &       &       &       &       &       & 0.148*** &       &       & 0.078*** \\
          &       &       &       &       &       & \textit{(0.026)} &       &       & \textit{(0.025)} \\
    Economic Uncertainty Index &       &       &       &       &       &       & 0.071*** &       & 0.013 \\
          &       &       &       &       &       &       & \textit{(0.009)} &       & \textit{(0.009)} \\
    Variance Risk Premium &       &       &       &       &       &       &       & 0.120*** & 0.096*** \\
          &       &       &       &       &       &       &       & \textit{(0.014)} & \textit{(0.012)} \\
 &       &   &       &       &       &       &       &       &  \\          
    Stock Market Volatility & 0.394*** & 0.395*** & 0.368*** & 0.356*** & 0.294*** & 0.363*** & 0.354*** & 0.395*** & 0.306*** \\
          & \textit{(0.042)} & \textit{(0.041)} & \textit{(0.043)} & \textit{(0.043)} & \textit{(0.039)} & \textit{(0.042)} & \textit{(0.041)} & \textit{(0.029)} & \textit{(0.029)} \\
    Default Yield Spread & 0.385*** & 0.358*** & 0.372*** & 0.365*** & 0.393*** & 0.282*** & 0.246*** & 0.285*** & 0.196*** \\
          & \textit{(0.090)} & \textit{(0.091)} & \textit{(0.088)} & \textit{(0.087)} & \textit{(0.078)} & \textit{(0.087)} & \textit{(0.081)} & \textit{(0.067)} & \textit{(0.058)} \\
    Default Return Spread & 0.026*** & 0.027*** & 0.025*** & 0.025*** & 0.019*** & 0.024*** & 0.021*** & 0.026*** & 0.020*** \\
          & \textit{(0.008)} & \textit{(0.008)} & \textit{(0.007)} & \textit{(0.007)} & \textit{(0.007)} & \textit{(0.007)} & \textit{(0.007)} & \textit{(0.006)} & \textit{(0.005)} \\
    Long Term Yield & 0.028** & 0.030** & 0.040*** & 0.043*** & 0.077*** & 0.016 & 0.045*** & 0.035*** & 0.063*** \\
          & \textit{(0.013)} & \textit{(0.013)} & \textit{(0.013)} & \textit{(0.012)} & \textit{(0.012)} & \textit{(0.011)} & \textit{(0.010)} & \textit{(0.010)} & \textit{(0.010)} \\
    Dividend Price Ratio & -0.457*** & -0.435*** & -0.525*** & -0.529*** & -0.638*** & -0.463*** & -0.403*** & -0.354*** & -0.471*** \\
          & \textit{(0.132)} & \textit{(0.127)} & \textit{(0.130)} & \textit{(0.129)} & \textit{(0.117)} & \textit{(0.133)} & \textit{(0.122)} & \textit{(0.091)} & \textit{(0.085)} \\
    Term Spread & -0.031* & -0.040** & -0.020 & -0.018 & -0.035** & -0.031* & -0.030* & -0.011 & -0.022* \\
          & \textit{(0.018)} & \textit{(0.017)} & \textit{(0.018)} & \textit{(0.018)} & \textit{(0.015)} & \textit{(0.018)} & \textit{(0.016)} & \textit{(0.012)} & \textit{(0.012)} \\          &       &       &       &       &       &       &       &       &  \\
    Other Controls & YES   & YES   & YES   & YES   & YES   & YES   & YES   & YES   & YES \\
          &       &       &       &       &       &       &       &       &  \\
    $R^2$  & 0.709 & 0.722 & 0.722 & 0.732 & 0.756 & 0.741 & 0.761 & 0.817 & 0.866 \\
    $\rho(1)$ & 0.607 & 0.579 & 0.620 & 0.628 & 0.612 & 0.626 & 0.584 & 0.326 & 0.383 \\
\\[0.0cm]
\\[-0.5cm]
\midrule
\bottomrule
\end{tabularx}
\end{footnotesize}
\par\end{centering}
{\small Note: The full-sample monthly regression results from January
1990 to December 2021, and the dependent variable is logarithm of
monthly averaged variance risk ratio estimated by option prices (the
last column of Table \ref{tab:JointEstimation}). The control variables
are from \citet{WelchGoyal2007} and \citet{BakerWurgler2006}. ``Other
Controls'' refer to six macro variables that control for macroeconomic
conditions (monthly growth rate of industrial production, real durables
consumption, real nondurables consumption, real services consumption,
employment, and the NBER recession indicator) and five financial variables
(earnings-to-price ratio, book-to-market ratio, net equity expansion,
long-term rate of return, and the consumer price index). Most of these
additional control variables are insignificant. Newey-West standard
errors are reported in parentheses. $\rho(1)$ the first-order autocorrelation
for residuals.\label{tab:EconomicFundamentals}}{\small\par}
\end{sidewaystable}

The full-sample monthly regression results are presented in Table
\ref{tab:EconomicFundamentals}. The seven core variables are listed
at the top of the table, followed by the six control variables that
contributed the most to explaining variation in $\eta_{t}$. There
are eleven additional control variables, labeled ``Other Controls'',
which were largely insignificant.\footnote{We follow {\small\citet{WelchGoyal2007} in the definitions of these
variables. }For instance, the yield spread is the difference between
BAA- and AAA-rated corporate bond yields, and the default return spread
is the difference between the return on long-term corporate bonds
and government bonds. See {\small\citet{WelchGoyal2007} for more
details.}} To conserve space, the results for all control variables (along with
their labels) are presented in the Appendix \ref{subsec:FullEconomic}.

The first column in Table \ref{tab:EconomicFundamentals} corresponds
to the specification that includes only the control variables, which
explain 70.9\% of the variation in the monthly average of $\eta$.
The six most significant control variables are stock market volatility,
the dividend-price ratio, and four variables related to interest rates
and credit risk.

Columns two through eight present the results for specifications that
include a single core variable in addition to the control variables.
Every core variable is significant in these regressions, and the signs
of their estimated coefficients align with expectations. The sentiment
variable has a negative coefficient, indicating that low sentiment
is associated with high volatility risk aversion. The five uncertainty-related
variables all have positive coefficients, suggesting that high uncertainty
is associated with high volatility aversion.

Finally, the variance risk premium is an empirical measure of the
difference between volatility under $\mathbb{Q}$ and $\mathbb{P}$.
As expected, this variable has a positive coefficient and contributes
the most to the increase in $R^{2}$. The last column shows the results
for a ``kitchen sink regression'' that includes all explanatory variables.
Most of the variables in Table \ref{tab:EconomicFundamentals} remain
significant in this regression, and the $R^{2}$ increases to nearly
87\%. However, three core variables--Forecast Dispersion of Unemployment,
Forecast Dispersion of GDP Growth, and the Economic Uncertainty Index--become
insignificant in the kitchen sink regression, possibly due to collinearity,
as the first two and the last two variables exhibit the highest correlations
among all core variables (88.6\% and 54.8\%), respectively.

\section{Summary\label{sec:Summary}}

We have introduced a novel option pricing model based on a flexible
pricing kernel with time-varying risk aversion. This model, denoted
DHNG, offers an elegant and tractable framework when combined with
the Heston-Nandi GARCH model. The variance risk ratio, $\eta_{t}$,
emerges as the fundamental variable that captures all time variation
in the pricing kernel. This ratio is functionally linked to the variance
risk aversion parameter, which defines the curvature of the pricing
kernel, and the framework can generate the empirically observed shapes
of the pricing kernel. Additionally, $\eta_{t}$ is closely related
to a range of well-known measures of sentiment, disagreement, uncertainty,
and other key economic variables.

Under the assumptions characterizing the DHNG model, we derived a
closed-form expression for the VIX and an approximate option pricing
formula. This formula is based on a novel approximation method that
can yield affine-type pricing formulae for non-affine models. The
method has proven to be highly accurate in simulation studies. Our
comprehensive empirical analysis demonstrated that DHNG significantly
reduces derivative pricing errors. Specifically, DHNG typically reduces
the root mean square error of pricing errors by more than 50\% compared
to the CHNG models where $\eta_{t}$ is constant. This substantial
reduction in pricing errors is observed both in-sample and out-of-sample.

The time-varying pricing kernel can also be combined with other GARCH
models, such as EGARCH, GJR-GARCH, and NGARCH, see \citet{Nelson91},
\citet{Glosten1993}, and \citet{Engle_Ng_1993}, as well as volatility
models that utilize realized measures of volatility, such as Realized
GARCH models by \citet{HansenHuangShek:2012} and \citet{HansenHuang:2016},
the GARV model by \citet{ChristoffersenFeunouJacobsMeddahi2014},
and the LHARG, see \citet{MajewskiBormettiCorsi2015}. For some of
these models, such as GARV and LHARG, it is possible to define an
Auxiliary models with an affine MGF, and it may be possible to establish
closed-form option pricing expressions with the approximation method
we have proposed in this paper. Other approximation methods, such
as that by \citet{DuanGauthierSimonato1999}, may be applicable to
the non-affine models. 

Analytical results are generally easier to establish in continuous-time
models, such as those by \citet{Heston1993} and \citet{CoxIngersollRoss85},
which partly explains their popularity among practitioners. However,
the discrete-time HNG model also offers an analytical option pricing
formula, and we successfully derived an analytical option pricing
formula for our discrete-time model using the new approximation method.
In contrast, obtaining analytical results in a continuous-time framework
with a time-varying VRR would be highly challenging due to the non-affine
structure. Additionally, the estimation process would be further complicated
by the presence of the latent variable $\eta_{t}$. The key advantage
of discrete-time GARCH models over stochastic volatility models lies
in their relative ease of estimation. This is particularly beneficial
in empirical studies involving a large number of options. Our empirical
implementation of the DHNG model retains a straightforward analytical
log-likelihood function, which greatly simplifies the estimation process.

\bibliographystyle{apalike}
\bibliography{prh}

\clearpage{}

\appendix

\section{Appendix of Proofs\label{sec:Appendix-of-Proofs}}

\subsection{Proof of Theorem \ref{thm:DynUnderQ}\label{sec:Proof_DynUnderQ}}

\setcounter{equation}{0}\renewcommand{\theequation}{A.\arabic{equation}}
\setcounter{lem}{0}\renewcommand{\thelem}{A.\arabic{lem}}

The pricing kernel with dynamic risk parameters takes the form\textit{
\[
M_{t+1,t}=\frac{\exp\left(\phi_{t}R_{t+1}+\xi_{t}h_{t+2}\right)}{\mathbb{E}_{t}^{\mathbb{P}}[\exp(\phi_{t}R_{t+1}+\xi_{t}h_{t+2})]},
\]
}where the parameters $\phi_{t}$ and $\xi_{t}$ are $\mathcal{G}_{t}$-measurable.
The pricing kernel implies some dynamic properties in risk-neutral
measure satisfying the no-arbitrage conditions: $\mathbb{E}_{t}^{\mathbb{P}}[M_{t+1,t}]=1$
and $\mathbb{E}_{t}^{\mathbb{P}}[M_{t+1,t}\exp(R_{t+1})]=\exp(r)$.
The first is trivial, whereas the second condition reads
\begin{equation}
\frac{\mathbb{E}_{t}^{\mathbb{P}}[\exp\{(\phi_{t}+1)R_{t+1}+\xi_{t}h_{t+2}\}]}{\mathbb{E}_{t}^{\mathbb{P}}[\exp(\phi_{t}R_{t+1}+\xi_{t}h_{t+2})]}=\exp(r).\label{eq:Noarb}
\end{equation}
The MGF of $(R_{t+1},h_{t+2})$ under $\mathbb{P}$ is
\begin{align}
 & \text{\ensuremath{\mathbb{E}_{t}^{\mathbb{P}}}}\left[\exp\left(v_{1}R_{t+1}+v_{2}h_{t+2}\right)\right]\nonumber \\
= & \exp\left[v_{1}r+v_{2}\omega+[v_{1}(\lambda-\tfrac{1}{2})+v_{2}(\beta+\alpha\gamma^{2})]h_{t+1}\right]\times\text{\ensuremath{\mathbb{E}_{t}^{\mathbb{P}}}}\left[\exp\left(v_{2}\alpha z_{t+1}^{2}+(v_{1}-2v_{2}\alpha\gamma)\sqrt{h_{t+1}}z_{t+1}\right)\right]\nonumber \\
= & \exp\left[v_{1}r+v_{2}\omega+[v_{1}(\lambda-\tfrac{1}{2})+v_{2}(\beta+\alpha\gamma^{2})]h_{t+1}-\tfrac{1}{2}\log\left(1-2v_{2}\alpha\right)+\tfrac{(v_{1}-2v_{2}\alpha\gamma)^{2}}{2\left(1-2v_{2}\alpha\right)}h_{t+1}\right]\nonumber \\
= & \exp\left[v_{1}r+v_{2}\omega-\tfrac{1}{2}\log\left(1-2v_{2}\alpha\right)+\left(v_{1}(\lambda-\tfrac{1}{2})+v_{2}(\beta+\alpha\gamma^{2})+\tfrac{(v_{1}-2v_{2}\alpha\gamma)^{2}}{2\left(1-2v_{2}\alpha\right)}\right)h_{t+1}\right]\label{eq:MGF_R_H}
\end{align}
Substituting $(\phi_{t+1}+1,\xi_{t})$ and $(\phi_{t+1},\xi_{t})$
for $(v_{1},v_{2})$ and taking their ratio, reveals that (\ref{eq:Noarb})
simplifies to $[\lambda-\tfrac{1}{2}+\frac{1+2\phi_{t}-4\xi_{t}\alpha\gamma}{2(1-2\xi_{t}\alpha)}]h_{t+1}=0.$
Next, solving for $\lambda$ yields 
\begin{equation}
\lambda=\eta_{t}\left(\gamma-\phi_{t}-\tfrac{1}{2}\right)-\gamma+\tfrac{1}{2},\quad\text{where}\quad\eta_{t}\equiv\frac{1}{1-2\alpha\xi_{t}},\label{eq:lambda}
\end{equation}
and the results stated in part $(i)$ of Theorem \ref{thm:DynUnderQ}
are established.

The risk-neutral dynamic is deduced from the MGF of $z_{t+1}$ under
the $\mathbb{\mathbb{Q}}$ measure is:
\begin{align}
\text{\ensuremath{\mathbb{E}_{t}^{\mathbb{Q}}}}\left[\exp\left(v_{1}z_{t+1}\right)\right] & =\exp\left[\left(\tfrac{\phi_{t}-2\xi_{t}\alpha\gamma}{1-2\xi_{t}\alpha}\sqrt{h_{t+1}}\right)v_{1}+\tfrac{1}{2\left(1-2\xi_{t}\alpha\right)}v_{1}^{2}\right]\nonumber \\
 & =\exp\left[\left[-\left(\lambda+\tfrac{1}{2}\eta_{t}-\tfrac{1}{2}\right)\sqrt{h_{t+1}}\right]v_{1}+\tfrac{1}{2}\eta_{t}v_{1}^{2}\right]\label{eq: MGFzInQ}
\end{align}
The last equality comes from (\ref{eq:lambda}). The expression (\ref{eq: MGFzInQ})
shows that $z_{t+1}$ is also normally distributed under $\mathbb{\mathbb{Q}}$,
but with a different mean and a different variance (than under $\mathbb{P}$).
So we define
\[
z_{t+1}^{*}=\frac{z_{t+1}-\mathbb{E}_{t}^{\mathbb{Q}}(z_{t+1})}{\sqrt{\mathrm{var}_{t}^{\mathbb{Q}}(z_{t+1})}}=\frac{1}{\sqrt{\eta_{t}}}\left(z_{t+1}+\left(\lambda+\frac{1}{2}\eta_{t}-\tfrac{1}{2}\right)\sqrt{h_{t+1}}\right),
\]
such that $z_{t}^{\ast}|\mathcal{G}_{t}\sim iid\ N(0,1)$ under $\mathbb{\mathbb{Q}}$.
Substituting the risk-neutralized quantities, $z_{t+1}^{*}$ and $h_{t+1}^{*}=\eta_{t}h_{t+1}$
into (\ref{eq:HNGreturn}), we arrive at the expressions stated in
part $(ii)$. Finally, part $(iii)$ is an implication of the MGF
for $z_{t+1}$ under $\mathbb{\mathbb{Q}}$, see (\ref{eq: MGFzInQ}).\hfill{}$\square$

\subsection{Proof of Lemma \ref{lem:ShapeDynamicKernel}\label{subsec:Proof-of-PK}}

Using (\ref{eq:MGF_R_H}), the pricing kernel in (\ref{eq:NewPK})
can expressed as
\[
\log M_{t+1,t}=\delta_{t}+\phi_{t}\left(R_{t+1}-r\right)+\theta_{t}h_{t+1}+\xi_{t}h_{t+2},
\]
where $\theta_{t}=-\left(\phi_{t}\left(\lambda-\tfrac{1}{2}\right)+\xi_{t}\left(\beta+\alpha\gamma^{2}\right)+\tfrac{\left(\phi_{t}-2\xi_{t}\alpha\gamma\right)^{2}}{2\left(1-2\xi_{t}\alpha\right)}\right)$
and $\delta_{t}=-\xi_{t}\omega+\tfrac{1}{2}\log\left(1-2\xi_{t}\alpha\right)$. 

Using the expression of $h_{t+2}$, we have
\begin{align*}
\log M_{t+1,t} & =\delta_{t}+\phi_{t}\left(R_{t+1}-r\right)+\theta_{t}h_{t+1}+\xi_{t}\left(\omega+\beta h_{t+1}+\alpha\left(z_{t+1}-\gamma\sqrt{h_{t+1}}\right)^{2}\right)\\
 & =[\delta_{t}+\xi_{t}\omega]+\phi_{t}\left(R_{t+1}-r\right)+\left(\theta_{t}+\xi_{t}\beta\right)h_{t+1}+\tfrac{\xi_{t}\alpha}{h_{t+1}}\left(R_{t+1}-r-\left(\lambda-\tfrac{1}{2}+\gamma\right)h_{t+1}\right)^{2}\\
 & =[\delta_{t}+\xi_{t}\omega]+\phi_{t}\left(R_{t+1}-r\right)+\left(\theta_{t}+\xi_{t}\beta\right)h_{t+1}+\tfrac{\xi_{t}\alpha}{h_{t+1}}\left(R_{t+1}-r\right)^{2}\\
 & \quad+\tfrac{\xi_{t}\alpha}{h_{t+1}}\left(\lambda-\tfrac{1}{2}+\gamma\right)^{2}h_{t+1}^{2}-2\tfrac{\xi_{t}\alpha}{h_{t+1}}\left(R_{t+1}-r\right)\left(\lambda-\tfrac{1}{2}+\gamma\right)h_{t+1}\\
 & =\kappa_{0,t}+\kappa_{1,t}h_{t+1}+\kappa_{2,t}\left(R_{t+1}-r\right)+\tfrac{\xi_{t}\alpha}{h_{t+1}}\left(R_{t+1}-r\right)^{2},
\end{align*}
where $\kappa_{0,t}=\delta_{t}+\xi_{t}\omega=\tfrac{1}{2}\log\left(1-2\xi_{t}\alpha\right)=-\frac{1}{2}\log\eta_{t}$,
and $\kappa_{2,t}=\phi_{t}-2\alpha\xi_{t}\left(\lambda-\tfrac{1}{2}+\gamma\right)=-\lambda,$
where we used that $\ensuremath{\lambda=\ensuremath{\frac{1}{1-2\alpha\xi_{t}}}\left(\gamma-\phi_{t}-\frac{1}{2}\right)-\gamma+\frac{1}{2}}$,
see (\ref{eq:lambda}), and
\begin{align*}
\kappa_{1,t} & =\theta_{t}+\xi_{t}\beta+\xi_{t}\alpha\left(\lambda-\tfrac{1}{2}+\gamma\right)^{2}\\
 & =\theta_{t}+\xi_{t}\left(\beta+\alpha\gamma^{2}\right)+\xi_{t}\alpha\left(\lambda-\tfrac{1}{2}\right)^{2}+2\xi_{t}\alpha\gamma\left(\lambda-\tfrac{1}{2}\right)\\
 & =-\phi_{t}\left(\lambda-\tfrac{1}{2}\right)+\xi_{t}\alpha\left(\lambda-\tfrac{1}{2}\right)^{2}+2\xi_{t}\alpha\gamma\left(\lambda-\tfrac{1}{2}\right)-\tfrac{\left(\phi_{t}-2\xi_{t}\alpha\gamma\right)^{2}}{2\left(1-2\xi_{t}\alpha\right)}\\
 & =-\alpha\xi_{t}\left(\lambda-\tfrac{1}{2}\right)^{2}+\lambda\left(\lambda-\tfrac{1}{2}\right)-\tfrac{\left(\lambda-2\xi_{t}\alpha(\lambda-\tfrac{1}{2})\right)^{2}}{2\left(1-2\xi_{t}\alpha\right)}\\
 & =\tfrac{1}{2}\left(1-2\alpha\xi_{t}\right)\left(\lambda-\tfrac{1}{2}\right)^{2}-\tfrac{1}{2}\left(\lambda-\tfrac{1}{2}\right)^{2}+\lambda\left(\lambda-\tfrac{1}{2}\right)-\tfrac{1}{2}\left[\tfrac{1}{\eta_{t}}\left(\lambda-\tfrac{1}{2}\right)^{2}+\tfrac{1}{4}\eta_{t}+(\lambda-\tfrac{1}{2})\right]\\
 & =-\tfrac{1}{8}\eta_{t}+\tfrac{1}{2}\left(\lambda-\tfrac{1}{2}\right)^{2}.
\end{align*}

\subsection{Proof of Theorem \ref{thm:VixPricing}\label{sec:Proof_VIXpricing}}

We rewrite the volatility dynamic under $\mathbb{Q}$ measure as
\[
h_{t+1}^{*}=(\omega_{t}^{*}+\alpha_{t}^{*})+(\beta_{t}^{*}+\alpha_{t}^{*}\gamma_{t-1}^{*2})h_{t}^{*}+\nu_{t},
\]
where $\nu_{t}=\alpha_{t}^{*}((z_{t}^{*})^{2}-1-2\gamma_{t-1}^{*}z_{t}^{*}\sqrt{h_{t}^{*}})$
is a zero-mean variation condition on $\mathcal{G}_{t-1}$ due to
the facts that $z_{t}^{*}|\mathcal{G}_{t-1}\overset{\mathbb{Q}}{\sim}iid\ N\left(0,1\right)$,
and the independence between $\varepsilon_{t}$ and $z_{t}$.

We denote $A_{t}\equiv\omega_{t}^{*}+\alpha_{t}^{*}$ and $B_{t}\equiv\beta_{t}^{\ast}+\alpha_{t}^{*}\gamma_{t-1}^{\ast2}$,
then we have
\begin{eqnarray*}
h_{t+k}^{*} & = & A_{t+k-1}+B_{t+k-1}h_{t+k-1}^{*}+\nu_{t+k-1}\\
 & = & A_{t+k-1}+B_{t+k-1}A_{t+k-2}+B_{t+k-1}B_{t+k-2}h_{t+k-2}^{*}+B_{t+k-1}\nu_{t+k-2}+\nu_{t+k-1}\\
 & = & A_{t+k-1}+B_{t+k-1}A_{t+k-2}+\cdots+B_{t+k-1}\cdots B_{t+2}A_{t+1}+B_{t+k-1}\cdots B_{t+1}h_{t+1}^{*}\\
 &  & +\sum_{i=1}^{k-2}B_{t+k-1}B_{t+k-2}\cdots B_{t+i+1}\nu_{t+i}+\nu_{t+k-1}.
\end{eqnarray*}
In order to get its analytical formula, we replace $B_{t}$ with $\tilde{B}_{t}\equiv\beta_{t}^{\ast}+\alpha_{t}^{*}\tilde{\gamma}_{t-1}^{\ast2}$
where $\tilde{\gamma}_{t}^{\ast}=\tfrac{1}{\eta_{t}}(\gamma+\lambda)$.
Empirically $\tilde{\gamma}_{t}^{\ast}$ is indistinguishable from
$\gamma_{t}^{*}=\tfrac{1}{\eta_{t}}(\gamma+\lambda+\ensuremath{\tfrac{1}{2}(\eta_{t}-1)})$
because the magnitude of $\gamma+\lambda$ is much larger than $\frac{1}{2}(\eta_{t}-1)$.
Empirically we have $\gamma+\lambda\approx253$ whereas $\frac{1}{2}(\eta_{t}-1)$
ranges between $-0.27$ and 1.39, with an average value of about 0.07.
Then it follows that\footnote{Recall that $\alpha_{t}^{\ast}=\alpha\eta_{t}\eta_{t-1}$ and $\beta_{t}^{*}=\beta\tfrac{\eta_{t}}{\eta_{t-1}}$.}
\[
\tilde{B}_{t}=\beta_{t}^{\ast}+\alpha_{t}^{\ast}\gamma_{t-1}^{\ast2}=\beta\frac{\eta_{t}}{\eta_{t-1}}+\alpha\eta_{t}\eta_{t-1}(\tfrac{\gamma+\lambda}{\eta_{t-1}})^{2}=\tilde{\beta}\frac{\eta_{t}}{\eta_{t-1}},
\]
where $\tilde{\beta}\equiv\beta+\alpha(\gamma+\lambda)^{2}$, so that
$\tilde{B}_{t+k}\tilde{B}_{t+k-1}\cdots\tilde{B}_{t+i}=\tilde{\beta}^{k-i+1}\frac{\eta_{t+k}}{\eta_{t+i-1}}$,
and 
\[
\tilde{B}_{t+k}\cdots\tilde{B}_{t+i+1}A_{t+i}=\tilde{\beta}^{k-i}\frac{\eta_{t+k}}{\eta_{t+i}}(\omega\eta_{t+i}+\alpha\eta_{t+i}\eta_{t+i-1})=\tilde{\beta}^{k-i}(\omega\eta_{t+k}+\alpha\eta_{t+k}\eta_{t+i-1}).
\]
Using these results, we have
\[
\mathbb{E}_{t}^{\mathbb{Q}}(h_{t+k}^{*})=\sum_{i=2}^{k}\tilde{\beta}^{k-i}\mathbb{E}_{t}^{\mathbb{Q}}\left(\omega\eta_{t+k-1}+\alpha\eta_{t+k-1}\eta_{t+i-2}\right)+\tilde{\beta}^{k-1}\mathbb{E}_{t}^{\mathbb{Q}}\left(\frac{\eta_{t+k-1}}{\eta_{t}}\right)h_{t+1}^{*}+\delta_{k},
\]
where $\delta_{k}$ is the term that arises from the substitution
of $\tilde{B}_{t}$ for $B_{t}$. 

The $M$-days ahead VIX can be calculated as the annualized arithmetic
average of the expected daily variance over the following month under
the risk-neutral measure, i.e., 
\[
\mathrm{VIX}_{t}=A\times\sqrt{\frac{1}{M}\sum_{k=1}^{M}\mathbb{E}_{t}^{\mathbb{Q}}(h_{t+k}^{*})}.
\]
Where $A=100\sqrt{252}$ is the annualizing factor. The model-implied
squared VIX is given by,
\begin{eqnarray}
\frac{A^{2}}{M}\sum_{k=1}^{M}\mathbb{E}_{t}^{\mathbb{Q}}(h_{t+k}^{*}) & = & \underset{a_{1}(M,\sigma^{2})}{\underbrace{\frac{A^{2}}{M}\sum_{k=2}^{M}\sum_{i=2}^{k}\tilde{\beta}^{k-i}\mathbb{E}_{t}^{\mathbb{Q}}(\omega\eta_{t+k-1}+\alpha\eta_{t+k-1}\eta_{t+i-2})}}\nonumber \\
 &  & +\underset{a_{2}(M,\sigma^{2})}{\underbrace{\frac{A^{2}}{M}\sum_{k=1}^{M}\tilde{\beta}^{k-1}\mathbb{E}_{t}^{\mathbb{Q}}\left(\frac{\eta_{t+k-1}}{\eta_{t}}\right)}h_{t+1}^{*}}+\Delta,\label{eq:ModelVixRaw}
\end{eqnarray}
where $\Delta=\frac{A^{2}}{M}\sum\delta_{k}$. We provide bounds for
the approximation error term $\Delta$ in Lemma \ref{lem:bound},
and show that this term is negligible in practice.

For the $\operatorname{ARMA}(p,q)$ structure, $\log\eta_{t}=\left(1-\sum_{i=1}^{p}\varphi_{i}\right)\zeta+\sum_{i=1}^{p}\varphi_{i}\log\eta_{t-i}+\sum_{j=1}^{q}\theta_{j}\varepsilon_{t-j}+\varepsilon_{t}$,
we will show that 
\begin{equation}
\log\eta_{t+k}=A_{k}+\sum_{i=1}^{p}B_{k,i}\log\eta_{t+1-i}+\sum_{j=1}^{q}C_{k,j}\varepsilon_{t+1-j}+\sum_{l=1}^{k}D_{k,l}\varepsilon_{t+l},\label{eq:logARMA}
\end{equation}
for $k\geq1$, for suitable constants, $A_{k}$, $B_{k,i}$, $C_{k,j}$
and $D_{k,l}$. We show this by induction. For $k=1$ it is easy to
verify that $A_{1}=\left(1-\sum_{i=1}^{p}\varphi_{i}\right)\zeta$,
$B_{1,i}=\varphi_{i}$, $C_{1,j}=\theta_{j}$ and $D_{1,1}=1$. Suppose
that (\ref{eq:logARMA}) holds for $k\geq1$, then for $k+1$ we have
\begin{align*}
\log\eta_{t+k+1} & =A_{k}+\sum_{i=1}^{p}B_{k,i}\log\eta_{t+2-i}+\sum_{j=1}^{q}C_{k,j}\varepsilon_{t+2-j}+\sum_{l=1}^{k}D_{k,l}\varepsilon_{t+1+l}\\
 & =A_{k}+B_{k,1}\log\eta_{t+1}+\sum_{i=1}^{p-1}B_{k,i+1}\log\eta_{t+1-i}+C_{k,1}\varepsilon_{t+1}+\sum_{j=1}^{q-1}C_{k,j+1}\varepsilon_{t+1-j}+\sum_{l=2}^{k+1}D_{k,l-1}\varepsilon_{t+l}\\
 & =A_{k}+\sum_{i=1}^{p-1}B_{k,i+1}\log\eta_{t+1-i}+C_{k,1}\varepsilon_{t+1}+\sum_{j=1}^{q-1}C_{k,j+1}\varepsilon_{t+1-j}+\sum_{l=2}^{k+1}D_{k,l-1}\varepsilon_{t+l}\\
 & \quad+B_{k,1}\left[\left(1-\sum_{i=1}^{p}\varphi_{i}\right)\zeta+\sum_{i=1}^{p}\varphi_{i}\log\eta_{t+1-i}+\sum_{j=1}^{q}\theta_{j}\varepsilon_{t+1-j}+\varepsilon_{t+1}\right]\\
 & =A_{k+1}+\sum_{i=1}^{p}B_{k+1,i}\log\eta_{t+1-i}+\sum_{j=1}^{q}C_{k+1,j}\varepsilon_{t+1-j}+\sum_{l=1}^{k+1}D_{k+1,l}\varepsilon_{t+l},
\end{align*}
where
\begin{align*}
A_{k+1} & =A_{k}+B_{k,1}\left(1-\sum_{i=1}^{p}\varphi_{i}\right)\zeta,\quad\qquad\qquad B_{k+1,i}=\begin{cases}
\varphi_{i}B_{k,1}+B_{k,i+1} & 1\leq i\leq p-1,\\
\varphi_{i}B_{k,1} & i=p,
\end{cases}\\
C_{k+1,j} & =\begin{cases}
\theta_{j}B_{k,1}+C_{k,j+1} & 1\leq j\leq q-1,\\
\theta_{j}B_{k,1} & j=q,
\end{cases}\quad D_{k+1,l}=\begin{cases}
C_{k,1}+B_{k,1} & l=1,\\
D_{k,l-1} & 2\leq l\leq k+1.
\end{cases}
\end{align*}
This completes the proof of (\ref{eq:logARMA}). 

Suppose that $\varepsilon_{t}$ is iid and let $\Psi(s)\equiv\log\mathbb{E}[\exp(s\varepsilon_{t})]$
denote its log-MGF. Then we have
\[
\mathbb{E}_{t}\left(\eta_{t+k}\right)=\exp\left\{ A_{k}+\sum_{i=1}^{p}B_{k,i}\log\eta_{t+1-i}+\sum_{j=1}^{q}C_{k,j}\varepsilon_{t+1-j}+\sum_{l=1}^{k}\Psi\left(D_{k,l}\right)\right\} ,
\]
and for $i<k$, we have
\begin{align*}
\log(\eta_{t+k}\eta_{t+i}) & =A_{k}+A_{i}+\sum_{i^{\prime}=1}^{p}\left(B_{k,i^{\prime}}+B_{i,i^{\prime}}\right)\log\eta_{t+1-i^{\prime}}+\sum_{j=1}^{q}\left(C_{k,j}+C_{i,j}\right)\varepsilon_{t+1-j}\\
 & \quad+\sum_{l=1}^{i}\left(D_{k,l}+D_{i,l}\right)\varepsilon_{t+l}+\sum_{l=i+1}^{k}D_{k,l}\varepsilon_{t+l},\\
\mathbb{E}_{t}\left(\eta_{t+k}\eta_{t+i}\right) & =\exp\left\{ A_{k}+A_{i}+\sum_{i^{\prime}=1}^{p}\left(B_{k,i^{\prime}}+B_{i,i^{\prime}}\right)\log\eta_{t+1-i^{\prime}}+\sum_{j=1}^{q}\left(C_{k,j}+C_{i,j}\right)\varepsilon_{t+1-j}\right.\\
 & \quad\left.+\sum_{l=1}^{i}\Psi\left(D_{k,l}+D_{i,l}\right)+\sum_{l=i+1}^{k}\Psi\left(D_{k,l}\right)\right\} .
\end{align*}
We can now compute the two key terms in (\ref{eq:ModelVixRaw}). 

The expression simplify in the special case with an AR(1) structure,
$\log\eta_{t}=(1-\varphi)\zeta+\varphi\log\eta_{t-1}+\varepsilon_{t}$,
where $|\varphi|<1$. It follows that $\log\eta_{t+k}=\left(1-\varphi^{k}\right)\zeta+\varphi^{k}\log\eta_{t}+\sum_{j=1}^{k}\varphi^{k-j}\varepsilon_{t+j},$
such that, for $i<k$, we have
\[
\log(\eta_{t+k}\eta_{t+i})=\left(2-\varphi^{k}-\varphi^{i}\right)\zeta+\left(\varphi^{k}+\varphi^{i}\right)\log\eta_{t}+\sum_{j=1}^{i}\left(\varphi^{k-j}+\varphi^{i-j}\right)\varepsilon_{t+j}+\sum_{j^{\prime}=i+1}^{k}\varphi^{k-j^{\prime}}\varepsilon_{t+j^{\prime}},
\]
$\mathbb{E}_{t}\left(\eta_{t+k}\right)=\exp\left\{ (1-\varphi^{k})\zeta+\varphi^{k}\log\eta_{t}+\sum_{j=0}^{k-1}\Psi(\varphi^{j})\right\} $,
and
\[
\mathbb{E}_{t}\left(\eta_{t+k}\eta_{t+i}\right)=\exp\left\{ (2-\varphi^{k}-\varphi^{i})\zeta+(\varphi^{k}+\varphi^{i})\log\eta_{t}+\sum_{j=1}^{i}\Psi(\varphi^{k-j}+\varphi^{i-j})+\sum_{j^{\prime}=i+1}^{k}\Psi(\varphi^{k-j^{\prime}})\right\} .
\]
We can therefore express the two key terms in (\ref{eq:ModelVixRaw})
as, 
\begin{align*}
a_{1}(M,\sigma^{2}) & =\frac{A^{2}}{M}\sum_{k=2}^{M}\sum_{i=2}^{k}\tilde{\beta}^{k-i}\left[\omega\Lambda_{1}(k)\eta_{t}^{\varphi^{k-1}}+\alpha\Lambda_{2}(k,i)\eta_{t}^{\varphi^{k-1}+\varphi^{i-2}}\right],
\end{align*}
and $a_{2}(M,\sigma^{2})=\frac{A^{2}}{M}\sum_{k=1}^{M}\tilde{\beta}^{k-1}\Lambda_{1}(k)\eta_{t}^{\varphi^{k-1}-1}$,
respectively,  where
\begin{align*}
\Lambda_{1}(k) & =\exp\left\{ \left(1-\varphi^{k-1}\right)\zeta+\sum_{j=0}^{k-2}\Psi(\varphi^{j})\right\} ,\\
\Lambda_{2}(k,i) & =\exp\left\{ (2-\varphi^{k-1}-\varphi^{i-2})\zeta+\sum_{j=1}^{i-2}\Psi(\varphi^{k-1-j}+\varphi^{i-2-j})+\sum_{j^{\prime}=i-1}^{k-1}\Psi(\varphi^{k-1-j^{\prime}})\right\} .
\end{align*}
If $\varepsilon_{t}$ is normally distributed with zero mean and variance
$\sigma^{2}$, we have $\Psi(\phi)=\frac{1}{2}\sigma^{2}\phi^{2}$.\hfill{}$\square$
\begin{lem}[Bound for $\Delta$]
\label{lem:bound}Suppose that $\ensuremath{\eta_{t}\in[\eta_{L},\eta_{H}]}$
is bounded where $\eta_{L},\eta_{H}>0$. Then $\Delta\leq\delta(\tilde{\beta}_{{\rm max}})-\delta(\tilde{\beta}_{{\rm min}})$,
where
\[
\delta(x)=A\times\sqrt{\frac{1}{M}\left[\sum_{k=2}^{M}\sum_{i=2}^{k}x^{k-i}\mathbb{E}_{t}^{\mathbb{Q}}(\omega\eta_{t+k-1}+\alpha\eta_{t+k-1}\eta_{t+i-2})+\sum_{k=1}^{M}x^{k-1}\mathbb{E}_{t}^{\mathbb{Q}}\left(\frac{\eta_{t+k-1}}{\eta_{t}}\right)h_{t+1}^{*}\right]},
\]
and $\tilde{\beta}_{{\rm min}}=\beta+\alpha(\gamma+\lambda+\ensuremath{\tfrac{1}{2}(\eta_{L}-1)})^{2}$
and $\tilde{\beta}_{{\rm max}}=\beta+\alpha(\gamma+\lambda+\ensuremath{\tfrac{1}{2}(\eta_{H}-1)})^{2}$.
\end{lem}
\begin{proof}
We have $B_{t}=\beta_{t}^{\ast}+\alpha_{t}^{\ast}\gamma_{t-1}^{\ast2}\in[\tilde{\beta}_{{\rm min}},\tilde{\beta}_{{\rm max}}]\frac{\eta_{t}}{\eta_{t-1}}$.
Since $\delta(x)$ is a monotone function in $x$, there exists an
$\bar{\beta}\in[\tilde{\beta}_{{\rm min}},\tilde{\beta}_{{\rm max}}]$,
such that $\delta(\bar{\beta})=A\times\sqrt{\tfrac{1}{M}\sum_{k=1}^{M}\mathbb{E}_{t}^{\mathbb{Q}}(h_{t+k}^{*})}$
(i.e., true value of VIX), and $\Delta=|\delta(\bar{\beta})-\delta(\tilde{\beta})|$,
then we have $\Delta\leq\delta(\tilde{\beta}_{{\rm max}})-\delta(\tilde{\beta}_{{\rm min}})$.
Note that $\delta(x)$ can be expressed in terms of $a_{1}(\eta_{t},M,\sigma^{2})$
and $a_{2}(\eta_{t},M,\sigma^{2})$. 
\end{proof}

\subsection{Proof of Lemma \ref{lem:Closed-form-OP}\label{sec:Proof_OP}}
\begin{lem}[Full Version of Lemma \ref{lem:Closed-form-OP}]
\label{lem:MGFpredetermined}In the predetermined case with path
$\bar{\boldsymbol{\eta}}_{t,M}$, the conditional MGF for future cumulative
returns has the following exponentially affine form:
\begin{align*}
g_{t,M}(s|\bar{\boldsymbol{\eta}}_{t,M}) & =\mathbb{E}_{t}^{\mathbb{Q}}(\exp(s\sum_{\tau=t+1}^{T}R_{\tau}))=\exp\left(A_{T}(s,M)+B_{T}(s,M)h_{t+1}^{*}\right),
\end{align*}
where $M=T-t$, $A_{T}(s,m)$ and $B_{T}(s,m)$ are recursively given
by
\begin{align}
A_{T}(s,m+1) & =A_{T}(s,m)+sr+B_{T}(s,m)\omega_{T-m}^{*}-\tfrac{1}{2}\log(1-2\alpha_{T-m}^{*}B_{T}(s,m)),\label{eq:A}\\
B_{T}(s,m+1) & =s(\gamma_{T-m-1}^{*}-\tfrac{1}{2})-\tfrac{1}{2}\gamma_{T-m-1}^{*2}+\beta_{T-m}^{*}B_{T}(s,m)+\tfrac{(s-\gamma_{T-m-1}^{*})^{2}}{2(1-2\alpha_{T-m}^{*}B_{T}(s,m))},\label{eq:B}
\end{align}
with initial values $A_{T}(s,1)=sr$ and $B_{T}(s,1)=\tfrac{1}{2}\left(s^{2}-s\right)$,
where all terms in (\ref{eq:A}) and (\ref{eq:B}) are predetermined.
\end{lem}
\begin{proof}
For later use, we note that the MGF for $(R_{t+1},h_{t+2}^{*})$ is
given by
\begin{align}
\mathbb{E}_{t}^{\mathbb{Q}}(\exp(sR_{t+1}+uh_{t+2}^{*})) & =\exp\Biggl[sr+u\omega_{t+1}^{*}-\tfrac{1}{2}\log(1-2u\alpha_{t+1}^{*})\nonumber \\
 & \quad+\left.\left(s(\gamma_{t}^{*}-\tfrac{1}{2})+u\beta_{t+1}^{*}-\tfrac{1}{2}\gamma_{t}^{*2}+\tfrac{(s-\gamma_{t}^{*})^{2}}{2(1-2u\alpha_{t+1}^{*})}\right)h_{t+1}^{*}\right],\label{eq:mgf_R_h}
\end{align}
where we used Theorem \ref{thm:DynUnderQ}. Next, we will establish
the exponentially affine form,
\[
\mathbb{E}_{t}^{\mathbb{Q}}(\exp(s\sum_{\tau=t+1}^{T}R_{\tau}))=\exp(A_{T}(s,M)+B_{T}(s,M)h_{t+1}^{*})
\]
 where $M=T-t$. This is proven by backwards induction. First, for
$t=T-1$ we have 
\[
\mathbb{E}_{T-1}^{\mathbb{Q}}(\exp(sR_{T}))=\exp[sr+\tfrac{1}{2}(s^{2}-s)h_{T}^{\ast}].
\]
So $A_{T}(s,1)=sr$ and $B_{T}(s,1)=\tfrac{1}{2}(s^{2}-s)$ define
the initial values for $A_{T}$ and $B_{T}$. Next, we establish the
recursions for $A_{T}$ and $B_{T}$, by showing that the exponentially
affine form holds for $t=T-m$, whenever it holds for $t=T-(m-1)$.
Thus, consider
\begin{eqnarray*}
\mathbb{E}_{T-m}^{\mathbb{Q}}\left[\exp\left(s\sum_{\tau=T-m+1}^{T}R_{\tau}\right)\right] & = & \mathbb{E}_{T-m}^{\mathbb{Q}}\left[\exp\left(sR_{T-m+1}+s\sum_{\tau=T-m+2}^{T}R_{\tau}\right)\right]\\
 & = & \mathbb{E}_{T-m}^{\mathbb{Q}}\left[\mathbb{E}_{T-m+1}^{\mathbb{Q}}\left[\exp\left(sR_{T-m+1}+s\sum_{\tau=T-m+2}^{T}R_{\tau}\right)\right]\right]\\
 & = & \mathbb{E}_{T-m}^{\mathbb{Q}}\left[\exp\left(sR_{T-m+1}\right)\mathbb{E}_{T-m+1}^{\mathbb{Q}}\left(s\sum_{\tau=T-m+2}^{T}R_{\tau}\right)\right]\\
 & = & \mathbb{E}_{T-m}^{\mathbb{Q}}\left[\exp\left(sR_{T-m+1}+A_{T}(s,m-1)+B_{T}(s,m-1)h_{T-m+2}^{*}\right)\right].
\end{eqnarray*}
Now substitute $B_{T}(s,m-1)$ for $u$ in (\ref{eq:mgf_R_h}) with
$t=T-m$. The exponentially affine form now follows if we set
\begin{align*}
A_{T}(s,m) & =A_{T}(s,m-1)+sr+B_{T}(s,m-1)\omega_{T-m+1}^{*}-\tfrac{1}{2}\log(1-2\alpha_{T-m+1}^{*}B_{T}(s,m-1)),\\
B_{T}(s,m) & =s(\gamma_{T-m}^{*}-\tfrac{1}{2})-\tfrac{1}{2}\gamma_{T-m}^{*2}+\beta_{T-m+1}^{*}B_{T}(s,m-1)+\tfrac{(s-\gamma_{T-m}^{*})^{2}}{2(1-2\alpha_{T-m+1}^{*}B_{T}(s,m-1))}.
\end{align*}
There are some interesting differences between the new expressions
and the original expressions. Our expressions for $A_{T}(s,m)$ and
$B_{T}(s,m)$ depend $s$ (the argument of the MGF) and $m$ (days
to maturity), but unlike the original expressions it also depends
on $T$ (the maturity date). The reason is that the coefficients,
$\omega^{\ast}$, $\alpha^{\ast}$, $\beta^{\ast}$, and $\gamma^{\ast}$,
are time-varying under $\mathbb{Q}$ due to their relation to $\eta_{t}$.
\end{proof}

\subsection{Proof of Theorem \ref{thm:OptionPricingRandom}\label{sec:ProofOPrandom}}

Let $\boldsymbol{\eta}_{t,M}=\left(\log\eta_{t+1},\ldots,\log\eta_{t+M}\right)^{\prime}$
and $\bar{\boldsymbol{\eta}}_{t,M}^{e}=\mathbb{E}_{t}^{\mathbb{Q}}(\boldsymbol{\eta}_{t,M})$
which embeds the expected path of $\log\eta_{t}$. The MGF for $M$-period
cumulative returns $R_{t,M}$ can be expressed as 
\[
g_{t,M}(s)\equiv\mathbb{E}_{t}^{\mathbb{Q}}[\exp(sR_{t,M})]=\mathbb{E}_{t}^{\mathbb{Q}}\{\mathbb{E}_{t}^{\mathbb{Q}}[\exp(sR_{t,T})\mid\boldsymbol{\eta}_{t,M}]\}=\mathbb{E}_{t}^{\mathbb{Q}}[f(\boldsymbol{\eta}_{t,M})],
\]
where $\ensuremath{f(\boldsymbol{\eta}_{t,M})\equiv\mathbb{E}_{t}^{\mathbb{Q}}[\exp(sR_{t,T})|\boldsymbol{\eta}_{t,M}]}$.
A second-order Taylor expansion of $f(\boldsymbol{\eta}_{t,M})$ around
$\bar{\boldsymbol{\eta}}_{t,M}^{e}$ gives us,
\begin{eqnarray*}
f(\boldsymbol{\eta}_{t,M}) & \approx & f(\bar{\boldsymbol{\eta}}_{t,M}^{e})+\nabla^{\prime}(\boldsymbol{\eta}_{t,M}-\bar{\boldsymbol{\eta}}_{t,M}^{e})+\tfrac{1}{2}(\boldsymbol{\eta}_{t,M}-\bar{\boldsymbol{\eta}}_{t,M}^{e})^{\prime}\tilde{H}_{t,M}(s)(\boldsymbol{\eta}_{t,M}-\bar{\boldsymbol{\eta}}_{t,M}^{e})\\
 & = & f(\bar{\boldsymbol{\eta}}_{t,M}^{e})+\nabla^{\prime}(\boldsymbol{\eta}_{t,M}-\bar{\boldsymbol{\eta}}_{t,M}^{e})+\tfrac{1}{2}\mathrm{tr}\left\{ \tilde{H}_{t,M}(s)\boldsymbol{\varepsilon}_{t,M}\boldsymbol{\varepsilon}_{t,M}^{\prime}\right\} ,
\end{eqnarray*}
where
\[
\tilde{H}_{t,M}(s)=\left.\frac{\partial^{2}f\left(\boldsymbol{\eta}_{t,M}\right)}{\partial\boldsymbol{\eta}_{t,M}\partial\boldsymbol{\eta}_{t,M}^{\prime}}\right|_{\boldsymbol{\eta}_{t,M}=\bar{\boldsymbol{\eta}}_{t,M}^{e}}
\]
is the Hessian, $\nabla$ is the Jacobian, and $\bm{\varepsilon}_{t,M}=\boldsymbol{\eta}_{t,M}-\bar{\boldsymbol{\eta}}_{t,M}^{e}$.
Since $\bar{\boldsymbol{\eta}}_{t,M}^{e}\in\mathcal{G}_{t}$ and $f(\bar{\boldsymbol{\eta}}_{t,M}^{e})=g_{t,M}(s|\bar{\boldsymbol{\eta}}_{t,M}^{e})$,
by taking conditional expectations on both sides, we arrive 
\[
\ensuremath{g_{t,M}(s)}\approx g_{t,M}(s|\bar{\boldsymbol{\eta}}_{t,M}^{e})+\frac{1}{2}\mathrm{tr}\left\{ \tilde{H}_{t,M}(s)\Sigma_{M}\right\} ,
\]
where $\Sigma_{M}=\mathbb{E}_{t}^{\mathbb{Q}}[\bm{\varepsilon}_{t,M}\bm{\varepsilon}_{t,M}^{\prime}]$.
From the AR(1) structure it follows that the $m$-th element of $\bm{\varepsilon}_{t,M}$
is given by, $[\bm{\varepsilon}_{t,M}]_{m}=\log\eta_{t+m}-\mathbb{E}_{t}^{\mathbb{Q}}\log\eta_{t+m}=\sum_{j=0}^{m-1}\varphi^{j}\varepsilon_{t+m-j}$,
$m=1,\ldots,M$, such that
\begin{eqnarray*}
\Sigma_{M} & = & {\rm var}_{t}^{\mathbb{Q}}\left[\boldsymbol{\eta}_{t,M}-\bar{\boldsymbol{\eta}}_{t,M}^{e}\right]={\rm var}_{t}^{\mathbb{Q}}\left(A_{M}\text{\ensuremath{\bm{\varepsilon}_{t,M}}}\right)=\sigma^{2}A_{M}A_{M}^{\prime},
\end{eqnarray*}
where 
\[
[A_{M}]_{i,j}=\begin{cases}
\varphi^{i-j} & \text{for }i\geq j\\
0 & \text{otherwise}.
\end{cases}
\]

The results above is for the AR(1) case. With an ARMA structure, $\varphi(L)\left(\log\eta_{t}-\zeta\right)=\theta(L)\varepsilon_{t},$
we have
\[
\log\eta_{t}-\zeta=\frac{\theta(L)}{\varphi(L)}\varepsilon_{t}=\sum_{j=0}^{\infty}\psi_{j}\varepsilon_{t-j},
\]
such that 
\[
[\bm{\varepsilon}_{t,M}]_{m}=\log\eta_{t+m}-\mathbb{E}\log\eta_{t+m}=\sum_{j=0}^{m-1}\psi_{j}\varepsilon_{t+m-j}=\sum_{j=1}^{m}\psi_{m-j}\varepsilon_{t+j}.
\]
So in the more general ARMA case we have
\[
[A_{M}]_{i,j}=\begin{cases}
\psi_{i-j} & \text{for }i\geq j\\
0 & \text{otherwise}.
\end{cases}
\]

After obtaining the analytical approximated MGF, the closed-form formula
for the price of a European call option can be simply adapted from
\citet{HestonNandi2000}.

\subsubsection{The Hessian Matrix $\tilde{H}_{t,M}(s)$ \label{subsec:Deriving-the-Hessian}}

First, to simplify the notations, we suppress the dependence on $T$
and $s$ for most terms and write,
\begin{eqnarray*}
\omega_{m} & \equiv & \omega_{T-m}^{*}=\omega\eta_{T-m},\\
\alpha_{m} & \equiv & \alpha_{T-m}^{*}=\alpha\eta_{T-m}\eta_{T-m-1},\\
\beta_{m} & \equiv & \beta_{T-m}^{*}=\beta\eta_{T-m}/\eta_{T-m-1},\\
\gamma_{m} & \equiv & \gamma_{T-m-1}^{*}=(\gamma+\lambda-\tfrac{1}{2})/\eta_{T-m-1}+\tfrac{1}{2},
\end{eqnarray*}
and $A(M)\equiv A_{T}(s,M),$ $B(M)\equiv B_{T}(s,M)$. For derivatives
of a variable $Y$ we write
\[
Y^{(i)}\equiv\frac{\partial Y}{\partial\log\eta_{t+i}}\quad\text{and}\quad Y^{(i,j)}\equiv\frac{\partial^{2}Y}{\partial\log\eta_{t+i}\partial\log\eta_{t+j}}.
\]
From Theorem \ref{lem:Closed-form-OP}, we have $f(\boldsymbol{\eta}_{t,M})=\exp\left[A(M)+B(M)h_{t+1}^{*}\right]$
with first and second derivatives given by
\[
f^{(i)}=f(\boldsymbol{\eta}_{t,M})[A^{(i)}(M)+B^{(i)}(M)h_{t+1}^{*}],
\]
and
\begin{eqnarray*}
f^{(i,j)} & = & f(\boldsymbol{\eta}_{t,M})[A^{(i)}(M)+B^{(i)}(M)h_{t+1}^{*}][A^{(j)}(M)+B^{(j)}(M)h_{t+1}^{*}]\\
 &  & +f(\boldsymbol{\eta}_{t,M})[A^{(i,j)}(M)+B^{(i,j)}(M)h_{t+1}^{*}],
\end{eqnarray*}
respectively. From the expression for $A(M)$ and $B(M)$ in Theorem
\ref{lem:Closed-form-OP}, we find (for the first derivatives)
\begin{align}
A^{(i)}(m+1) & =A^{(i)}(m)+\omega_{m}^{(i)}B(m)+\omega_{m}B^{(i)}(m)+\tfrac{{\alpha}_{m}^{(i)}B(m)+\alpha_{m}B^{(i)}(m)}{1-2\alpha_{m}B(m)},\label{eq:A(i)Recursive}\\
B^{(i)}(m+1) & =(s-\gamma_{m}){\gamma}_{m}^{(i)}+\beta_{m}B^{(i)}(m)+{\beta}_{m}^{(i)}B(m)\nonumber \\
 & \quad+(s-\gamma_{m})^{2}\tfrac{{\alpha}_{m}^{(i)}B(m)+\alpha_{m}B^{(i)}(m)}{\left[1-2\alpha_{m}B(m)\right]^{2}}+(\gamma_{m}-s)\tfrac{{\gamma}_{m}^{(i)}}{\left[1-2\alpha_{m}B(m)\right]},\label{eq:B(i)Recursive}
\end{align}
with $A^{(i)}(1)=B^{(i)}(1)=0$. Similarly, for the second derivatives
we find
\begin{eqnarray*}
A^{(i,j)}(m+1) & = & A^{(i,j)}(m)+{\omega}_{m}^{(i,j)}B(m)+{\omega}_{m}^{(i)}B^{(j)}(m)+B^{(i)}(m){\omega}_{m}^{(j)}+\omega_{m}B^{(i,j)}(m)\\
 &  & +\tfrac{1}{\left[1-2\alpha_{m}B(m)\right]}\left({\alpha}_{m}^{(i,j)}B(m)+{\alpha}_{m}^{(i)}B^{(j)}(m)+B^{(i)}(m){\alpha}_{m}^{(j)}+\alpha_{m}B^{(i,j)}(m)\right)\\
 &  & +\tfrac{2}{\left[1-2\alpha_{m}B(m)\right]^{2}}\left({\alpha}_{m}^{(i)}B(m)+\alpha_{m}B^{(i)}(m)\right)\left({\alpha}_{m}^{(j)}B(m)+\alpha_{m}B^{(j)}(m)\right),\\
{B}^{(i,j)}(m+1) & = & {\beta}_{m}^{(i,j)}B(m)+{\beta}_{m}^{(i)}{B}^{(j)}(m)+B^{(i)}(m){\beta}_{m}^{(j)}+\beta_{m}{B}^{(i,j)}(m)\\
 &  & +\tfrac{2\alpha_{m}B(m)}{\left[1-2\alpha_{m}B(m)\right]}\left[{\gamma}_{m}^{(i)}{\gamma}_{m}^{(j)}+(\gamma_{m}-s){\gamma}_{m}^{(i,j)}\right]\\
 &  & +\tfrac{2(\gamma_{m}-s)}{\left[1-2\alpha_{m}B(m)\right]^{2}}\left[\left({\alpha}_{m}^{(i)}B(m)+\alpha_{m}B^{(i)}(m)\right){\gamma}_{m}^{(j)}+{\gamma}_{m}^{(i)}\left({\alpha}_{m}^{(j)}B(m)+\alpha_{m}B^{(j)}(m)\right)\right]\\
 &  & +\tfrac{(s-\gamma_{m})^{2}}{\left[1-2\alpha_{m}B(m)\right]^{2}}\left({\alpha}_{m}^{(i,j)}B(m)+{\alpha}_{m}^{(i)}B^{(j)}(m)+B^{(i)}(m){\alpha}_{m}^{(j)}+\alpha_{m}{B}^{(i,j)}(m)\right)\\
 &  & +\tfrac{4(s-\gamma_{m})^{2}}{\left[1-2\alpha_{m}B(m)\right]^{3}}\left({\alpha}_{m}^{(i)}B(m)+\alpha_{m}B^{(i)}(m)\right)\left({\alpha}_{m}^{(j)}B(m)+\alpha_{m}B^{(j)}(m)\right),
\end{eqnarray*}
with $A^{(i,j)}(1)=B^{(i,j)}(1)=0$ and
\[
\omega_{m}^{(i)}=\begin{cases}
\omega_{m} & \text{for }i=M-m,\\
0 & \text{otherwise},
\end{cases}\quad\alpha_{m}^{(i)}=\begin{cases}
\alpha_{m} & \text{for }i=M-m,M-m-1\\
0 & \text{otherwise,}
\end{cases},
\]
\[
\beta_{m}^{(i)}=\begin{cases}
\text{\ensuremath{}}\beta_{m} & \text{for }i=M-m\\
-\beta_{m} & \text{for }i=M-m-1\\
0 & \text{otherwise,}
\end{cases},\quad\gamma_{m}^{(i)}=\begin{cases}
\tfrac{1}{2}-\gamma_{m} & \text{for }i=M-m-1\\
0 & \text{otherwise,}
\end{cases},
\]
\[
\omega_{m}^{(i,j)}=\begin{cases}
\omega_{m} & \text{for }i=j=M-m\\
0 & \text{otherwise,}
\end{cases},\quad\alpha_{m}^{(i,j)}=\begin{cases}
\alpha_{m} & \text{for }i,j=M-m,M-m-1\\
0 & \text{otherwise,}
\end{cases},
\]
\[
\beta_{m}^{(i,j)}=\begin{cases}
\text{\ensuremath{}}\beta_{m} & \text{for }i=j\in\{M-m,M-m-1\}\\
-\beta_{m} & \text{for }i\neq j\in\{M-m,M-m-1\}\\
0 & \text{otherwise,}
\end{cases},\quad\gamma_{m}^{(i,j)}=\begin{cases}
\gamma_{m}-\tfrac{1}{2} & \text{for }i=j=M-m-1\\
0 & \text{otherwise.}
\end{cases}
\]
Finally, the scaled Hessian matrix presented in Theorem \ref{thm:OptionPricingRandom}
is defined by $H_{t,M}(s)=\frac{1}{g_{t,M}(s|\bar{\boldsymbol{\eta}}_{t,M}^{e})}\tilde{H}_{t,M}(s)$.
We adopt his formulation because it simplifies the expression.

\subsection{Proof of Theorem \ref{thm:moments}\label{sec:Proof_moment}}

Next we prove Theorem \ref{thm:moments} that enables us to computing
moments from the moment-generating function The $k$-th moment of
cumulative returns can be expressed as
\begin{eqnarray*}
\mathbb{E}\left(R^{k}\right) & = & \int_{-\infty}^{+\infty}R^{k}f(R)\mathrm{d}R=\underbrace{\int_{-\infty}^{0}R^{k}f(R)\mathrm{d}R}_{(*)}+\underbrace{\int_{0}^{+\infty}R^{k}f(R)\mathrm{d}R}_{(**)},
\end{eqnarray*}
where $f(R)$ is the conditional density of cumulated returns under
$\mathbb{Q}$. By the Fourier inverse transform (FIT), $f(R)$ can
be expressed as
\[
f(R)=\frac{1}{\pi}\int_{0}^{\infty}\mathrm{Re}\left[e^{-uR}\phi(u)\right]\mathrm{d}u_{I},\quad\phi(u)\equiv\mathbb{E}\left[e^{uR}\right],
\]
where $u=u_{R}+iu_{I}\in\mathbb{C}$, with $u_{R}=\mathrm{Re}[u]$
being the real part of $u$, and $\phi\left(u\right)$ is the MGF
of cumulated returns. Note that different from traditional Fourier
inverse transform based on characteristic function, such form FIT
requires the existence of MGF. 

Then after changing the order of integration by Fubini's theorem,
we find
\begin{eqnarray*}
(*) & = & \frac{1}{\pi}\int_{0}^{\infty}\ensuremath{\mathrm{Re}\left[\left(\int_{-\infty}^{0}R^{k}e^{-uR}\mathrm{d}R\right)\phi(u)\right]\mathrm{d}u_{I}},\\
(**) & = & \frac{1}{\pi}\int_{0}^{\infty}\ensuremath{\mathrm{Re}\left[\left(\int_{0}^{+\infty}R^{k}e^{-vR}\mathrm{d}R\right)\phi(v)\right]\mathrm{d}v_{I}}.
\end{eqnarray*}
We have here used the following Laplace transformations for $R^{k}$:
\begin{eqnarray*}
\int_{-\infty}^{0}R^{k}e^{-uR}\mathrm{d}R & = & -\frac{k!}{u^{k+1}},\quad u_{R}<0,\\
\int_{0}^{+\infty}R^{k}e^{-vR}\mathrm{d}R & = & \frac{k!}{v^{k+1}},\quad\quad v_{R}>0.
\end{eqnarray*}
Thus, it follows that
\begin{eqnarray*}
(*) & = & \frac{1}{\pi}\int_{0}^{\infty}\ensuremath{\mathrm{Re}\left[-\frac{k!}{u^{k+1}}\phi(u)\right]\mathrm{d}u_{I},}\\
(**) & = & \frac{1}{\pi}\int_{0}^{\infty}\ensuremath{\mathrm{Re}\left[\frac{k!}{v^{k+1}}\phi(v)\right]\mathrm{d}v_{I}}.
\end{eqnarray*}
More general results for moment deduced by integrating the moment-generating
function can be found in \citet{HansenTong:2024FracMoments}.

\subsection{Proof of Theorem \ref{thm:ScoreProp}\label{sec:Proof_FormDisScore}}

We observe $N_{t}$ derivative prices at time $t$, resulting in the
vector of pricing error, $e_{t}\in\mathbb{R}^{N_{t}\times1}$. So,
the derivative with respect to $\log\eta_{t}$ is
\begin{align*}
\nabla_{t}\equiv\frac{\partial\ell(X_{t}|\mathcal{G}_{t})}{\partial\log\eta_{t}} & =-\frac{1}{2\sigma_{e}^{2}}\frac{\partial\left(e_{t}^{\prime}\Omega_{N_{t}}^{-1}e_{t}\right)}{\partial\log\eta_{t}}=-\frac{1}{2\sigma_{e}^{2}}\frac{\partial\left(e_{t}^{\prime}\Omega_{N_{t}}^{-1}e_{t}\right)}{\partial e_{t}^{\prime}}\frac{\partial e_{t}}{\partial\log\eta_{t}}\\
 & =-\frac{1}{\sigma_{e}^{2}}e_{t}^{\prime}\Omega_{N_{t}}^{-1}\frac{\partial e_{t}}{\partial\log\eta_{t}}=\frac{1}{\sigma_{e}^{2}}\left(\frac{\partial X_{t}^{m}}{\partial\log\eta_{t}}\right)^{\prime}\Omega_{N_{t}}^{-1}e_{t},
\end{align*}
such that
\begin{align*}
\mathbb{E}^{\mathbb{P}}[\nabla_{t}^{2}|\mathcal{G}_{t}] & =\frac{1}{\sigma_{e}^{2}}\left(\frac{\partial X_{t}^{m}}{\partial\log\eta_{t}}\right)^{\prime}\Omega_{N_{t}}^{-1}\left(\frac{\partial X_{t}^{m}}{\partial\log\eta_{t}}\right).
\end{align*}
In the special case with a single derivative price, $N_{t}=1$, our
expressions simplify to
\[
\nabla_{t}=\frac{1}{\sigma_{e}^{2}}\left(X_{t}-X_{t}^{m}\right)\frac{\partial X_{t}^{m}}{\partial\log\eta_{t}}\quad\text{and}\quad\mathbb{E}^{\mathbb{P}}[\nabla_{t}^{2}|\mathcal{G}_{t}]=\frac{1}{\sigma_{e}^{2}}\left(\frac{\partial X_{t}^{m}}{\partial\log\eta_{t}}\right)^{2},
\]
such that $s_{t}=\frac{1}{\sigma_{e}}\left(X_{t}-X_{t}^{m}\right)\operatorname{sign}\left(\frac{\partial X_{t}^{m}}{\partial\log\eta_{t}}\right).$

\subsubsection{The Distribution of the Score}

Define
\[
\psi_{t}^{\prime}=\frac{1}{\sigma_{e}^{2}}\left(\frac{\partial X_{t}^{m}}{\partial\log\eta_{t}}\right)^{\prime}\Omega_{N_{t}}^{-1},\quad\Sigma_{N_{t}}=\sigma_{e}^{2}\Omega_{N_{t}}.
\]
Recall that $\psi_{t}^{\prime}$ is $\mathcal{G}_{t}$measurable,
then under Assumption \ref{assu:Pricing-Errors} where $e_{t}|\mathcal{G}_{t}\sim iid\ N(0,\Sigma_{N_{t}})$,
the conditional MGF of score $s_{t}$ can be expressed as
\begin{align*}
\mathbb{E}_{t}^{\mathbb{P}}\left[\exp\left(us_{t}\right)\right] & =\mathbb{E}_{t}^{\mathbb{P}}\left[\exp\left(u\frac{\psi_{t}^{\prime}e_{t}}{\sqrt{\psi_{t}^{\prime}\Sigma_{N_{t}}\psi_{t}}}\right)\right]=\exp\left(\tfrac{1}{2}\frac{u\psi_{t}^{\prime}}{\sqrt{\psi_{t}^{\prime}\Sigma_{N_{t}}\psi_{t}}}\Sigma_{N_{t}}\frac{u\psi_{t}}{\sqrt{\psi_{t}^{\prime}\Sigma_{N_{t}}\psi_{t}}}\right)=\exp\left(\tfrac{1}{2}u^{2}\right).
\end{align*}
Therefore, the scaled score, $s_{t}$, is iid with a standard normal
distribution under $\mathbb{P}$.

To derive the distribution under $\mathbb{Q}$, we first obtain,
\begin{eqnarray*}
 &  & \mathbb{E}_{t}^{\mathbb{P}}\left[\exp\left(us_{t}+\phi_{t}R_{t+1}+\xi_{t}h_{t+2}\right)\right]\\
 & = & \mathbb{E}_{t}^{\mathbb{P}}\left\{ \mathbb{E}_{t}^{\mathbb{P}}\left[\exp\left(us_{t}+\phi_{t}R_{t+1}+\xi_{t}h_{t+2}\right)|\phi_{t}R_{t+1},\xi_{t}h_{t+2}\right]\right\} \\
 & = & \mathbb{E}_{t}^{\mathbb{P}}\left\{ \mathbb{E}_{t}^{\mathbb{P}}\left[\exp\left(\tfrac{1}{2}u^{2}+\phi_{t}R_{t+1}+\xi_{t}h_{t+2}\right)|\phi_{t}R_{t+1},\xi_{t}h_{t+2}\right]\right\} \\
 & = & \exp\left(\tfrac{1}{2}u^{2}\right)\mathbb{E}_{t}^{\mathbb{P}}\left[\exp\left(\phi_{t}R_{t+1}+\xi_{t}h_{t+2}\right)\right],
\end{eqnarray*}
where the last equality is a consequence of $e_{t}$ being assumed
to be independent of any return shock. Then, we directly have
\begin{align*}
\mathbb{E}_{t}^{\mathbb{Q}}\left(\exp\left(us_{t}\right)\right) & =\frac{\mathbb{E}_{t}^{\mathbb{P}}\left[\exp\left(us_{t}+\phi_{t}R_{t+1}+\xi_{t}h_{t+2}\right)\right]}{\mathbb{E}_{t}^{\mathbb{P}}[\exp(\phi_{t}R_{t+1}+\xi_{t}h_{t+2})]}=\exp\left(\tfrac{1}{2}u^{2}\right),
\end{align*}
which shows that the score, $s_{t}$, is also iid and distributed
as a standard normal under $\mathbb{Q}$.

\subsection{The Scores for Derivatives\label{subsec:ScoresDeri}}

\subsubsection{Score for VIX\label{subsec:Score-for-VIX}}

The simplest case to drive the score for is $X_{t}=\log{\rm VIX}_{t}$,
where $N_{t}=1$. The log-likelihood function for $\ell({\rm VIX}_{t}|\mathcal{G}_{t})$
is here given by
\[
\ell({\rm VIX}_{t}|\mathcal{G}_{t})=-\frac{1}{2\sigma_{e}^{2}}\left(\log\mathrm{VIX}_{t}-\log\mathrm{VIX}_{t}^{m}\right)^{2},
\]
and its first derivative is,
\[
\nabla_{t}=\frac{1}{\sigma_{e}^{2}}\left(\log\mathrm{VIX}_{t}-\log\mathrm{VIX}_{t}^{m}\right)\left(\frac{\partial\log\mathrm{VIX}_{t}^{m}}{\partial\log\eta_{t}}\right).
\]
The conditional variance of $\nabla_{t}$ is
\begin{eqnarray*}
\mathbb{E}^{\mathbb{P}}[\nabla_{t}^{2}|\mathcal{G}_{t}] & = & \frac{1}{\sigma_{e}^{2}}\left(\frac{\partial\log\mathrm{VIX}_{t}^{m}}{\partial\log\eta_{t}}\right)^{2},
\end{eqnarray*}
and it follows that
\[
s_{t}=\frac{1}{\sigma_{e}}\left(\log\mathrm{VIX}_{t}-\log\mathrm{VIX}_{t}^{m}\right){\rm sign}\left(\frac{\partial\log\mathrm{VIX}_{t}^{m}}{\partial\log\eta_{t}}\right).
\]
For the last term, we recall the expression, $\mathrm{VIX}_{t}^{m}=\sqrt{a_{1}\left(M,\sigma^{2}\right)+a_{2}\left(M,\sigma^{2}\right)h_{t+1}^{*}}$,
which leads to

\[
\frac{\partial\log\mathrm{VIX}_{t}^{m}}{\partial\log\eta_{t}}=\frac{1}{2\left(\mathrm{VIX}_{t}^{m}\right)^{2}}\left[\frac{\partial a_{1,t}}{\partial\log\eta_{t}}+\left(\frac{\partial a_{2,t}}{\partial\log\eta_{t}}+a_{2,t}\right)h_{t+1}^{*}\right]
\]
 with
\begin{align*}
\frac{\partial a_{1}(M,\sigma^{2})}{\partial\log\eta_{t}} & =\frac{A^{2}}{M}\sum_{k=2}^{M}\sum_{i=2}^{k}\tilde{\beta}^{k-i}\left[\omega\mathbb{E}_{t}^{\mathbb{Q}}(\eta_{t+k-1})B_{k-1,1}+\alpha\mathbb{E}_{t}^{\mathbb{Q}}\left(\eta_{t+k-1}\eta_{t+i-2}\right)\left(B_{k-1,1}+B_{i-2,1}\right)\right]
\end{align*}
and
\[
\frac{\partial a_{2}(M,\sigma^{2})}{\partial\log\eta_{t}}+a_{2}(M,\sigma^{2})=\frac{A^{2}}{M}\sum_{k=1}^{M}\tilde{\beta}^{k-1}\mathbb{E}_{t}^{\mathbb{Q}}\left(\frac{\eta_{t+k-1}}{\eta_{t}}\right)B_{k-1,1}
\]
Therefore, according to recursive formula for $B_{k,i}$ in (\ref{eq:logARMA}),
they are all positive if all $\varphi_{i}$ are positive for $i=1,\ldots,p$.
The term, ${\rm sign}\left(\frac{\partial\log\mathrm{VIX}_{t}^{m}}{\partial\log\eta_{t}}\right)$,
is therefore redundant and we arrived at the simple expression,
\[
s_{t}=\frac{1}{\sigma_{e}}\left(\log\mathrm{VIX}_{t}-\log\mathrm{VIX}_{t}^{m}\right).
\]
By Assumption \ref{assu:Pricing-Errors} it follows that $s_{t}$
is iid standard normally distributed.

\subsubsection{Score for Option Prices\label{subsec:Score-for-Options}}

When the model is estimated with option prices, the elements of $X_{t}\in\mathbb{R}^{N_{t}\times1}$
are given by $\log\left({\rm IV}_{bs}\left(C_{t},S_{t},M,K,r\right)\right)$,
for the $N_{t}$ options selected at time $t$. The logarithm of Black-Scholes
implied volatility makes this variable comparable to $\log\mathrm{VIX}_{t}$,
which is used when the model is estimated with VIX data.

Let $X_{t}^{m}$ denote the logarithm of model-based implied volatility
for a particular option. We seek an expression for 
\[
\frac{\partial X_{t}^{m}}{\partial\log\eta_{t}}=\frac{\partial\log\left({\rm IV}_{bs}^{m}\right)}{\partial{\rm IV}_{bs}^{m}}\frac{\partial{\rm IV}_{bs}^{m}}{\partial C_{t}^{m}}\frac{\partial C_{t}^{m}}{\partial\log\eta_{t}}=\frac{1}{{\rm IV}_{bs}^{m}\cdot{\rm Vega}_{bs}^{m}}\frac{\partial C_{t}^{m}}{\partial\log\eta_{t}},
\]
where 
\[
{\rm Vega}_{bs}=\frac{1}{\sqrt{2\pi}}\exp\left(-\frac{1}{2}d_{t}^{2}\right)S_{t}\sqrt{M},
\]
is the Black-Scholes implied Vega and
\[
d_{t}=\frac{1}{{\rm IV}_{bs}\sqrt{M}}\left[\log\left(\frac{S_{t}}{K}\right)+\left(r+\frac{{\rm IV}_{bs}^{2}}{2}\right)M\right].
\]
For $\partial\hat{C}^{m}/\partial\log\eta_{t}$, we find
\[
\frac{\partial\hat{C}^{m}(S_{t},M,K,r;h_{t+1}^{*})}{\partial\log\eta_{t}}=S_{t}\frac{\partial P_{1,t}}{\partial\log\eta_{t}}-K\exp(-rM)\frac{\partial P_{2,t}}{\partial\log\eta_{t}},
\]
where
\begin{align*}
\frac{\partial P_{1,t}}{\partial\log\eta_{t}} & =\text{\ensuremath{\frac{\exp(-rM)}{\pi}\int_{0}^{\infty}{\rm Re}\left[\frac{K^{-iu}}{iuS_{t}}\frac{\partial\hat{g}_{t,M}(iu+1)}{\partial\log\eta_{t}}\right]du},}\\
\frac{\partial P_{2,t}}{\partial\log\eta_{t}} & =\frac{1}{\pi}\int_{0}^{\infty}{\rm Re}\left[\frac{K^{-iu}}{iu}\frac{\partial\hat{g}_{t,M}(iu)}{\partial\log\eta_{t}}\right]du,
\end{align*}
with
\begin{align*}
\hat{g}_{t,M}(s) & =g_{t,M}(s|\bar{\boldsymbol{\eta}}_{t,M}^{e})\left[1+\frac{1}{2}{\rm tr}\left(H_{t,M}(s)\Sigma_{M}\right)\right].
\end{align*}
Using the simplified notations in Section \ref{subsec:Deriving-the-Hessian},
we have 
\begin{equation}
\frac{\partial\hat{g}_{t,M}(s)}{\partial\log\eta_{t}}=\hat{g}_{t,M}(s)\left(A^{(0)}(M)+\left(B^{(0)}(M)+B(M)\right)h_{t+1}^{*}\right)+\frac{1}{2}g_{t,M}\left(s|\bar{\boldsymbol{\eta}}_{t,M}^{e}\right){\rm tr}\left(\dot{H}_{t,M}(s)\Sigma_{M}\right),\label{eq:2ndorder}
\end{equation}
which is the expression we use to price options with the second-order
approximation. Pricing options with first-order approximation, amounts
to dropping the last term in (\ref{eq:2ndorder}). With a first-order
approximation, we simply need to compute $A^{(0)}(m)$ and $B^{(0)}(m)$,
that are given recursively from (\ref{eq:A(i)Recursive}) and (\ref{eq:B(i)Recursive}),
with initial condition $\dot{A}(1)=\dot{B}(1)=0$, and use
\begin{eqnarray*}
\omega_{m} & = & \omega\Psi_{t}\left(M-m+1\right),\\
\alpha_{m} & = & \alpha\Psi_{t}\left(M-m+1\right)\Psi_{t}\left(M-m\right),\\
\beta_{m} & = & \beta\Psi_{t}\left(M-m+1\right)/\Psi_{t}\left(M-m\right),\\
\gamma_{m} & = & (\gamma+\lambda-\tfrac{1}{2})/\Psi_{t}\left(M-m\right)+\tfrac{1}{2},
\end{eqnarray*}
where $\Psi_{t}\left(k\right)$ is defined as
\[
\Psi_{t}\left(k\right)=\exp\left\{ \mathbb{E}_{t}\left(\log\eta_{t+k-1}\right)\right\} =\exp\left\{ \left(1-\varphi^{k-1}\right)\zeta+\varphi^{k-1}\log\eta_{t}\right\} .
\]
The corresponding first derivatives (with respect to $\log\eta_{t}$)
are given by
\begin{eqnarray*}
\dot{\omega}_{m} & = & \omega_{m}\varphi^{M-m},\\
\dot{\alpha}_{m} & = & \alpha_{m}\left(\varphi^{M-m}+\varphi^{M-m-1}\right),\\
\dot{\beta}_{m} & = & \beta_{m}\left(\varphi^{M-m}-\varphi^{M-m-1}\right),\\
\dot{\gamma}_{m} & = & (\tfrac{1}{2}-\gamma_{m})\varphi^{M-m-1}.
\end{eqnarray*}

For the second-order approximation in (\ref{eq:2ndorder}), we need
to evaluate $\dot{H}_{t,M}(s)=\partial H_{t,M}(s)/\partial\log\eta_{t}$,
for which we have
\begin{eqnarray*}
\dot{H}_{i,j} & = & \left(A^{(i,0)}(M)+B^{(i,0)}(M)h_{t+1}^{*}\right)\left(A^{(j)}(M)+B^{(j)}(M)h_{t+1}^{*}\right)\\
 &  & +\left(A^{(i)}(M)+B^{(i)}(M)h_{t+1}^{*}\right)\left(\dot{A}^{(j)}(M)+\dot{B}^{(j)}(M)h_{t+1}^{*}\right)\\
 &  & +\left(A^{(i,j,0)}(M)+B^{(i,j,0)}(M)h_{t+1}^{*}\right).
\end{eqnarray*}
The expressions for the third derivatives, $A^{(i,j,0)}(m)=\partial A^{(i,j,0)}(m)/\partial\log\eta_{t}$
and $B^{(i,j,0)}(m)=\partial A^{(i,j,0)}(m)/\partial\log\eta_{t}$
are relatively cumbersome, and are omitted to conserve space. 

We now have all the terms needed for the analytical formula for $\partial C^{m}/\partial\log\eta_{t}$.
In practice, however, it may be more convenient to compute this derivative
with numerical methods, i.e.
\[
\frac{\partial C_{t}^{m}}{\partial\log\eta_{t}}\approx\frac{1}{\Delta}\left[C_{t}^{m}\left(\log\eta_{t}+\Delta\right)-C_{t}^{m}\left(\log\eta_{t}\right)\right],
\]
for a very small value of $\Delta$, because this simplifies the calculation
of score with respect to option prices, and is very fast. 

When the number of daily available option prices, $N_{t}$, is large,
computing the derivatives $\partial\log\text{IV}_{t,i}^{m}/\partial\log\eta_{t}$
for $i=1,2,\ldots,N_{t}$ becomes very time-consuming. In this case,
one could assume that $\partial\log\text{IV}_{t,i}^{m}/\partial\log\eta_{t}\approx\chi_{t}>0$
(which is only $t$-dependent), such that the scaled score $s_{t}$,
defined by (\ref{eq:ScaledScore}) and Theorem \ref{thm:ScoreProp},
is approximated by
\[
s_{t}^{\text{appro}}=\frac{1}{\sigma_{e}}\frac{\iota_{N_{t}}^{\prime}\Omega_{N_{t}}^{-1}e_{t}}{\sqrt{\iota_{N_{t}}^{\prime}\Omega_{N_{t}}^{-1}\iota_{N_{t}}}}
\]
where $\iota_{N_{t}}$ is a $N_{t}\times1$ vector of ones. With this
approximation, there is no need to compute any derivatives. In our
empirical analysis, we find the correlation between $s_{t}^{\text{appro}}$
and the true $s_{t}$ reaches up to 99.11\%, such that the correlation
between $\eta_{t}^{\text{appro}}$ and the true $\eta_{t}$ is 99.98\%. 

\subsection{Computation of Option Prices by Numerical Integrations\label{subsec:NumeCom}}

Our expression for option prices in Theorem \ref{lem:Closed-form-OP}
requires the evaluation of integrals of the type, $\int_{0}^{\infty}f(x)dx$,
for the computation of $P_{1}(t)$ and $P_{2}(t)$. For this, we use
the Gauss-Laguerre quadrature method, which is an extension of the
Gaussian quadrature method for approximating the value of integrals
of the following kind: 
\[
\int_{0}^{\infty}e^{-x}g(x)dx\approx\sum_{i=1}^{L}w(x_{i})g(x_{i}).
\]
The Gauss-Laguerre quadrature method is also based the abscissae $x_{1},\ldots,x_{L}$
and their associated weights $w(x_{i})$ with $i=1,...,L$, leading
to the following numerical approximation,
\[
\int_{0}^{\infty}f(x)dx=\int_{0}^{\infty}e^{-x}[e^{x}f(x)]dx\approx\sum_{i=1}^{L}w(x_{i})e^{x_{i}}f(x_{i}).
\]
In our implementation we use $L=32$.\newpage{}

\setcounter{equation}{0}\renewcommand{\theequation}{S.\arabic{equation}}
\setcounter{subsection}{0}\renewcommand{\thesubsection}{S.\arabic{subsection}}
\setcounter{section}{18}
\setcounter{figure}{0}\renewcommand{\thefigure}{S.\arabic{figure}}
\setcounter{table}{0}\renewcommand{\thetable}{S.\arabic{table}}
\setcounter{page}{1}\renewcommand{\thepage}{S.\arabic{page}}

\section*{Online Appendix\label{app:SupplementaryResults}}

\subsection{Alternative RMSE Measure\label{subsec:RMSEIV}}

An alternative metric for evaluation the derivative pricing is to
compute the RMSE in volatility levels. For the VIX this is defined
by
\begin{equation}
{\rm RMSE_{VIX}}=\sqrt{\frac{1}{T}\sum_{t=1}^{T}\text{\ensuremath{\left[\mathrm{VIX}_{t}^{m}-\mathrm{VIX}_{t}\right]^{2}}}},\label{eq:RMSE_VIX_Level}
\end{equation}
and for option prices it is defined by
\begin{equation}
{\rm RMSE_{IV}}=\sqrt{\frac{1}{\sum_{t=1}^{T}N_{t}}\sum_{t=1}^{T}\text{\ensuremath{\sum_{i=1}^{N_{t}}\left[\mathrm{IV}_{t,i}^{m}-\mathrm{IV}_{t,i}\right]^{2}}}}\times100.\label{eq:RMSE_IV_Level}
\end{equation}
The results for the RMSE in volatility levels are reported in Tables
\ref{tab:DerivativePricing2} and \ref{tab:OutOfSample2}, where the
former is analogous to Table \ref{tab:DerivativePricing} and the
latter is analogues to the out-of-sample results in Table \ref{tab:OutOfSample}.
The results are qualitatively very similar to the RMSE for logarithmically
transformed volatilities in Tables \ref{tab:DerivativePricing} and
\ref{tab:OutOfSample}.

\begin{table}
\caption{VIX and Option Pricing Performance (RMSE)}

\begin{centering}
\vspace{0.2cm}
\begin{footnotesize}%
\begin{tabularx}{\textwidth}{p{2.5cm}YYYYYYYY}
\toprule
\midrule
    Model & CHNG   & CHNG      & DHNG  & DHNG   \\
      & [VIX] & [Opt]  & [VIX] & [Opt]  \\
    \midrule
          &       &       &       &        \\
    \multicolumn{5}{l}{\it A: RMSE for VIX Pricing}\\[6pt]
    Full Sample   & 4.117 & 4.143 & 1.024 & 1.795 \\
          &       &       &       &         \\
    \multicolumn{5}{l}{\it B: RMSE for Option Pricing} \\[6pt]
    Full Sample   & 4.277 & 4.151 & 2.327 & 1.746 \\
          &       &       &       &         \\
    \multicolumn{5}{l}{\it Partitioned by moneyness } \\[3pt]
   \ Delta<0.3      & 4.500 & 4.149 & 3.108 & 1.947 \\
   \ 0.3$\leq$Delta<0.4      & 4.069 & 3.985 & 2.830 & 1.859 \\
   \ 0.4$\leq$Delta<0.5     & 3.743 & 3.862 & 2.415 & 1.639 \\
   \ 0.5$\leq$Delta<0.6      & 3.798 & 3.999 & 2.033 & 1.444 \\
   \ 0.6$\leq$Delta<0.7      & 4.017 & 4.062 & 1.876 & 1.455 \\
   \ 0.7$\leq$Delta    & 4.771 & 4.417 & 2.098 & 1.949 \\
           &       &       &       &        \\
    \multicolumn{5}{l}{\it Partitioned by maturity} \\[3pt]
    \ DTM<30 & 4.649 & 4.110 & 2.385 & 2.252 \\
    \ 30$\leq$DTM<60 & 4.235 & 4.058 & 2.061 & 1.475 \\
    \ 60$\leq$DTM<90 & 4.065 & 4.086 & 2.156 & 1.311 \\
    \ 90$\leq$DTM<120 & 4.059 & 4.234 & 2.492 & 1.598 \\
    \ 120$\leq$DTM<150 & 4.073 & 4.212 & 2.623 & 1.781 \\
    \ 150$\leq$DTM & 4.192 & 4.537 & 2.757 & 1.819 \\
          &       &       &       &       \\
    \multicolumn{5}{l}{\it Partitioned by the level of VIX} \\[3pt]
    \ VIX<15     & 3.726 & 3.727 & 1.821 & 1.155 \\
    \ 15$\leq$VIX<20     & 2.980 & 3.252 & 2.051 & 1.438 \\
    \ 20$\leq$VIX<25     & 3.711 & 3.808 & 2.265 & 1.718 \\
    \ 25$\leq$VIX<30     & 4.525 & 4.266 & 2.703 & 2.126 \\
    \ 30$\leq$VIX<35    & 5.460 & 5.227 & 3.071 & 2.444 \\
    \ 35$\leq$VIX     & 10.68 & 9.164 & 4.727 & 4.089 \\
\\[0.0cm]
\\[-0.5cm]
\midrule
\bottomrule
\end{tabularx}
\end{footnotesize}
\par\end{centering}
{\small Note: This table reports the in-sample VIX and option pricing
performance for each model in Table \ref{tab:JointEstimation}. We
evaluate the model's option pricing ability through the root of mean
square errors of implied volatility (${\rm RMSE_{IV}}$). We summarize
the results by option moneyness, maturity and market VIX level. Moneyness
is measured by Delta computed from the Black-Scholes model. DTM denotes
the number of calendar days to maturity. \label{tab:DerivativePricing2}}{\small\par}
\end{table}
\begin{table}
\caption{Out-of-sample VIX and Option Pricing (RMSE)}

\begin{centering}
\vspace{0.2cm}
\begin{footnotesize}%
\begin{tabularx}{\textwidth}{p{2.5cm}YYYYYYYY}
\toprule 
\midrule
    Model & CHNG   & CHNG      & DHNG  & DHNG   \\
      & [VIX] & [Opt]  & [VIX] & [Opt]  \\
    \midrule
          &       &       &       &        \\
    \multicolumn{5}{l}{\it A: RMSE for VIX Pricing}\\[6pt]
    Full Sample   & 5.083 & 4.660 & 1.278 & 2.067 \\
          &       &       &       &         \\
    \multicolumn{5}{l}{\it B: RMSE for Option Pricing} \\[6pt]
    Full Sample   & 4.710 & 4.043 & 2.724 & 1.932 \\
          &       &       &       &         \\
    \multicolumn{5}{l}{\it Partitioned by moneyness } \\[3pt]
   \ Delta<0.3      & 4.973 & 4.068 & 3.839 & 2.108 \\
   \ 0.3$\leq$Delta<0.4      & 4.622 & 3.910 & 3.467 & 2.018 \\
   \ 0.4$\leq$Delta<0.5     & 4.255 & 3.709 & 2.849 & 1.701 \\
   \ 0.5$\leq$Delta<0.6      & 4.339 & 3.907 & 2.281 & 1.508 \\
   \ 0.6$\leq$Delta<0.7      & 4.696 & 4.108 & 2.021 & 1.635 \\
   \ 0.7$\leq$Delta    & 4.934 & 4.213 & 2.202 & 2.195 \\
           &       &       &       &        \\
    \multicolumn{5}{l}{\it Partitioned by maturity} \\[3pt]
    \ DTM<30 & 5.276 & 4.243 & 2.858 & 2.301 \\
    \ 30$\leq$DTM<60 & 4.692 & 3.968 & 2.526 & 1.591 \\
    \ 60$\leq$DTM<90 & 4.217 & 3.794 & 2.547 & 1.524 \\
   \ 90$\leq$DTM<120 & 4.317 & 4.013 & 2.758 & 1.818 \\
   \ 120$\leq$DTM<150 & 4.133 & 3.870 & 2.941 & 2.157 \\
   \  150$\leq$DTM & 4.263 & 4.182 & 2.965 & 2.117 \\
          &       &       &       &       \\
    \multicolumn{5}{l}{\it Partitioned by the level of VIX} \\[3pt]
    \ VIX<15     & 4.062 & 3.121 & 2.003 & 1.257 \\
    \ 15$\leq$VIX<20     & 3.081 & 3.060 & 2.403 & 1.531 \\
    \ 20$\leq$VIX<25     & 3.026 & 3.537 & 2.615 & 1.714 \\
    \ 25$\leq$VIX<30     & 4.073 & 3.902 & 3.169 & 2.287 \\
    \ 30$\leq$VIX<35    & 5.352 & 4.856 & 3.688 & 2.719 \\
    \ 35$\leq$VIX     & 12.05 & 9.423 & 5.199 & 4.376 \\
\\[0.0cm]
\\[-0.5cm]
\midrule
\bottomrule
\end{tabularx}
\end{footnotesize}
\par\end{centering}
{\small Note: This table reports the out-of-sample option pricing performance
for each model. We conduct our out-of-sample performance evaluation
by splitting our original dataset into two subsamples: the in-sample
data consists of the years before 2008 and the out-of-sample consists
of the years 2008--2021, which spans a 14-year period. The estimation
for each model is done only once for the in-sample data, and then
price the out-of-sample data by the estimated parameters. Therefore,
the in-sample data is 1990--2007. We evaluate the model's option
pricing ability through the root of mean square errors of implied
volatility (IVRMSE). We summarize the results by option moneyness,
maturity and market VIX level. Moneyness is measured by Delta computed
from the Black-Scholes model. DTM denotes the number of calendar days
to maturity. \label{tab:OutOfSample2}}{\small\par}
\end{table}

\subsection{A Variance Risk Ratio based on Model-Free Realized Variances\label{subsec:EMVRR}}

As a robustness check, we compute an alternative monthly measure of
the VRR, $\tilde{\eta}=\frac{\mathrm{VIX}^{2}}{\sigma_{\mathbb{P}}^{2}}$,
where $\sigma_{\mathbb{P}}^{2}$ is an empirical measure of the expected
variance under $\mathbb{P}$, deduced from monthly realized variances.
So, $\sigma_{\mathbb{P}}^{2}$ does not rely on the Heston-Nandi GARCH
model. The monthly realized variance for the S\&P 500 index is computed
as the sum of the squared daily close-to-close log returns within
the same month. This quantity is multiplied by $12\times100^{2}$
to bring it to the same scale as the $\mathrm{VIX}^{2}$ (squared
annualized percentage). Finally, $\sigma_{\mathbb{P}}^{2}$ is defined
to be the one-month ahead prediction of the monthly realized variance
using a AR(1) model. 
\begin{figure}
\centering{}\includegraphics[width=1\textwidth]{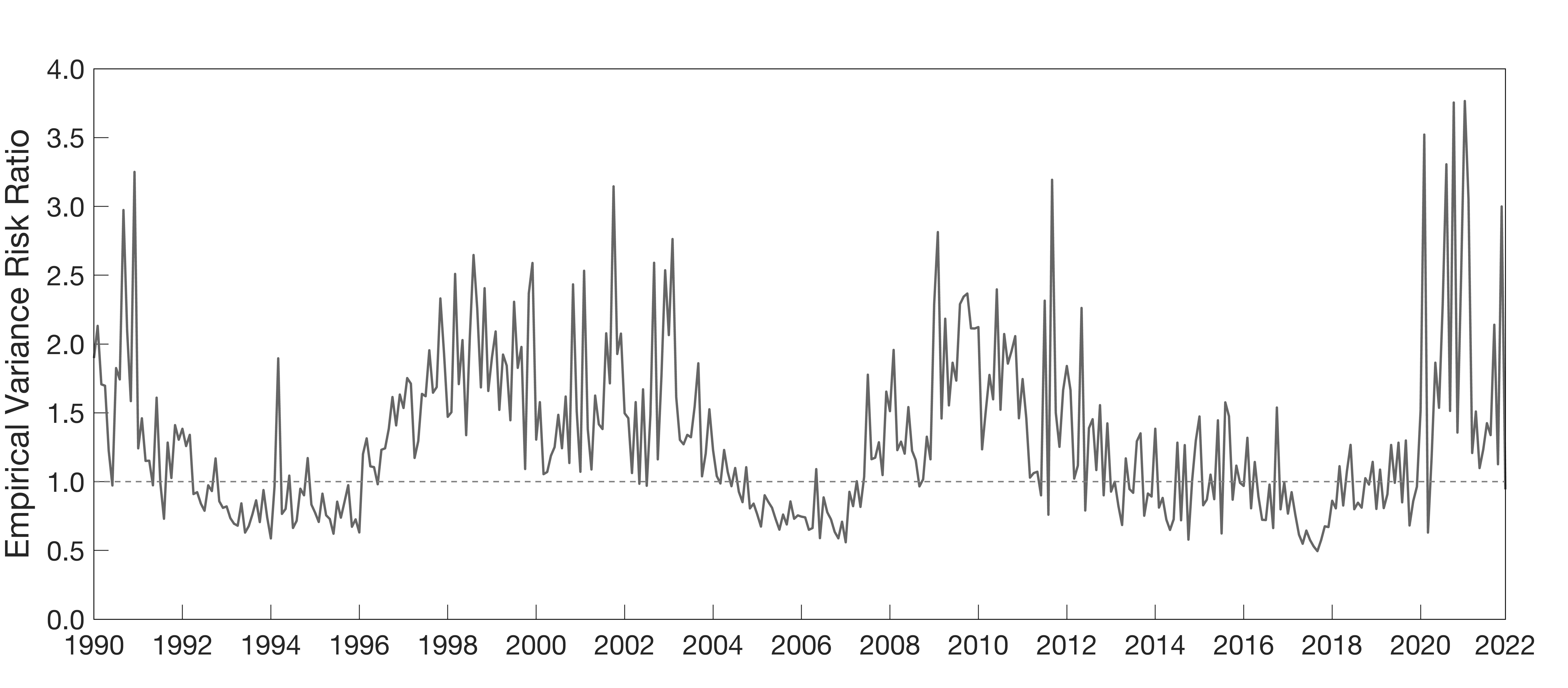}\caption{An empirical monthly variance risk ratio, defined as the ratio of
${\rm VIX}_{t}^{2}$ and the predicted value of the monthly realized
variance using an AR(1) model. \label{fig:EmpiricalVRR}\protect \\
\protect \\
Alt text: Time series of the empirical monthly variance risk ratio,
calculated as the ratio of VIX squared to the predicted monthly realized
variance, where the prediction is based on a simple first-order autoregressive
model.}
\end{figure}

This time-series of this empirical VRR is presented in Figure \ref{fig:EmpiricalVRR},
and it is very similar to the model-based time-series presented in
Figure \ref{fig:eta}. 

Figure \ref{fig:Miss} contains scatterplots of one-month-ahead volatility
prediction errors under $\mathbb{P}$ and the level of $\log\eta$.
The former is defined as the difference between the ex-post model-implied
variance and the ex-ante model-based one-month ahead volatility prediction.
The scatterplots indicate some positive correlation, but we do not
find that a large proportion of overpredictions of volatility occurs
when $\log\eta_{t}$ is relatively small, nor that a large proportion
of underpredictions of volatility occurs when $\log\eta_{t}$ is large.
\begin{figure}
\centering{}\includegraphics[width=1\textwidth]{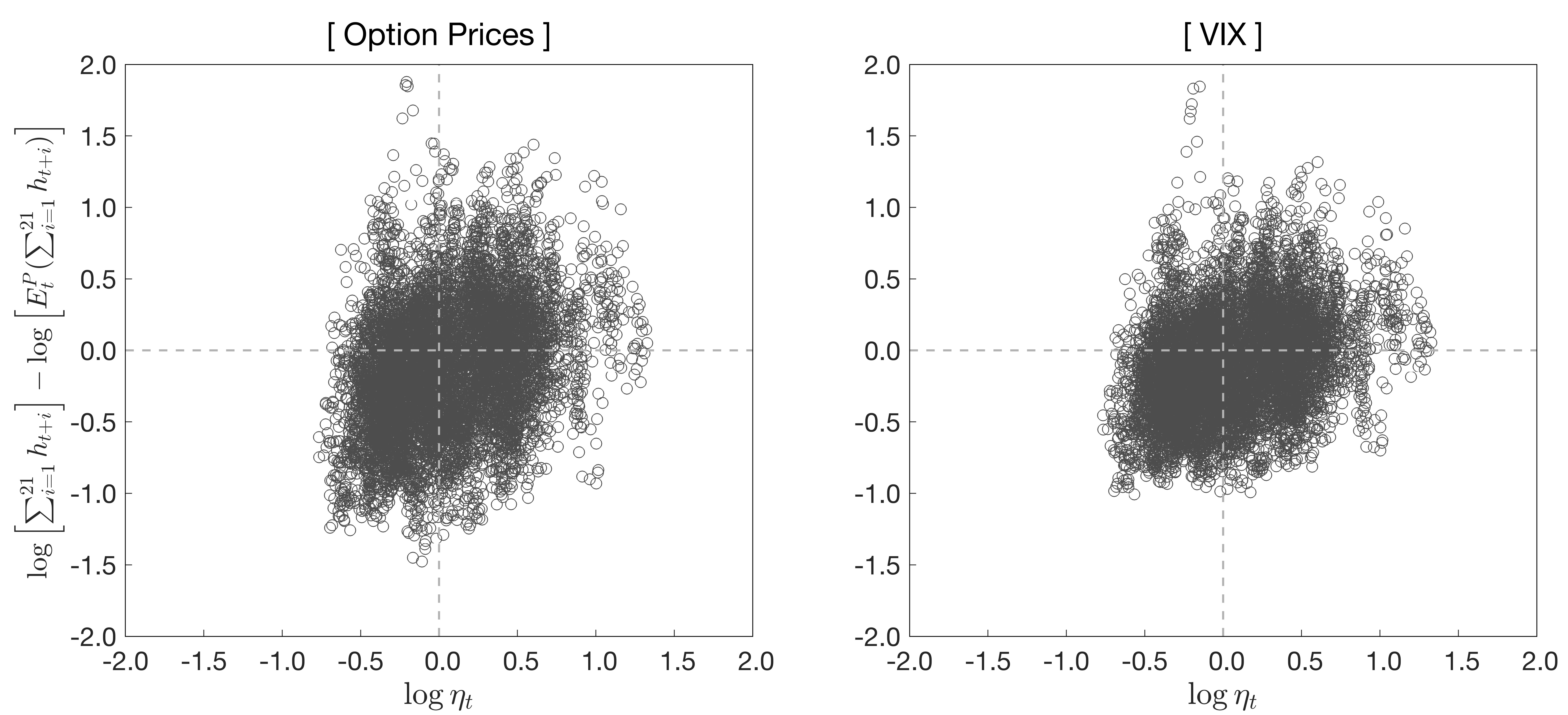}\caption{Volatility prediction errors for one-month ahead predictions plotted
against the variance risk ratio. Based on the DHNG model estimated
with option prices (left panel) and VIX (right panel). \label{fig:Miss}\protect \\
\protect \\
Alt text: Volatility prediction errors for one-month-ahead forecasts
plotted against the variance risk ratio. Predictions are based on
the DHNG model estimated using either option prices or VIX data. }
\end{figure}

\subsection{Empirical Results for Auxiliary Structure (Predetermined VRR)\label{subsec:ResultsAuxi}}

While the Auxiliary structure does not constitute a coherent option
pricing model, we can estimate the parameters. We should not expect
all parameters to be consistent, because the auxiliary structure employs
and option pricing that poorly approximates the true formula, as shown
in Table \ref{tab:ApproError}. The parameter estimates of the Auxiliary
structure are shown in Table \ref{tab:JointEstimation_Appro}, where
we also include those of the DHNG model, for comparison. Most parameter
estimates are very similar, with the important exception with the
intercept in the score-driven model for $\log\eta_{t}$. This can
be seen from the estimate of $\zeta=\mathbb{E}\left(\log\eta_{t}\right)$,
which is substantially larger for the Auxiliary structure than the
DHNG model. The reason is simple that the Auxiliary structure need
an upwards biased value of $\eta_{t}$ to compensate for the shortcomings
of its option pricing formula, that neglects the random variations
in $\eta_{t}$. So, by inflating the model-implied value of $\eta_{t}$,
the Auxiliary structure is able to reduce option pricing errors, which
is part of the objective in the estimation problem.

The option pricing performance for the Auxiliary structure is shown
in Table \ref{tab:DerivativePricing-Appro}, where we, for the sake
of comparison, include the results for the DHNG model. We evaluate
the derivative pricing in terms of the RMSE for both the logarithmically
transformed implied volatility (the first two columns) and the level
of implied volatility (the last two columns). Interestingly, the Auxiliary
structure only results in slightly larger pricing errors than the
DHNG model. This can be ascribed to the adaptive nature of the score-driven
model, which shifts $\log\eta_{t}$ to a higher level to compensate
for errors embedded in the Auxiliary pricing formulae. Thus, while
the first-order approximation is insufficient, the estimated score-drive
model ``rescues'' the Auxiliary structure, by inflating the variance
risk ratio above its true value.

\begin{table}
\caption{Estimation Results for Auxiliary and DHNG}

\begin{centering}
\vspace{0.2cm}
\begin{footnotesize}
\begin{tabularx}{\textwidth}{p{4cm}YYYY}
\toprule  
\midrule
   Model & Auxiliary  & DHNG \\
        & (1st-order) & (2nd-order) \\   
          & [Opt] & [Opt] \\
    \midrule
          &       &  \\
    $\lambda$ & 3.080 & 3.088 \\
          & \textit{(0.588)} & \textit{(0.644)} \\
          &       &  \\
    $\beta$ & 0.587 & 0.570 \\
          & \textit{(0.023)} & \textit{(0.015)} \\
          &       &  \\
    $\alpha(\times10^{-6})$ & 5.783 & 5.833 \\
          & \textit{(0.389)} & \textit{(0.256)} \\
          &       &  \\
    $\gamma$ & 248.08 & 249.59 \\
          & \textit{(12.52)} & \textit{(10.28)} \\
          &       &  \\
    $\zeta$ & 0.227 & 0.102 \\
          & \textit{(0.012)} & \textit{(0.011)} \\
          &       &  \\
    $\varphi$ & 0.994 & 0.994 \\
          & \textit{(0.003)} & \textit{(0.001)} \\
          &       &  \\
    $\sigma$ & 0.046 & 0.042 \\
          & \textit{(0.008)} & \textit{(0.010)} \\
          &       &  \\
    $\rho$ & 0.114 & 0.087 \\
          & \textit{(0.015)} & \textit{(0.010)} \\
          &       &  \\
    $\sigma_e$ & 0.090 & 0.091 \\
          & \textit{(0.006)} & \textit{(0.006)} \\
          &       &  \\
    $\mathbb{E}{\eta}$ & 1.362 & 1.190 \\
          &  &  \\
          &  &  \\
    $\widehat{\mathrm{var}}(s_{t})$ & 1.010 & 1.007 \\
          &  &  \\
    $\pi^\mathbb{P}$ & 0.945 & 0.934 \\
    $\pi^\mathbb{Q}$ & 0.956 & 0.944 \\
   
          &  &  \\
    LogL  &  &  \\
    $\ell(\rm{R})$ & 26,413 & {26,419} \\
    $\ell(\rm{VIX})$ & \textit{7,828}  & \textit{7,932} \\
    $\ell(\rm{Opt})$ & 36,534 & {36,682} \\
    $\ell(\rm{R,VIX})$ & \textit{34,241} & \textit{34,351} \\
    $\ell(\rm{R,Opt})$ & 62,947 & {63,101} \\
\\[0.0cm]
\\[-0.5cm]
\midrule
\bottomrule
\end{tabularx}
\end{footnotesize}
\par\end{centering}
{\small Note: Estimation results based on option prices for the Auxiliary
structure and the DHNG model using the full sample period, January
1990 to December 2021. Estimates are reported with robust standard
errors (in parentheses), and $\pi^{\mathbb{P}}$ and $\pi^{\mathbb{Q}}$
refer to the persistence of volatility under $\mathbb{P}$ and $\mathbb{Q}$,
respectively. The last five rows report the components of the maximized
log-likelihood function.\label{tab:JointEstimation_Appro}}{\small\par}
\end{table}

\begin{table}
\caption{Pricing Performance for Auxiliary and DHNG}

\begin{centering}
\vspace{0.2cm}
\begin{footnotesize}%
\begin{tabularx}{\textwidth}{p{2.5cm}YYYp{0cm}YYYYYYYYYYYYYY}
\toprule  
\midrule
          &       & \multicolumn{2}{c}{Errors in Log-Difference} &       & \multicolumn{2}{c}{Errors in Level-Difference} \\
\cmidrule{3-4}\cmidrule{6-7}    Model &       & Auxiliary  & DHNG  &       & Auxiliary  & DHNG \\
       &       & (1st-order) & (2nd-order) &       & (1st-order) & (2nd-order) \\          
          &       & [Opt] & [Opt] &       & [Opt] & [Opt] \\
    \midrule
          &       &       &       &       &       &  \\
    \multicolumn{7}{l}{\it A: RMSE for VIX Pricing}\\[6pt]
    Full Sample &  & 9.164 & 9.049 &       & 1.804 & 1.795  \\
          &       &       &       &       &       &  \\
    \multicolumn{7}{l}{\it B: RMSE for Option Pricing} \\[6pt]
    Full Sample  & & 9.241 & 9.191 &       & 1.767 & 1.746 \\
          &       &       &       &       &       &  \\
    \multicolumn{7}{l}{\it Partitioned by moneyness } \\[3pt]
   \ Delta<0.3   & &  12.63 & 12.77 &       & 1.961 & 1.947 \\
   \ 0.3$\leq$Delta<0.4   &    & 11.33 & 11.28 &       & 1.880 & 1.859 \\
   \ 0.4$\leq$Delta<0.5  & & 9.496 & 9.375 &       & 1.670 & 1.639 \\
   \ 0.5$\leq$Delta<0.6   &    & 7.385 & 7.271 &       & 1.467 & 1.444 \\
   \ 0.6$\leq$Delta<0.7   &  & 6.925 & 6.859 &       & 1.468 & 1.455 \\
   \ 0.7$\leq$Delta   & & 8.746 & 8.634 &       & 1.981 & 1.949 \\
          &       &       &       &       &       &  \\
    \multicolumn{7}{l}{\it Partitioned by maturity} \\[3pt]
    \ DTM<30 & & 12.30 & 12.19 &       & 2.284 & 2.252 \\
    \ 30$\leq$DTM<60 & & 7.969 & 7.848 &       & 1.500 & 1.475 \\
    \ 60$\leq$DTM<90 & & 6.892 & 6.914 &       & 1.312 & 1.311 \\
   \ 90$\leq$DTM<120 & & 7.898 & 7.963 &       & 1.604 & 1.598 \\
   \ 120$\leq$DTM<150 & & 8.556 & 8.622 &       & 1.792 & 1.781 \\
    \ 150$\leq$DTM & & 9.265 & 9.226 &       & 1.862 & 1.819 \\
          &       &       &       &       &       &  \\
    \multicolumn{7}{l}{\it Partitioned by the level of VIX} \\[3pt]
    \ VIX<15     & &  10.04 & 10.02 &       & 1.160 & 1.155 \\
    \ 15$\leq$VIX<20    &  & 9.062 & 9.021 &       & 1.450 & 1.438 \\
    \ 20$\leq$VIX<25   &   & 8.288 & 8.224 &       & 1.739 & 1.718 \\
    \ 25$\leq$VIX<30   &   &  8.613 & 8.518 &       & 2.155 & 2.126 \\
    \ 30$\leq$VIX<35   &  & 8.593 & 8.542 &       & 2.464 & 2.444 \\
    \ 35$\leq$VIX    &  & 9.661 & 9.508 &       & 4.172 & 4.089 \\       
\\[0.0cm]
\\[-0.5cm]
\midrule
\bottomrule
\end{tabularx}
\end{footnotesize}
\par\end{centering}
{\small Note: This table reports the in-sample VIX and option pricing
performance for the Auxiliary structure and DHNG. The pricing performance
is evaluated in terms of the RMSEs of logarithmically transformed
VIX/implied volatility (the first two columns) and the levels of VIX/implied
volatilities (last two columns). For option pricing the performance
is also show for different ranges of moneyness, DTM, and the level
of VIX level.\label{tab:DerivativePricing-Appro}}{\small\par}
\end{table}

\subsection{Supplementary Material for Table \ref{tab:EconomicFundamentals}\label{subsec:FullEconomic}}

\begin{sidewaystable}
\begin{centering}
\vspace{0.2cm}
\begin{scriptsize}
\begin{tabularx}{\textwidth}{p{4.5cm}YYYYYYYYYYYYYYYYYYYYYYY}
\toprule
\midrule
           \multicolumn{9}{c}{Log Variance Risk Ratio (Monthly Average)} \\
    \midrule
 &       &   &       &       &       &       &       &       &  \\
    Sentiment &       & -0.095*** &       &       &       &       &       &       & -0.093*** \\
          &       & \textit{(0.026)} &       &       &       &       &       &       & \textit{(0.020)} \\
    Dispersion-UNEMP &       &       & 0.139*** &       &       &       &       &       & 0.008 \\
          &       &       & \textit{(0.052)} &       &       &       &       &       & \textit{(0.044)} \\
    Dispersion-GDP &       &       &       & 0.176*** &       &       &       &       & 0.050 \\
          &       &       &       & \textit{(0.040)} &       &       &       &       & \textit{(0.049)} \\
    Economic Policy Uncertainty &       &       &       &       & 0.380*** &       &       &       & 0.208*** \\
          &       &       &       &       & \textit{(0.060)} &       &       &       & \textit{(0.045)} \\
    Survey-based Uncertainty &       &       &       &       &       & 0.148*** &       &       & 0.078*** \\
          &       &       &       &       &       & \textit{(0.026)} &       &       & \textit{(0.025)} \\
    Economic Uncertainty Index &       &       &       &       &       &       & 0.071*** &       & 0.013 \\
          &       &       &       &       &       &       & \textit{(0.009)} &       & \textit{(0.009)} \\
    Variance Risk Premium &       &       &       &       &       &       &       & 0.120*** & 0.096*** \\
          &       &       &       &       &       &       &       & \textit{(0.014)} & \textit{(0.012)} \\       
    Stock Market Volatility & 0.394*** & 0.395*** & 0.368*** & 0.356*** & 0.294*** & 0.363*** & 0.354*** & 0.395*** & 0.306*** \\
          & \textit{(0.042)} & \textit{(0.041)} & \textit{(0.043)} & \textit{(0.043)} & \textit{(0.039)} & \textit{(0.042)} & \textit{(0.041)} & \textit{(0.029)} & \textit{(0.029)} \\
    Default Yield Spread & 0.385*** & 0.358*** & 0.372*** & 0.365*** & 0.393*** & 0.282*** & 0.246*** & 0.285*** & 0.196*** \\
          & \textit{(0.090)} & \textit{(0.091)} & \textit{(0.088)} & \textit{(0.087)} & \textit{(0.078)} & \textit{(0.087)} & \textit{(0.081)} & \textit{(0.067)} & \textit{(0.058)} \\
    Default Return Spread & 0.026*** & 0.027*** & 0.025*** & 0.025*** & 0.019*** & 0.024*** & 0.021*** & 0.026*** & 0.020*** \\
          & \textit{(0.008)} & \textit{(0.008)} & \textit{(0.007)} & \textit{(0.007)} & \textit{(0.007)} & \textit{(0.007)} & \textit{(0.007)} & \textit{(0.006)} & \textit{(0.005)} \\
    Long Term Yield & 0.028** & 0.030** & 0.040*** & 0.043*** & 0.077*** & 0.016 & 0.045*** & 0.035*** & 0.063*** \\
          & \textit{(0.013)} & \textit{(0.013)} & \textit{(0.013)} & \textit{(0.012)} & \textit{(0.012)} & \textit{(0.011)} & \textit{(0.010)} & \textit{(0.010)} & \textit{(0.010)} \\
    Dividend Price Ratio & -0.457*** & -0.435*** & -0.525*** & -0.529*** & -0.638*** & -0.463*** & -0.403*** & -0.354*** & -0.471*** \\
          & \textit{(0.132)} & \textit{(0.127)} & \textit{(0.130)} & \textit{(0.129)} & \textit{(0.117)} & \textit{(0.133)} & \textit{(0.122)} & \textit{(0.091)} & \textit{(0.085)} \\
    Term Spread & -0.031* & -0.040** & -0.020 & -0.018 & -0.035** & -0.031* & -0.030* & -0.011 & -0.022* \\
          & \textit{(0.018)} & \textit{(0.017)} & \textit{(0.018)} & \textit{(0.018)} & \textit{(0.015)} & \textit{(0.018)} & \textit{(0.016)} & \textit{(0.012)} & \textit{(0.012)} \\
    Earnings Price Ratio & -0.080 & -0.079 & -0.089 & -0.056 & -0.096 & -0.094 & 0.086 & -0.015 & -0.007 \\
          & \textit{(0.075)} & \textit{(0.075)} & \textit{(0.074)} & \textit{(0.070)} & \textit{(0.064)} & \textit{(0.064)} & \textit{(0.061)} & \textit{(0.045)} & \textit{(0.046)} \\
    Book to Market Ratio & 0.484 & 0.149 & 0.574 & 0.485 & 0.402 & 0.625 & 0.023 & 0.230 & -0.095 \\
          & \textit{(0.431)} & \textit{(0.435)} & \textit{(0.422)} & \textit{(0.418)} & \textit{(0.392)} & \textit{(0.435)} & \textit{(0.397)} & \textit{(0.284)} & \textit{(0.285)} \\
    Net Equity Expansion & 2.874** & 3.765*** & 1.957 & 1.777 & 1.137 & 1.218 & 0.648 & 0.239 & -0.940 \\
          & \textit{(1.368)} & \textit{(1.394)} & \textit{(1.390)} & \textit{(1.304)} & \textit{(1.269)} & \textit{(1.215)} & \textit{(1.225)} & \textit{(1.019)} & \textit{(0.926)} \\
    Long Term Rate of Return & 0.002 & 0.003 & 0.004 & 0.004 & 0.001 & 0.003 & 0.004 & 0.006 & 0.006* \\
          & \textit{(0.005)} & \textit{(0.005)} & \textit{(0.005)} & \textit{(0.004)} & \textit{(0.004)} & \textit{(0.004)} & \textit{(0.004)} & \textit{(0.004)} & \textit{(0.003)} \\
    Inflation & 5.472 & 7.530* & 6.593 & 4.323 & 0.471 & 5.131 & -4.108 & -3.496 & -4.550 \\
          & \textit{(4.464)} & \textit{(4.401)} & \textit{(4.330)} & \textit{(4.357)} & \textit{(3.900)} & \textit{(4.209)} & \textit{(4.571)} & \textit{(4.030)} & \textit{(3.116)} \\
    Recession & 0.137* & 0.191*** & 0.118* & 0.104 & 0.054 & 0.044 & 0.071 & 0.015 & -0.023 \\
          & \textit{(0.073)} & \textit{(0.073)} & \textit{(0.071)} & \textit{(0.070)} & \textit{(0.068)} & \textit{(0.070)} & \textit{(0.068)} & \textit{(0.059)} & \textit{(0.057)} \\
    Industrial Production & 2.977 & 2.460 & 2.150 & 1.761 & 2.694* & 2.504* & 0.299 & 0.380 & -0.879 \\
          & \textit{(1.820)} & \textit{(1.787)} & \textit{(1.575)} & \textit{(1.578)} & \textit{(1.626)} & \textit{(1.514)} & \textit{(1.514)} & \textit{(1.499)} & \textit{(1.213)} \\
    Employment & -0.563 & 0.561 & 6.544** & 7.386*** & 0.354 & 3.161 & 1.329 & 0.851 & 7.141*** \\
          & \textit{(2.328)} & \textit{(2.446)} & \textit{(2.629)} & \textit{(2.109)} & \textit{(2.027)} & \textit{(1.919)} & \textit{(1.809)} & \textit{(2.017)} & \textit{(2.025)} \\
    Durables Consumption & 0.617 & 0.734* & 0.234 & 0.204 & 0.391 & 0.376 & 0.641 & 0.023 & -0.127 \\
          & \textit{(0.417)} & \textit{(0.435)} & \textit{(0.436)} & \textit{(0.410)} & \textit{(0.333)} & \textit{(0.360)} & \textit{(0.422)} & \textit{(0.339)} & \textit{(0.259)} \\
    Nondurables Consumption & -0.287 & -0.999 & -1.277 & -1.034 & -0.220 & -1.098 & 0.259 & 2.099* & 0.352 \\
          & \textit{(1.224)} & \textit{(1.284)} & \textit{(1.125)} & \textit{(1.051)} & \textit{(0.872)} & \textit{(1.016)} & \textit{(1.020)} & \textit{(1.250)} & \textit{(0.729)} \\
    Service Consumption & -2.309 & -2.341 & -5.582*** & -6.529*** & -2.955* & -4.037*** & -4.348*** & -4.670* & -7.240*** \\
          & \textit{(1.704)} & \textit{(1.685)} & \textit{(1.773)} & \textit{(1.758)} & \textit{(1.506)} & \textit{(1.432)} & \textit{(1.507)} & \textit{(2.452)} & \textit{(1.400)} \\
          &       &       &       &       &       &       &       &       &  \\
    $R^2$  & 0.709 & 0.722 & 0.722 & 0.732 & 0.756 & 0.741 & 0.761 & 0.817 & 0.866 \\
\\[0.0cm]
\\[-0.5cm]
\midrule
\bottomrule
\end{tabularx}
\end{scriptsize}
\par\end{centering}
{\small Note: Supplementary material for Table \ref{tab:EconomicFundamentals}
in ``}{\small\emph{Option Pricing with Time-Varying Volatility Risk
Aversion}}{\small ''.}{\small\par}
\end{sidewaystable}

\end{document}